\documentclass[a4paper,11pt]{article}
 \usepackage[english]{babel}
\usepackage[a4paper]{geometry}
\usepackage{amsmath}
\usepackage{amsthm}
\usepackage{amssymb}
\usepackage{bbm}
\usepackage{color}
\usepackage[dvipsnames]{xcolor}

\usepackage{pgfplots}
\usepackage{tikz}
\usetikzlibrary{positioning}
\usetikzlibrary{arrows}
\usepackage{caption}
\usepackage{subcaption}
 
\usepackage{mathrsfs}
\usepackage{enumerate}
 
 \usepackage{hyperref} 

\numberwithin{equation}{section}

\def\bR{\mathbb{R}}
\def\bE{\mathbb{E}}
\def\bP{\mathbb{P}}

\def\bN{\mathbb{N}}

\def\cD{\mathcal{D}}

\def\cV{\mathcal{V}}

\def\cF{\mathcal{F}}

\def\cN{\mathcal{N}}
\def\cE{\mathcal{E}}

\def\tbf{\textbf }
\def\wt{\widetilde}
\def\wh{\widehat}

\def\be{\begin{equation}}
\def\ee{\end{equation}}
\def\Gl{\Gamma_{\text{loop} }}
\def\Gsc{ \Gamma_{\text{sc}} } 
\def\Gs{\Gamma_{\text{s} }}
\def\Gp{\Gamma_{\text{p} }}
 \def\km{k_{N}}

\def\Wlog{w.l.o.g.\ }  
\def\rhs{r.h.s.\ }

\def\eg{e.g.\ }
\def\st{s.t.\ }

\newtheorem{theorem}{Theorem}[section]  
\newtheorem{cor}[theorem]{Corollary}
\newtheorem{prop}[theorem]{Proposition}
\newtheorem{lemma}[theorem]{Lemma}

\begin{document}

\title{The Two Point Function of the SK Model without External Field at High Temperature}

\author{Christian Brennecke\thanks{Institute for Applied Mathematics, University of Bonn, Endenicher Allee 60, 53115 Bonn, Germany} 
\and  Adrien Schertzer\footnotemark[1] 
\and Changji Xu\thanks{Center of Mathematical Sciences and Applications, Harvard University, Cambridge MA 02138,USA}
\and Horng-Tzer Yau\thanks{Department of Mathematics, Harvard University, One Oxford Street, Cambridge MA 02138, USA}
}
\date{October 3, 2023}

\maketitle

\begin{abstract}
We show that the two point correlation matrix $ \tbf M= (\langle \sigma_i \sigma_j\rangle)_{1\leq i,j\leq N} $ of the Sherrington-Kirkpatrick model with zero external field satisfies
		\[ \lim_{N\to\infty} \| \tbf M - ( 1+\beta^2 - \beta \textbf{G})^{-1}   \|_{\text{op}}  =0 \]
in probability, in the full high temperature regime $\beta < 1$. Here, $\tbf{G}$ denotes the GOE interaction matrix of the model. 
\end{abstract}

\section{Introduction and Main Result}\label{sec:intro}

Consider $N$ spins $ \sigma_i$, $i\in \{1,\dots,N\}$, with values in $\{-1,1\}$ whose interactions are described by the Sherrington-Kirkpatrick \cite{SK} Hamiltonian $H_N:\{-1,1\}^N\to\mathbb{R}$ 
		\begin{equation} \label{eq:HN}
		H_N(\sigma) = \beta \sum_{1\leq i < j\leq N} g_{ij} \sigma_{i}\sigma_j = \frac{\beta}{2} (\sigma, \tbf{G}\sigma ).
		\end{equation}
The symmetric interaction $ \textbf{G} = (g_{ij})_{1\leq i,j\leq N}$ is a GOE matrix, i.e. its entries are (up to symmetry) i.i.d. centered Gaussians of variance $N^{-1}$ for $i\neq j$ (we set $g_{ii} := 0$ for simplicity), and $\beta \geq 0$ denotes the inverse temperature. We assume the $\{g_{ij}\}$ to be realized in some probability space $(\Omega, \cF, \bP)$ and denote the expectation with respect to them by $\mathbb{E}(\cdot)$. We denote the $L^p(\Omega, \cF, \bP)$ norms by $ \| \cdot\|_{L^p(\Omega)} = (\bE\, |\cdot|^p)^{1/p}$, for $p\geq 1$.

In this paper, we are interested in analyzing the behavior of the two point correlation matrix of the model at high temperature. The correlation matrix is given by 
		\be \label{eq:mij}  \tbf M   = \big(\langle \sigma_i \sigma_j\rangle\big)_{1\leq i,j\leq N} \in \bR^{N\times N}, \ee 
where $\langle \cdot \rangle$ denotes from now on the Gibbs expectation which is defined so that
		\be\label{eq:defGibbs}  \langle f \rangle = \frac{1}{Z_N}\frac{1}{2^N} \sum_{\sigma \in \{-1,1\}^N}  f(\sigma) \,e^{H_N(\sigma)} \hspace{0.5cm} \text{with} \hspace{0.5cm} Z_N =   \frac{1}{2^N}\sum_{\sigma \in \{-1,1\}^N}  e^{H_N(\sigma)} \ee
for every observable $ f: \{-1,1\}^N \to \mathbb{R} $.

As discussed recently in \cite{ABSY} (see \cite[Theorem 1.1, Corollary 1.4 \& Remark 1.5]{ABSY} and the references therein), standard mean field heuristics suggests that, at sufficiently high temperature, the correlations $ \langle \sigma_i \sigma_j\rangle, i\neq j$, satisfy the self-consistent equations\footnote{Notice that $  \langle \sigma_i \rangle= 0$ for all $i\in\{1,\dots,N\}$ by the global spin-flip symmetry of the model. }
		\be\label{eq:sceq} 
		\begin{split}
		\langle \sigma_i \sigma_j\rangle &\approx \beta g_{ij} + \sum_{k:k\neq j} \beta g_{ik}\langle \sigma_k \sigma_j\rangle^{(i)}\\
		&\approx  \sum_{k} \beta g_{ik}\langle \sigma_k \sigma_j\rangle - \beta^2 \langle \sigma_i \sigma_j\rangle, 
		\end{split}\ee
where $ \langle \cdot \rangle^{(i)}$ denotes the Gibbs measure after particle $i$ has been removed from the system. The right hand side in the second line of \eqref{eq:sceq} is analogous to the well-known Thouless-Anderson-Palmer equations \cite{TAP} for the magnetizations at high temperature and non-zero external field. It suggests, in a suitable sense to be made precise, that
		\be \label{eq:resid} \tbf M \approx \frac{1}{ 1+\beta^2 - \beta \textbf{G} }. \ee
Since the spectrum of a GOE matrix $\tbf G$ is contained with high probability in $ [-2;2]$, the simple mean field heuristics naturally suggests the (well-known) phase transition at $\beta = 1$, without referring to replica arguments or the powerful Parisi theory \cite{Par1,Par2}. This is our main motivation and giving a rigorous proof of \eqref{eq:resid} for all $\beta <1$ is our main result. In the sequel, we denote by $\|\cdot \|$ the standard Euclidean norm in $\bR^N$ while $ \| \tbf{A}\|_{\text{op}} $  and $\|\tbf{A}\|_{\text{F}}$, for $\tbf{A}= (a_{ij})_{1\leq i,j\leq N}\in\bR^{N\times N}$, refer to the operator and Frobenius norms, respectively. They are defined by
		\[ \|\tbf{A}\|_{\text{op}} := \sup_{\substack{ v\in \bR^N:\|v\| = 1}} \| \tbf A  v  \| \hspace{0.5cm} \text{ and }\hspace{0.5cm} \|\tbf{A}\|_{\text{F}} := \bigg( \sum_{1\leq i,j\leq N} a_{ij}^2\bigg)^{1/2}.  \]
\begin{theorem}\label{thm:main}
Assume that $\beta < 1$ and denote by $\tbf{P} = (p_{ij})_{1\leq i,j\leq N}\in \bR^{N\times N}$ the matrix whose diagonal entries are equal to one and whose off-diagonal entries are given by
		\[ p_{ij} = \sum_{\substack{ \gamma:  \text{ self-avoiding} \\\text{ path from } i \text{ to } j }} \;\prod_{ e \in \gamma  } \beta g_{e} \]
for $i\neq j$, where $ g_e := g_{i_1i_2}$ for edges $e = \{i_1,i_2\}$. 
Then it holds true that
		\[ \lim_{N\to\infty} \| \tbf{M} - \tbf{P}\,\|_{\emph{F}} =0 \hspace{0.8cm} \text{ and }\hspace{0.8cm} \lim_{N\to \infty} \| \tbf P - ( 1+\beta^2 - \beta \textbf{G})^{-1}   \|_{\emph{op}}  =0 \]
in the sense of probability. In particular, we have that 
		\be\label{eq:resID} \lim_{N\to \infty} \| \tbf M - ( 1+\beta^2 - \beta \textbf{G})^{-1}   \|_{\emph{op}}  =0.  \ee
\end{theorem}

\pagebreak
\noindent \textbf{Remarks:} 
\begin{itemize}

\item[1)] As a corollary of Theorem \ref{thm:main}, it follows that the norm of $ \tbf{M}$ is given by 
		\[  \lim_{N\to \infty} \| \tbf M    \|_{\text{op}} =  \lim_{N\to \infty} \|   ( 1+\beta^2 - \beta \textbf{G})^{-1}   \|_{\text{op}}  = (1-\beta)^{-2}, \]
in the sense of probability. This follows, for instance, from \cite[Theorem C]{BY}.
 
\item[2)] Observe that the matrix $\tbf P$ in Theorem \ref{thm:main}, that approximates $ \tbf M$ to leading order in terms of the Frobenius norm, emerges naturally by iterating the identity in the first line of \eqref{eq:sceq} (and using that $ \langle \sigma_i \sigma_i\rangle =1$ for the diagonal entries of $\tbf M$). 
 
\item[3)] The boundedness of $\bE \|\tbf{M}\|_{\text{op}}$ has been conjectured in \cite[Conjecture 11.5.1]{Tal2}, based on the simple observation that $m_{ij}\approx \beta g_{ij}$ to leading order in $\beta$ (for $i\neq j$). While writing up our manuscript, we noticed that this has been verified quite recently in an independent work \cite{AG}, combining the TAP heuristics \eqref{eq:resid} with well established higher-order asymptotics for overlap moments \cite{Tal2,BMCRT}. The main result of \cite{AG} shows that $\bE \| \tbf M\,( 1+\beta^2 - \beta \textbf{G})- \text{id}_{\bR^N}\|_{\text{F}}=O(1)$, implying $ \bE \| \tbf{M}\|_{\text{op}} =O(1)$. In particular, the entries of $\tbf M\,( 1+\beta^2 - \beta \textbf{G})- \text{id}_{\bR^N}$ are typically of size $O(N^{-1})$. This information is, however, not enough to conclude the norm convergence \eqref{eq:resID} in Theorem \ref{thm:main}. Inspired by \cite{ALR}, our proof of \eqref{eq:resID} is instead based on a microcanonical analysis that determines the entries of $\tbf{M}$ explicitly in terms of $\tbf{G}$. This enables us not only to determine the typical size of the entries of $\tbf M\,( 1+\beta^2 - \beta \textbf{G})- \text{id}_{\bR^N}$, but also to conclude \eqref{eq:resID}. In particular, our arguments are independent of \cite{Tal2,BMCRT,AG}.    

\item[4)] The recent paper \cite{BB} establishes a logarithmic Sobolev inequality for the Gibbs measure induced by $H_N$ under the stronger high temperature condition $\beta < 1/4$. This implies in particular the boundedness of the norm of $\tbf{M}$. Whether the results of \cite{BB} or, possibly simpler, a spectral gap inequality can be proved for all $\beta <1$ remains an interesting open question. Theorem \ref{thm:main} may be viewed as an initial step in this direction as it implies a spectral gap inequality on the space of linear combinations of the magnetizations $\sigma_i$.  

\item[5)] As mentioned earlier, we assume for simplicity that $g_{ii}=0$. Note that Theorem \ref{thm:main} remains valid if $ \tbf G$ in \eqref{eq:HN} was replaced by a standard GOE matrix $\wt{ \tbf G}$ which is such that $\wt g_{ij} = g_{ij}$ for $i\neq j$ and which has i.i.d. Gaussian diagonal entries $\wt g_{ii}$ of variance $1/N$. In this case, the diagonal term in the interaction $ \frac\beta2 \sum_{i,j=1}^N g_{ij} \sigma_i\sigma_j$ is a constant independent of $\sigma \in \{-1,1\}^N$ that does not affect $\tbf M$ and  one has 
		\[\lim_{N\to \infty} \|   ( 1+\beta^2 - \beta \textbf{G})^{-1} - ( 1+\beta^2 - \beta \wt{\textbf{G}}) ^{-1}  \|_{\text{op}} =0 \]
in the sense of probability. 

\item[6)] Theorem \ref{thm:main} is trivial if $\beta =0$. We therefore assume from now on that $\beta > 0$. 

\end{itemize}

Let us briefly motivate the strategy of our proof. As shown in \cite{ALR}, which studies the fluctuations $ \Phi = \log Z_N - \frac14 \beta^2 N$ of the log partition function around its leading order contribution, one may think of $ \log Z_N$ for $\beta<1$ heuristically as 
		\be\label{eq:heuristic} \begin{split}
		\log Z_N &=   \sum_{1\leq i < j \leq N} \log\cosh(\beta g_{ij})  + \log \frac1{2^N}\sum_{\sigma\in \{-1,1\}^N} \prod_{1\leq i< j\leq N} \big (1+ \tanh(\beta g_{ij} ) \sigma_i\sigma_j  \big)\\
		& \approx   \sum_{1\leq i < j \leq N} \log\cosh(\beta g_{ij}) +  \sum_{\gamma \;\text{simple loop}}  \log (1 + w(\gamma) ), 
		 \end{split}\ee
where by simple loops, we mean simple, connected graphs $\gamma$ with vertices in $\{1, \dots, N\}$ each having degree two (see section \ref{sec:notation} for the details) and with corresponding weights
		\be \label{eq:gweight} w(\gamma) = \prod_{\{ i,j\}\in \gamma}  \tanh(\beta g_{ij}).  \ee
As a consequence of \eqref{eq:heuristic} we expect that up to negligible errors, we have for $i\neq j$ that
		\be \label{eq:heuristics2} 
		\begin{split}
		\langle \sigma_i\sigma_j\rangle = \beta^{-1} \partial_{g_{ij}} \log Z_N &\approx \tanh(\beta g_{ij}) +  \sum_{\{i,j\}\in\gamma \;\text{simple loop}}  \frac{w(\gamma)}{\tanh(\beta g_{ij})} \frac1{1+w(\gamma)}\\
		& \approx  \beta g_{ij}  +  \sum_{\{i,j\}\in\gamma \;\text{simple loop}} \; \prod_{\substack { e\in \gamma:e\neq \{i,j\}  }} \beta g_e\\
		& = \sum_{\substack{ \gamma:  \text{ self-avoiding} \\\text{ path from } i \text{ to } j }} \;\prod_{ e \in \gamma  } \beta g_{e},
		\end{split}
		\ee
where we used that removing the edge $\{i,j\}$ from a cycle yields a unique self-avoiding path from vertex $i$ to vertex $j$ (note indeed that there is a one-to-one correspondence between cycles containing $ \{i,j\}$ and self-avoiding paths from $i$ to $j$ of length greater or equal to two). In other words, we expect that the correlation between spins $\sigma_i$ and $\sigma_j$ can be written as the sum over weights of self-avoiding paths from $i$ to $j$. 

On the other hand, by standard mean field arguments, we also expect that 
		\be \label{eq:heur1}  \big(\tbf M( 1+\beta^2 - \beta \textbf{G})-\text{id}_{\bR^N} \big)_{ij} \approx 0.    \ee
That this is consistent with \eqref{eq:heuristics2} corresponds, on a heuristic level, to the equivalence of the two approximations in \eqref{eq:sceq}, and it is indeed easily checked that e.g.\
		\be  \label{eq:heur2} \bE   \| \tbf P ( 1+\beta^2 - \beta \tbf G) - \text{id}_{\bR^N}  \|_{\text F}^2 = O(1). \ee
For completeness, we outline the quick proof of \eqref{eq:heur2} in Appendix \ref{appx}. By Markov, such a bound implies \eqref{eq:heur1}, in the sense that for $i\neq j$ we have with high probability that
		\[ \big(\tbf P( 1+\beta^2 - \beta \textbf{G})-\text{id}_{\bR^N} \big)_{ij} = O(N^{-1}).   \]

Let us remark that proving the approximate identity \eqref{eq:heur1} based on Frobenius norm bounds as in \eqref{eq:heur2} is part of the main strategy of \cite{AG}. More precisely, \cite{AG} verifies \eqref{eq:heur2} directly with $\tbf P$ replaced by $\tbf M$, based on previously established overlap asymptotics. The quantity $O(1)$ is furthermore shown to satisfy $O(1)=C_\beta + O(N^{-1/2})$ for some explicit value $C_\beta>0$ and the uniform Frobenius norm bound is used to deduce a uniform operator norm bound on $\tbf M$, in the sense that $ \bE \|\tbf M\|_{\text{op}} = O(1)$. 

While such a strategy implies the operator norm boundedness of $ \tbf M$, the information obtained from Frobenius norm bounds as in \eqref{eq:heur2} is not enough to deduce the operator norm convergence  
		\be \label{eq:heur4}  \lim_{N\to\infty} \| \tbf M -  ( 1+\beta^2 - \beta \textbf{G})^{-1}\|_{\text{op}} = 0. \ee
To approach this problem, the heuristic remarks from above suggest to analyze the matrix $\tbf P( 1+\beta^2 - \beta \textbf{G}) - \text{id}_{\bR^{N }} $ in more detail, anticipating that $\tbf M \approx \tbf P$. Following this idea, we derive an exact identity of the form
		\[  \tbf P( 1+\beta^2 - \beta \textbf{G}) - \text{id}_{\bR^{N }}    = \tbf P  \, \tbf{Q}'   + \tbf{Q}'' \]
for two errors $ \tbf{Q}', \tbf{Q}''$ which turn out to have vanishing operator norm in the limit $N\to\infty$ if $\beta<1$. Using the boundedness of the resolvent $( 1+\beta^2 - \beta \textbf{G})^{-1} $ for $\beta<1$ and rewriting 
		\[ \tbf P  \, \tbf{Q}' = \big(\tbf P( 1+\beta^2 - \beta \textbf{G}) - \text{id}_{\bR^{N }}\big)  ( 1+\beta^2 - \beta \textbf{G})^{-1}\, \tbf{Q}' +  ( 1+\beta^2 - \beta \textbf{G})^{-1} \tbf{Q}', \]
this implies the norm convergence \eqref{eq:heur4} with $\tbf M $ replaced by $\tbf P$. To conclude \eqref{eq:heur4}, our main task is therefore reduced to justifying the approximation $\tbf M \approx \tbf P$ in a sufficiently strong sense. Based on an exact graphical representation of $\tbf M$, that follows from the identity $\langle \sigma_i\sigma_j\rangle = \beta^{-1} \partial_{g_{ij}} \log Z_N$  and the graphical representation of $Z_N$ from \cite{ALR}, we do this in the Hilbert-Schmidt sense $ \| \tbf M - \tbf P\|_{\text{F}} \to 0$ as $N\to \infty$. From the technical point of view, this is our main contribution, thereby extending the high temperature analysis of \cite{ALR} and providing a direct explanation of the emergence of the resolvent of $\tbf{G}$.

The paper is structured as follows. In Section \ref{sec:notation} we set up the notation, we collect several preliminary results and we explain in detail the reduction step from $\tbf M$ to $\tbf P$. Using this information, we conclude Theorem \ref{thm:main} in Section \ref{sec:proofmain} by comparing the matrix $\tbf{P}$ with the resolvent $ (1+\beta^2- \beta \tbf{G})^{-1}$.

\section{Setup and Reduction Step}\label{sec:notation}

In this section, we set up the notation following \cite{ALR}, we collect several preliminary results and we determine the main contribution $\tbf{P}$ to the two point correlation matrix $\tbf{M}$.

We start with some notation. In the sequel, $\langle \cdot\rangle_0$ denotes the expectation with regards to the law of $N$ i.i.d. Bernoulli variables $\sigma_i$, for $i\in \{1,\dots, N\}$. Notice that $\langle \cdot \rangle_0$ is obtained from $\langle \cdot\rangle$, defined in \eqref{eq:defGibbs}, by setting $\beta =0$ in \eqref{eq:HN}. 

We denote by $\Gs=\Gs[N]$ the set of simple graphs with vertices in $[N]:=\{ 1,\dots, N \}$: a graph $\gamma = (\cV_\gamma, \cE_\gamma) \in \Gs=\Gs[N]$ is called simple for subsets $ \cV_\gamma\subset [N] $ (the vertex set of $\gamma$) and $ \cE_\gamma\subset \big\{ \{i,j\}:   i,j \in \cV_{\gamma},\, i\neq j  \big\}$ (the set of edges of $\gamma$) if $ \gamma$ contains no isolated vertices (for every $i\in \cV_\gamma$ there exists some $e\in\cE_\gamma$ with $i\in e$) and if the multiplicity $n_{ij}(\gamma)\in \{0,1\}$ is at most one, for each edge $\{i,j\}$. We set $n_{ii}(\gamma):= 0$ for all $i\in \{1,\dots, N\}$ (self-loops are not allowed). $ \Gsc = \Gsc([N])\subset\Gs $ denotes the set of simple, closed graphs with vertices in $ [N] $: a graph $\gamma \in \Gsc$ is called a simple, closed graph if it is simple and if the degree
		\[  n_i(\gamma) =  \sum_{ j=1}^N n_{ij}(\gamma)\in 2\mathbb{N}_0 \]
of each vertex $i\in \cV_\gamma$ is even. By convention, we include $\emptyset \in \Gsc$ with weight $w(\emptyset) := 1$.

From \cite{ALR}, we recall the graphical representation of the log partition function 
		\be\begin{split} \label{eq:ZNgraph}  
		\log Z_N 
		& = \sum_{1\leq u<v\leq N} \log\cosh(\beta g_{uv}) + \log \bigg\langle  \prod_{1\leq u<v\leq N} (1+\tanh(\beta g_{uv}) \sigma_u\sigma_v) \bigg\rangle_0\\
		& = \sum_{1\leq u<v\leq N} \log\cosh(\beta g_{uv}) + \log \sum_{\gamma\in \Gsc } w(\gamma)\\
		& =: \sum_{1\leq u<v\leq N} \log\cosh(\beta g_{uv}) + \log \wh Z_N,
		\end{split}\ee
where $ w(\gamma)$ is defined as in \eqref{eq:gweight} (for every $ \gamma\in\Gamma_s$). Analogously, one obtains a graphical representation of the two point correlation functions
		\be \begin{split} \label{eq:sigmaij}  
		\langle \sigma_i\sigma_j \rangle& =\beta^{-1} \partial_{g_{ij}} \log Z_N =  \tanh(\beta g_{ij}) + (1-\tanh^2(\beta g_{ij}) )\; \wh Z_N^{-1} \!\!\sum_{\substack{ \gamma\in \Gsc:\\\{i,j\}\in \gamma} } \!\frac{ w(\gamma) }{\tanh(\beta g_{ij})}
		\end{split}\ee
for $i\neq j$ (recall that $\langle \sigma_i\sigma_i \rangle=1$).

Besides $\Gsc$, the set $\Gl \subset \Gsc$ of simple loops (or, equivalently, cycles) will be of particular importance. It is defined as the set 
		\[ \Gl = \big\{\gamma \in \Gsc: n_i(\gamma) =2 \;\forall \; i\in \cV_\gamma \text{ and } \gamma \text{ is connected} \big\}.\]
Recall that a graph is connected if for every pair $i,j \in \cV_\gamma$ of vertices, $\gamma$ contains a path from $i$ to $j$. That is, there exist vertices $v_1, \dots, v_k\in \cV_\gamma$ so that 
		\[  \big( \{i, v_1 \},  \{v_{1}, v_2 \},\dots,  \{v_{k-1}, v_k\}, \{v_k, j \}\big) \in \cE_{\gamma}^k. \]
We denote by $\Gp \subset \Gs$ the set of simple, self-avoiding paths, that is, the set of connected $\gamma\in \Gs$ in which exactly two vertices $ i,j\in [N]$, $i\neq j$, have degree $n_i(\gamma) = n_j(\gamma)=1$ (corresponding to the end points of the path) and the remaining vertices have degree two. The set of self-avoiding paths with endpoints $i\neq j\in[N] $ is denoted by $ \Gp^{ij} = \Gp^{ji} $. As already mentioned in the introduction, note that every path in $ \Gp$ with at least two edges can be identified uniquely with a cycle $\gamma\in\Gl$ by connecting its end points and, conversely, removing an edge from a given cycle in $ \Gl$ defines a unique path $ \Gp$ with at least two edges and whose end points are the vertices of the edge that has been removed.

An important result that is frequently used below is Veblen's theorem \cite[Theorem 2.7]{BM} which states that every $ \gamma \in \Gsc$ is equal to an edge-disjoint union of a finite number of cycles $ \gamma_1, \dots, \gamma_k \in \Gl$, 
which we write as 
		\be \label{eq:gsccycle} \gamma = \gamma_1 \circ \gamma_2\circ \ldots \circ \gamma_k. \ee

If $e = \{i,j\}\in \cE_\gamma$, we set $g_e := g_{ij}$, and from now on, by slight abuse of notation, we frequently write $e\in \gamma$ for edges $e\in \cE_\gamma$. Similarly, we write $|\gamma| :=|\cE_\gamma|$ to denote the number of edges in $\gamma$ and by $ \gamma\cap \gamma'=\emptyset$, we abbreviate that two graphs $\gamma, \gamma'$ are edge-disjoint (i.e. $\cE_\gamma\cap \cE_{\gamma'}=\emptyset $). Finally, constants independent of $N$ are typically denoted by $C, C', c, c', etc.$ and they may vary from line to line in our estimates.

Let us now prepare the analysis of the two point functions by collecting several preliminary results. We first recall a crucial observation of \cite{ALR} that tells us that large graphs have exponentially decaying $L^2(\Omega)$ norm, if $\beta < 1$. In the sequel, we frequently use this result with mildly diverging large graph cutoff $  (\log N)^{1+\epsilon}\to\infty  $ as $N\to\infty$ for suitable $\epsilon>0$. 
\begin{lemma}(\cite[Lemma 3.3]{ALR}) \label{lm:lrg}
Let $\beta< 1$, then there exists a constant $C=C_\beta>0$, which is independent of $N$, such that 
		\[ \bE \bigg( \sum_{\gamma\in \Gamma_{\emph{sc}}: |\gamma|\geq k} w(\gamma)  \bigg)^2 \leq  C \exp\big(  -  k\log \beta^{-1}\big)    \]
for every $k\geq 0$. Similarly, for every $i,j\in[N]$ with $i\neq j$, it holds true that
		\[ \bE \bigg( \sum_{ \substack{ \gamma\in \Gamma_{\emph{sc}}: \\ \{i,j\}\in\gamma, |\gamma|\geq k }} \frac{ w(\gamma)}{\tanh(\beta g_{ij})}  \bigg)^2 \leq  C \exp\big(  -  k\log \beta^{-1}\big).    \]
\end{lemma}
\begin{proof} For the reader's convenience, we recall the quick argument together with some useful facts about the weights $w(\gamma)$. We consider two cases. 

First, if $ k \leq \max\big( \frac12 (1-\beta^2)^{-2},  (16)^2/(\log \beta^{-2})^2 \big)$, we bound
		\[\begin{split}
		\bE \bigg( \sum_{\gamma\in \Gsc: |\gamma|\geq k} w(\gamma)  \bigg)^2  = \sum_{\gamma\in \Gsc: |\gamma|\geq k} \bE w^2(\gamma)  \leq \prod_{\gamma \in \Gl} \big(1+ \bE w^2(\gamma)\big) &\leq \exp\bigg( \sum_{ j=3 }^\infty \frac{\beta^{2j}}{2j}  \bigg)\\&\leq \frac{1}{\sqrt{1-\beta^2}}.
		\end{split}\]
In the first step, we used that the weights $ w(\gamma_1)$ and $w(\gamma_2)$ of two different graphs $\gamma_1\neq \gamma_2$ are orthogonal in $L^2(\Omega)$ and in the second step we used the observation from \eqref{eq:gsccycle}. Finally, in the third step we used the estimate (see \cite[Eq. (3.12)]{ALR} for the details)
		\[ \sum_{\gamma \in \Gl: |\gamma|=j}  \bE w^2(\gamma)\leq \frac{\beta^{2j}}{2j}. \]
For $ k \leq \max\big( \frac12 (1-\beta^2)^{-2},  (16)^2/(\log \beta^{-2})^2 \big) $, we thus obtain trivially that
		\[ \bE \bigg( \sum_{\gamma\in \Gsc: |\gamma|\geq k} w(\gamma)  \bigg)^2\leq  C' \exp\big( - k\log \beta^{-1}\big),  \]
where $C' =  (1-\beta^2)^{-1/2}\max\big( e^{ (1-\beta^2)^{-2} \log \beta^{-1/2}},  e^{ (16)^2 / (4 \log \beta^{-1}})  \big). $
		
On the other hand, assume $ k >  \max\big( \frac12 (1-\beta^2)^{-2},  (16)^2/(\log \beta^{-2})^2 \big)$ \st in particular
		\[ \epsilon  = \log \beta^{-2} (1-1/\sqrt{2k} ) > 0.  \]
In this case, we bound
		\[\begin{split}
		\bE \bigg( \sum_{\gamma\in \Gsc: |\gamma|\geq k} w(\gamma)  \bigg)^2 &  \leq \sum_{l\geq 1} \sum_{\substack{  \{ \gamma_1, \ldots, \gamma_l\}: \gamma_i\in \Gl,\\ \gamma_i\neq \gamma_j, |\gamma_1| +\ldots +|\gamma_l|\geq k } } \;\;\prod_{j=1}^l \bE w^2(\gamma_j)  \\
		& \leq  \sum_{l\geq 1} \sum_{\substack{  \{ \gamma_1, \ldots, \gamma_l\}: \gamma_i\in \Gl,\\ \gamma_i\neq \gamma_j, |\gamma_1| +\ldots +|\gamma_l|\geq k } } e^{\epsilon( |\gamma_1| +\ldots +|\gamma_l| - k)}\prod_{j=1}^l  \bE w^2(\gamma_j)  \\
		& \leq e^{-\epsilon k} \sum_{l\geq 1} \sum_{\substack{  \{ \gamma_1, \ldots, \gamma_l\}: \\\gamma_i\in \Gl, \gamma_i\neq \gamma_j  } } \prod_{j=1}^l e^{\epsilon |\gamma_j|} \bE w^2(\gamma_j)  \\
		& \leq e^{-\epsilon  k }\prod_{\gamma \in \Gl} \big(1+ e^{\epsilon |\gamma| }\bE w^2(\gamma)\big) 
		\end{split}\]
and thus obtain that
		\[\begin{split}
		\bE \bigg( \sum_{\gamma\in \Gsc: |\gamma|\geq k} w(\gamma)  \bigg)^2   &\leq \beta^{2k }\exp\bigg(- k\log (e^\epsilon\beta^2) + \sum_{ j=3 }^\infty \frac{e^{\epsilon j}\beta^{2j}}{2j}  \bigg)\leq \exp\big(8\sqrt{k} - k\log \beta^{-2} \big)  .
		\end{split}\]
	
Since $ k >  \max\big( \frac12 (1-\beta^2)^{-2},  (16)^2/(\log \beta^{-2})^2 \big)$, we have that $ k\log \beta^{-1} \geq 8\sqrt{k} $ and thus 
		\[\bE \bigg( \sum_{\gamma\in \Gsc: |\gamma|\geq k} w(\gamma)  \bigg)^2   \leq \exp\big(- k\log \beta^{-1} \big). \]
Finally, setting $C  = \max(1, C')$ with $C'>0$ from the first step concludes the first bound of the lemma. For the second bound, we use the rough estimate
		\[	\sum_{\substack{ \gamma \in \Gl: \\ |\gamma|=l, \{i,j\}\in \gamma }}  \bE \,\frac{ w^2(\gamma) }{\tanh^2(\beta g_{ij})}\leq N^{-1} \beta^{2l}\leq  \beta^{2l}\]
and proceed similarly as in the first step.
\end{proof}

As indicated in the introduction, the proof of Theorem \ref{thm:main} relies crucially on the approximation \eqref{eq:heuristics2}. To make this more precise, we need some preparation. We define the map $ \Phi: \cD_{\Phi} \to \Gsc$ by
		\be\label{eq:defPhi} \begin{split}
		\cD_{\Phi} &:=\big \{(\gamma,\tau)\in \Gl\times \Gsc:\gamma\cap\tau=\emptyset \big\}, \\
		 \Phi( \gamma,\tau) &:= \gamma \circ \tau \in \Gsc, \forall \,(\gamma,\tau)\in\cD_{\Phi}. 
		\end{split}\ee
Note that $\Phi( \gamma,\tau)\in \Gsc  $ due to $\gamma\cap\tau=\emptyset$. Furthermore, for $i_1\neq i_2$, we set
		\be\label{eq:defSij}S_{i_1i_2} := \big\{ (\gamma,\tau)\in \Gl\times \Gsc: \{i_1,i_2\} \in \gamma, |\cV_{\gamma}\cap\cV_{\tau}|\leq 1 \big\} \subset \Gl\times\Gsc.\ee
Observe that $ \gamma\cap \tau=\emptyset$ for $(\gamma,\tau)\in S_{i_1,i_2}$, due to the vertex set constraint $|\cV_{\gamma}\cap\cV_{\tau}|\leq 1$. 
\begin{lemma}\label{lm:Phi}
Let $\Phi$ be defined as in \eqref{eq:defPhi} and let $S_{i_1i_2} $ be defined as in \eqref{eq:defSij}, for $i_1\neq i_2$. Then $\Phi_{i_1i_2}:=\Phi_{|S_{i_1i_2}}:S_{i_1i_2}\to \Gamma_{\emph{sc}}$ is injective. 
\end{lemma}
\begin{proof} 
We prove the more general property that if $ \gamma = \Phi (\gamma_1,  \gamma_2) =\gamma_1\circ\gamma_2$ for $(\gamma_1,\gamma_2)\in S_{i_1i_2}$ and if there exist $ \gamma_1'\in \Gl$ with $\{i_1,i_2\}\in \gamma_1'$ as well as $\gamma_2' \in \Gsc$ with $\gamma_1'\cap \gamma_2' =\emptyset$ so that $\gamma$ can also be represented as 
		\[ \gamma =\gamma_1\circ\gamma_2= \gamma_1'\circ\gamma_2',\] 
then it already follows that $\gamma_1' = \gamma_1$ and, as a consequence, that $\gamma_2' = \gamma_2$ (because $ \gamma_1\cap \gamma_2=\emptyset$ and $ \gamma_1'\cap \gamma_2'=\emptyset$). To show that $\gamma_1' = \gamma_1$, we infer from $\gamma = \Phi_{i_1i_2} (\gamma_1,  \gamma_2)$ that $\gamma_1\in \Gl$ contains at most one vertex of $\gamma$ that has degree greater than two. This happens if and only if $ \gamma_1$ and $\gamma_2$ share exactly one vertex, $ |\cV_{\gamma_1}\cap\cV_{\gamma_2}| = 1$. Elementary arguments imply that removing this unique vertex disconnects $\gamma$ into two connected components which are unique (if $ |\cV_{\gamma_1}\cap\cV_{\gamma_2}| = 0$, on the other hand, the claim is trivial). 
\end{proof}

Lemma \ref{lm:Phi} shows that $ \Phi_{ij}:S_{ij}\to \Phi(S_{ij})$ is a bijection and we write 
		\be\label{eq:defR1}\begin{split}
		&\sum_{\{i,j\}\in\gamma\in \Gsc} \frac{w(\gamma)}{\tanh(\beta g_{ij})} \\
		&=    \sum_{ \substack{  \gamma =  \gamma_1\circ\gamma_2\in  \Phi(S_{ij})  :\\ \{i,j\}\in \gamma_1\in \Gl, \gamma_2\in \Gsc,\\ \gamma_1\cap\gamma_2=\emptyset, |\cV_{\gamma_1}\cap\cV_{\gamma_2}|\leq 1 }} \frac{w(\gamma)}{\tanh(\beta g_{ij})} + \sum_{ \substack{ \{i,j\}\in \gamma \in  \Gsc\setminus \Phi(S_{ij})   }} \frac{w(\gamma)}{\tanh(\beta g_{ij})}=: \Lambda_{ij} +R_{ij}^{(1)}.
		\end{split}\ee
In the following, we analyze $\Lambda_{ij}$ and $R_{ij}^{(1)}$ separately. We find that the second term on the \rhs in \eqref{eq:defR1} is typically of order $O(N^{-3/2})$ while $\Lambda_{ij}$ is close to the \rhs in \eqref{eq:heuristics2}. To see this latter fact, recalling that $ \Phi_{ij}: S_{ij}\to \Phi(S_{ij})$ is a bijection and that $ S_{ij}\subset \Gl\times \Gsc$, we can split the summation in $\Lambda_{ij}$ as
		\[\begin{split}
		\Lambda_{ij} &= \sum_{ \substack{  \gamma' =  \gamma\circ\tau\in  \Phi(S_{ij})  :\\ \gamma\in \Gl:\{i,j\}\in\gamma , \tau\in \Gsc,\\ \gamma \cap\tau=\emptyset, |\cV_{\gamma}\cap\cV_{\tau}|\leq 1 }} \frac{w(\gamma )}{\tanh(\beta g_{ij})}  \,w(\tau)  \\
		& = \sum_{ \substack{   \gamma\in \Gl:\{i,j\}\in\gamma   }} \frac{w(\gamma)}{\tanh(\beta g_{ij})}  \bigg( \sum_{ \substack{   \tau\in  \Gsc  : \gamma\circ\tau\in \Phi(S_{ij}) }}  \!\!\! w(\tau)  \bigg).
		\end{split}\]
Now, for fixed cycle $ \gamma \in \Gl	$ with $\{i,j\}\in \gamma$, we know that $\gamma\circ\tau\in \Phi(S_{ij})$ or, equivalently, $ (\gamma, \tau)\in  S_{ij}$, if and only if  $|\cV_{\gamma}\cap\cV_{\tau}|\leq 1$ (implying $ \gamma\cap\tau = \emptyset$) so that 
		\[ \begin{split}
		\sum_{ \substack{   \tau\in  \Gsc  : \gamma\circ\tau\in \Phi(S_{ij}) }}  \!\!\! w(\tau) =  \sum_{ \substack{   \tau\in  \Gsc  : (\gamma, \tau)\in  S_{ij} }}  \!\!\! w(\tau)  &= \sum_{\substack{  \tau\in \Gsc   }}   w(\tau)- \!\!\!\sum_{ \substack{   \tau\in \Gsc :  \gamma \cap \tau  =\emptyset, \\ |\cV_{\gamma}\cap\cV_{\tau}|\geq 2}}  \!\!\! w(\tau)  -  \!\!\!\sum_{    \tau\in \Gsc :  \gamma \cap \tau  \not=\emptyset }  \!\!\! w(\tau) \\
		&=  \wh Z_N- \!\!\!\sum_{ \substack{   \tau\in \Gsc : |\cV_{\gamma}\cap\cV_{\tau}|\geq 2}}  \!\!\! w(\tau).
		\end{split}\]
Here, we used in the second step the definition \eqref{eq:ZNgraph} of $\wh Z_N$ and the fact that two graphs that share an edge also share at least two vertices. 

To connect \eqref{eq:sigmaij} with the \rhs of \eqref{eq:heuristics2}, we now split 
		\[\begin{split}
		\Lambda_{ij} & = \wh Z_N\sum_{ \substack{   \gamma\in \Gl:\{i,j\}\in\gamma   }}  \frac{w(\gamma)}{\tanh(\beta g_{ij}) }  	- R_{ij}^{(2)} -R_{ij}^{(3)}-R_{ij}^{(4)} + R_{ij}^{(5)},	
		\end{split}\]
where the errors $ R_{ij}^{(2)}, R_{ij}^{(3)}, R_{ij}^{(4)} , R_{ij}^{(5)}$ are given by
 		\be\label{eq:defR2}\begin{split}
		R_{ij}^{(2)}& := \sum_{ \substack{   \gamma\in \Gl:\{i,j\}\in\gamma   }}  \frac{w(\gamma) \mathbf1 _{\{|\cV_\gamma| < \km^2 \}}}{\tanh(\beta g_{ij}) }   \sum_{ \substack{   \tau\in \Gsc : |\cV_{\gamma}\cap\cV_{\tau}|\geq 2}}  \!\!\! w(\tau) \mathbf1 _{\{  |\cV_\tau| < \km^4 \}} \,,\\
		R_{ij}^{(3)}& := \sum_{ \substack{   \gamma\in \Gl:\{i,j\}\in\gamma   }}  \frac{w(\gamma)  \mathbf1 _{\{|\cV_\gamma| < \km^2 \}}}{\tanh(\beta g_{ij}) }   \sum_{ \substack{   \tau\in \Gsc : |\cV_{\gamma}\cap\cV_{\tau}|\geq 2}}  \!\!\! w(\tau) \mathbf1 _{\{ |\cV_\tau| \geq \km^4\}}\,, \\
		R_{ij}^{(4)}& := \sum_{ \substack{   \gamma\in \Gl:\{i,j\}\in\gamma   }}  \frac{w(\gamma) \mathbf1 _{\{|\cV_\gamma| \geq \km^2 \}}}{\tanh(\beta g_{ij}) }\sum_{ \substack{   \tau\in  \Gsc }}   w(\tau)\,, \\
		R_{ij}^{(5)}& :=\sum_{ \substack{   \gamma\in \Gl:\{i,j\}\in\gamma   }}  \frac{w(\gamma) \mathbf1 _{\{|\cV_\gamma| \geq \km^2 \}}}{\tanh(\beta g_{ij}) } \sum_{ \substack{   \tau\in  \Gsc  : \gamma\circ\tau\in \Phi(S_{ij}) }}  \!\!\! w(\tau)
		\end{split}\ee
for fixed $k_N $, specified below. Defining in addition the two errors
		\be\label{eq:defR67} 
		\begin{split} 
		R_{ij}^{(6)}&:= \big( \beta g_{ij} - \tanh(\beta g_{ij})\big) + \sum_{ \substack{   \gamma\in \Gl:\{i,j\}\in\gamma   }}  \bigg(  \prod_{\substack { e\in \gamma:e\neq \{i,j\}  }} \beta g_e -\frac{w(\gamma)}{\tanh(\beta g_{ij}) }\bigg) , \\
		R_{ij}^{(7)}&:= \tanh^2(\beta g_{ij})  \sum_{ \substack{  \gamma\in \Gl: \{i,j\}\in \gamma  }}  \frac{w(\gamma)}{\tanh(\beta g_{ij}) }
		\end{split}\ee 
and denoting the matrix corresponding to the \rhs of \eqref{eq:heuristics2} by $\tbf{P} = (p_{ij}) \in\bR^{N\times N}$ \st
		\be\label{eq:defP} 
		\begin{split} 
		p_{ij}:=& \begin{cases}  \beta g_{ij} +   \sum_{ \substack{   \gamma\in \Gl:\{i,j\}\in\gamma    }} \prod_{\substack { e\in \gamma:e\neq \{i,j\}  }} \beta g_e =   \sum_{ \substack{   \nu\in  \Gp^{ij}   }} \prod_{\substack { e\in \nu  }} \beta g_e&: i\neq j \\ 1 &: i=j,  \end{cases}
		\end{split}
		\ee				
we conclude from the identities \eqref{eq:sigmaij},  \eqref{eq:defR1}, \eqref{eq:defR2} and \eqref{eq:defR67} that
		\be \label{eq:coruseR1to7}\begin{split}
		\langle \sigma_i \sigma_j\rangle & = p_{ij}+  \big(1-\tanh^2(\beta g_{ij}) \big) \wh Z_N^{-1} \big(  R_{ij}^{(1)} - R_{ij}^{(2)}-R_{ij}^{(3)}-R_{ij}^{(4)}+R_{ij}^{(5)})  -  R_{ij}^{(6)} - R_{ij}^{(7)} 
		\end{split}\ee
for $i\neq j$ (recall that for the diagonal elements, we trivially have that $ \langle \sigma_j^2\rangle = 1 = p_{jj}$). 
	
In the rest of this section, our goal is to show that the errors $ R_{ij}^{(k)}$ are negligible contributions to the Frobenius norm of $\tbf{M}$. For most of the bounds, it is sufficient to apply Lemma \ref{lm:lrg} to graphs which have a mildly growing size of order $ O ((\log N)^{1+\epsilon}) $ for a suitable, positive $\epsilon>0$. This yields faster than polynomial (in $N$) decay for the $L^2(\Omega)$ norm of large graphs while small graphs can be controlled directly using combinatorial arguments. For definiteness, we denote by $k_N $ from now on the large graph threshold 
		\be \label{eq:defkN} k_N:=  (\log N)^{3/2}  .  \ee

Let us start the analysis by showing that the terms $ R_{ij}^{(k)}$ for $k\in \{1,3,4,5,6,7\}$ are small. To this end, we need some preparation. For $k\in \bN, l\in \bN_0$, we define
		\be \label{eq:defAkl} A_{k,l}:=\{\gamma \in \Gsc: |\cV_\gamma| = k, |\cE_\gamma| = k+l\},\quad A_{l}:=\{\gamma \in \Gsc: |\cE_\gamma|-|\cV_\gamma| = l\}, \ee
so that $A_{k,l}=\emptyset$ if $l> \frac12k(k-1)$ (the complete graph of $k$ vertices has $\frac12 k(k-1)$ edges). 

Moreover, for subsets $S\subset \Gs$, let us denote in the sequel by $ w(S)$ the sum of weights
		\[ w(S):=  \sum_{\gamma \in S} w(\gamma).\]
\begin{lemma}\label{lm:graphnumber}
For every $k\in \bN, l \in\bN_0$, the number of unlabeled simple closed graphs with $k+l$ edges and $k$ vertices is bounded by $C (k+l)^{2l}e^{C \sqrt{k+l}}$, for some universal constant $C>0$ independent of $k$ and $l$. As a consequence, for $A_{k,l}$ defined as in \eqref{eq:defAkl}, we have  
		\be\label{eq:bndAkl} |A_{k,l}|\leq C (k+l)^{2l}e^{C \sqrt{k+l}} N^k. \ee
\end{lemma}
\begin{proof}  The bound \eqref{eq:bndAkl} follows from the first claim by noticing that any graph with $k$ vertices can be assigned labels in $[N]$ in less than $N^k$ ways. So, let us focus on the claim that the number of unlabeled simple closed graphs with $k+l$ edges and $k$ vertices is bounded by $C (k+l)^{2l}e^{C \sqrt{k+l}}$ for some universal constant $C>0$. To this end, note first that the number of unlabeled simple closed graphs with $k+l$ edges and $k+l$ vertices is bounded from above by the number of solutions $(x_1, x_2, \ldots, x_{k+l})\in \bN_0^{k+l} $ of the equation
		\be\label{eq:numg} x_1 + 2 x_2 +...+(k+l) x_{k+l} = k+l,\ee
which corresponds to the number of unlabeled simple closed graphs with $k+l$ edges and $k+l$ vertices allowing in addition for self-loops. Here, $x_i \in \bN_0$ denotes the number of loops of size $i\in\bN$ contained in the graph. The upper bound via \eqref{eq:numg} simply follows from decomposing a given graph into an edge disjoint union of cycles (recall \eqref{eq:gsccycle}) and to note that the cycles in the decomposition are in fact vertex disjoint if both the number of edges and vertices equal $k+l$ (that is, the cycles characterize the unlabeled graph uniquely). Now, the number of solutions to \eqref{eq:numg} corresponds to the number of partitions of $k+l$, which is bounded by $C e^{C \sqrt{k+l}}$ (see \eg \cite[Def. 1.2 \& Eq. (5.1.2)]{An}) for some universal constant $C>0$. To conclude the claim, notice that every unlabeled simple closed graph with $k+l$ edges and $k$ vertices can be obtained from a simple closed graph of $k+l$ edges and $k+l$ vertices by merging $l$ (not necessarily vertex-disjoint) pairs of vertices which can be done in less than $ (k+l)^{2l}$ ways.
\end{proof}

\begin{lemma} \label{lm:Akl}
Let $\beta <1$, let $m, n\in \bN, l\in \bN_0$ and let $k_N$ be defined as in \eqref{eq:defkN}. Then, there exists $C=C_\beta>0$ such that for every $\epsilon>0$, we have that
	\be\label{eq:Alm}	 \begin{split}
	 \bE[w(A_{l})^2] &\leq C \max \big( N^{-l(1-\epsilon)}, \beta^{k_N^n+l} \big) \,,\\
	\bE[w(\{\gamma \in A_{l}: i_1,i_2,\ldots,i_m \in \cV_\gamma\})^2] &\leq C \max \big(  N^{-l(1-\epsilon) -m}, \beta^{k_N^n+l} \big)
	\end{split}\ee 
for all $l \leq k_N^n $, $N$ large enough and $i_1,i_2,\ldots,i_m\in [N]$. Similarly, we have that 
		\be \label{eq:almij}
		\begin{split}
		&\bE \bigg[ \frac{w(\{\gamma \in A_{l}:  \{i,j\} \in \gamma \})^2} { \tanh^2(\beta g_{ij})} \bigg]  \leq CN^{-1} \max \big(  N^{-l(1-\epsilon)}, \beta^{ k_N^n/2+l} \big)
		\end{split}\ee
for all $i, j\in [N]$, $i\neq j$ and $N$ large enough. 

Finally, if the number of vertices is bounded by $ k_N^n$, we have the improved bounds
		\be \label{eq:almij2} 
		\begin{split}
		\bE[w(\{\gamma \in A_{l}: |\cV_\gamma|\leq k_N^n \})^2]] &\leq C  N^{-l(1-\epsilon)}  \,,\\
		\bE[w(\{\gamma \in A_{l}: |\cV_{\gamma}|\leq k_N^n;  i_1,i_2,\ldots,i_m \in \cV_\gamma\})^2] &\leq C  N^{-l(1-\epsilon) -m} 
		\end{split} \ee
for every $\epsilon>0$ and $N$ large enough.	
\end{lemma}
\begin{proof}
From the $L^2(\Omega)$ orthogonality of different graphs in $\Gsc$, we get the upper bound
		\[ \bE[w( A_{l})^2] \leq \sum_{k=0}^{k_N^n} \bE[w( A_{k,l})^2] +\bE \bigg( \sum_{\gamma\in \Gsc: |\gamma| > k_N^n+l} w(\gamma)  \bigg)^2. \] 
Combining Lemma \ref{lm:graphnumber} with the fact that
		\[ \bE[w(\gamma)^2] \leq (\beta^2/N)^{k+l}\] 
for every $\gamma \in A_{k,l}$, we get
	\[ \bE[w( A_{k,l})^2] \leq C (k+l)^{2l}e^{C \sqrt{k+l}} \, N^k \,  (\beta^2/N)^{k+l}  \leq  C (k+l)^{C l}e^{C \sqrt{k+l}- (k+l)\log \beta^{-2}}  N^{-l}.\]
Using this bound and Lemma \ref{lm:lrg}, we get for every $\epsilon>0$ 
		\[\begin{split} 
		\bE[w( A_{l})^2] &\leq C (k_N^n+l)^{C l} N^{-l}\sum_{k=0}^{k_N^n}e^{C \sqrt{k+l}- (k+l)\log \beta^{-2}} + C \beta^{k_N^n+l} \\
		&\leq C (k_N^n+l)^{C l} N^{-l} + C \beta^{k_N^n+l} \\
		& \leq  N^{-l(1-\epsilon)} + C \beta^{k_N^n+l} .
		 \end{split}\]
Here, we used that $ l \leq k_N^n $ so that $  \log (k_N^n +l ) \ll \epsilon \log N  $ for every $\epsilon>0$ and large $N\in \bN$.

The second bound in \eqref{eq:Alm} follows in the same way if we use that for $k\geq m$ 
		\begin{align*}
		\frac{\bE[w(\{\gamma \in A_{k,l}: i_1,\ldots,i_m \in \cV_\gamma\})^2]}{\bE[w(A_{k,l})^2]} &= \frac{|  \{  V \subset [N] : |V| =k, \{i_1,\ldots,i_m\}\subset V  \}|}{|  \{  V \subset [N] : |V| =k  \}|} \\
		&= \frac{(N-m)(N-m-1) \ldots   (N-k+1)}{N  (N-1) \ldots (N-k+1)}. 
		\end{align*}
This follows from $\bE[w(\gamma)^2]=\bE[w(\gamma')^2] $ for all $\gamma,\gamma' \in A_{k,l}$ and the $L^2(\Omega)$ orthogonality of different graphs in $\Gsc$. The bound \eqref{eq:almij} follows with the same argument used to prove the second bound in \eqref{eq:Alm} (for $m=2$), noting that we lose a decay factor of $1/N$ in the rate (compared to the second bound in \eqref{eq:Alm}), because we divide $w(\gamma)$ by $ \tanh(g_{ij})$ and using that $N e^{- \frac12 k_N^n \log \beta^{-1}}\leq C $ for $N$ large enough. And, finally, for \eqref{eq:almij2}, we proceed as above, but we do not need a large graph cutoff. That is, we find
		\begin{align*}
		 \bE[w(\{\gamma \in A_{l}: | \cV_\gamma|\leq k_N^n \})^2] & \leq C (k_N^n+l)^{C l} N^{-l}\sum_{k=0}^{k_N^n}e^{C \sqrt{k+l}- (k+l)\log \beta^{-2}}\leq CN^{-l(1-\epsilon)}, 
		\end{align*}
uniformly in $l\geq 0$, using that $A_{k,l} =\emptyset $ whenever $ l > \frac12 k_N^{n} (k_N^n-1)$ and $ k_N^{n} \ll N $. The second bound in \eqref{eq:almij2} follows like the second bound in \eqref{eq:Alm}.
\end{proof}

\begin{lemma}\label{lm:R134567}
Let $\beta <1$ and let $R_{ij}^{(k)}  $, for $k\in \{1,3,4,5,6,7\}$ be defined as in \eqref{eq:defR1}, \eqref{eq:defR2} and \eqref{eq:defR67}, respectively. Then, we have for every $\epsilon>0$ that
		\[\max_{i,j \in [N], i\neq j}\big \|R_{ij}^{(k)} \big\|_{L^2(\Omega) }  + \max_{i,j \in [N], i\neq j}\big \| \wh Z_N^{-1} R_{ij}^{(4)} \big\|_{L^2(\Omega) } \leq N^{-3/2+\epsilon} \]
for all $ k\in \{1,5,6,7\}$ and sufficiently large $N$ . Moreover, we have that 
		\[ \lim_{N\to\infty}\bP\Big( \max_{i,j \in [N], i\neq j}   \big|R_{ij}^{(3)}\big|  > N^{-\log N}\Big) =0. \]
\end{lemma}	 
\begin{proof} 
By symmetry, it is enough to prove the bounds for fixed $i,j\in [N]$ with $i\neq j$. 

We consider first $R^{(1)}_{ij}$, defined in \eqref{eq:defR1}. Recalling the definitions in \eqref{eq:defAkl}, we use the $L^2(\Omega)$ orthogonality of different graphs and bound
		\be\label{eq:R11}\begin{split} 
		\bE \big(R_{ij}^{(1)}\big)^2  \leq \sum_{k=3}^{k_N} \sum_{l=2}^{\frac12 k_N(k_N-1)} \sum_{\gamma \in A_{k,l}: \{i,j\}\in \gamma } \bE\, \frac{w(\gamma)^2}{\tanh^2(\beta g_{ij})} + Ce^{-k_N \log \beta^{-1}}, 
		\end{split}\ee
where we applied Lemma \ref{lm:lrg}	 for the large graph contributions. Moreover, we used the fact that for $\gamma \in \Gsc\setminus \Phi(S_{ij})$ with $\{i,j\}\in\gamma$, we find a decomposition $ \gamma = \gamma'\circ \tau$ such that $\{i,j\}\in \gamma'\in \Gl, \tau \in \Gsc, \gamma'\cap\tau=\emptyset$ and $|\cV_{\gamma'}\cap \cV_{\tau}|\geq 2$. In particular, we have that $ \cE_{\gamma} = \frac12 \sum_{i=1}^N n_i(\gamma) \geq \cV_{\gamma}+ 2$, because at least two vertices have degree greater or equal to four, justifying that the sum over $l$ on the \rhs in \eqref{eq:R11} starts at $l=2$. Bounding the first term on the \rhs in \eqref{eq:R11} from above by 
		\[\begin{split}
		\sum_{k=3}^{k_N} \sum_{l=2}^{\frac12 k_N(k_N-1)} \sum_{\gamma \in A_{k,l}: \{i,j\}\in \gamma } \bE\, \frac{w(\gamma)^2}{\tanh^2(\beta g_{ij})} &\leq \sum_{l=2}^{  k_N^2} \sum_{\gamma \in A_{l}: \{i,j\}\in \gamma } \bE\, \frac{w(\gamma)^2}{\tanh^2(\beta g_{ij})}\\
		& = \sum_{l=2}^{  k_N^2} \,  \bE\, \bigg[  \frac{w (\{ \gamma\in A_l: \{i,j\}\in \gamma \})^2}{\tanh^2(\beta g_{ij})} \bigg],
		\end{split}\]
we apply Lemma \ref{lm:Akl} and conclude that
		\[ \begin{split}
		\bE \big(R_{ij}^{(1)}\big)^2 & \leq   CN^{-1}   \sum_{l=2}^{k_N^2}  \max \big(  N^{-l(1-\epsilon)}, \beta^{ k_N^2/2+l} \big)+ Ce^{-k_N \log \beta^{-1}} \\
		&\leq C N^{-3+ 2\epsilon }  + Ce^{-k_N \log \beta^{-1}} \leq CN^{-3+2\epsilon}
		\end{split}\]
for every $\epsilon >0$ and $N\in \bN$ large enough. 

Next, consider the term $ R^{(3)}_{ij}$. Applying Lemma \ref{lm:lrg} and Markov's inequality for fixed $\gamma \in \Gl$ with $\{i,j\}\in \gamma$, we find for large $N$ that
		\begin{equation*}
		\bP	\bigg( \Big| \sum_{ \substack{   \tau\in \Gsc : |\cV_{\gamma}\cap\cV_{\tau}|\geq 2}}  \!\!\! w(\tau) \mathbf1 _{\{|\cV_\tau| \geq \km^4\}}  \Big| > N^{-\km^3} \bigg) \leq N^{-\km^3}\,.
		\end{equation*}
Since the number of those $\gamma \in \Gl$ with $|\cV_\gamma| < \km^2$ is bounded by $N^{\km^2}$ and since 
		\[\begin{split} \lim_{N\to\infty} \bP \big( \max_{\gamma\in \Gl} | w(\gamma)| \geq N^{-3/2+\epsilon} \big) &=0, \\
		  \lim_{N\to\infty} \bP \big( \min_{i,j\in [N]:i\neq j} | g_{ij}| \leq N^{-3} \big) &=0 \end{split}\]
for each $\epsilon>0$ (see \cite[Eq. (3.5)]{ALR} for the first statement; the second follows from standard bounds on Gaussian integrals), a simple union bound implies for large $N$ that
		\begin{align*}
		\max_{i,j\in [N]:i\neq j}|R_{ij}^{(3)}| &= \max_{i,j\in [N]:i\neq j}\Big|\sum_{ \substack{   \gamma\in \Gl:\{i,j\}\in\gamma   }}  \frac{w(\gamma)\mathbf1 _{\{  |\cV_\gamma| < \km^2 \}}}{\tanh(\beta g_{ij}) }   \sum_{ \substack{   \tau\in \Gsc : |\cV_{\gamma}\cap\cV_{\tau}|\geq 2}}  \!\!\! w(\tau) \mathbf1 _{\{|\cV_\tau| \geq \km^4\} }\Big| \\
		&\leq N^{-\km^3}\max_{i,j\in [N]:i\neq j}\sum_{ \substack{   \gamma\in \Gl:\{i,j\}\in\gamma   }}\bigg|  \frac{w(\gamma)\mathbf1 _{\{  |\cV_\gamma| < \km^2 \}}}{\tanh(\beta g_{ij}) }    \bigg|\!\leq  N^{2 +k_N^2-\km^3} \leq \!N^{-\log N}
		\end{align*}
with probability tending to one in the limit $N\to \infty$.

The bound on $\wh Z_N^{-1} R_{ij}^{(4)} $ is obtained from Lemma \ref{lm:lrg}, observing that
		\[\wh Z_N^{-1}R_{ij}^{(4)} = \sum_{ \substack{   \gamma\in \Gl:\\ \{i,j\}\in\gamma,  |\cV_\gamma| \geq \km^2  }}  \frac{w(\gamma)  }{\tanh(\beta g_{ij}) },  \]
so that for every $\epsilon>0$ and large enough $N\in \bN$
		\[\begin{split} 
		\big \| \wh Z_N^{-1} R_{ij}^{(4)} \big\|_{L^2(\Omega) }^2 \leq \sum_{ n\geq k_N^2} \sum_{ \substack{   \gamma\in \Gl:\\\{i,j\}\in\gamma, |\gamma|=n   }}   \bE \,\frac{w^2(\gamma) }{\tanh^2(\beta g_{ij}) } &\leq CN^{-1} \sum_{n\geq k_N^2} (\beta^2)^{n}   \leq CN^{-3 +\epsilon}.
		  \end{split}\]	

The large graph bounds from Lemma \ref{lm:lrg} imply furthermore that
		\[\begin{split}
		\big \|R_{ij}^{(5)} \big\|_{L^2(\Omega) }^2 & = \bE \bigg( \sum_{ \substack{   \gamma\in \Gl:\{i,j\}\in\gamma   }}  \frac{w(\gamma) \mathbf1 _{\{|\cV_\gamma| \geq \km^2 \}}}{\tanh(\beta g_{ij}) } \sum_{ \substack{   \tau\in  \Gsc  : \gamma\circ\tau\in \Phi(S_{ij}) }}  \!\!\! w(\tau)\bigg)^2\\
		&\leq \bE \bigg( \sum_{ \substack{   \gamma'\in \Gsc: \\ \{i,j\}\in\gamma', |\gamma'| \geq \km^2  }}  \frac{w(\gamma') }{\tanh(\beta g_{ij}) } \bigg)^2\\
		&\leq Ce^{-k_N^2 \log \beta^{-1}}\leq N^{-3+\epsilon}, 
		\end{split}\]
recalling that $\gamma' = \gamma\circ\tau \in \Gsc$ if $  \gamma\in\Gl$ and $\tau\in\Gsc$ are such that $\gamma \circ\tau \in \Phi(S_{ij})$ and using the $L^2(\Omega)$ orthogonality of different graphs in $\Gsc$.  

Finally, the $L^2(\Omega)$ estimates on $ R^{(6)}_{ij}$ and $ R^{(7)}_{ij}$ are straightforward. Indeed, combining a simple second moment computation with the Taylor expansion 
		\[  \tanh(\beta g_{uv}) -\beta g_{uv} = - \int_0^1 ds \, \tanh^2( s \beta g_{uv})  \beta g_{uv} = O(g_{uv}^3), \]
we obtain on the one hand that 
		\[\begin{split}
		\big \|R_{ij}^{(6)} \big\|_{L^2(\Omega) }&\leq \big\|  \beta g_{ij} - \tanh(\beta g_{ij})\big\|_{L^2(\Omega) }  \\
		&\hspace{0.5cm}+ \sum_{k=3}^{N} \bigg\| \sum_{ \substack{   \gamma\in \Gl: \\ |\gamma|=n, \{i,j\}\in\gamma   }}  \bigg(  \prod_{\substack { e\in \gamma:e\neq \{i,j\}  }} \beta g_e -\prod_{\substack { e\in \gamma:e\neq \{i,j\}  }} \tanh( \beta g_e) \bigg)  \bigg\|_{L^2(\Omega) }\\
		&\leq C N^{-3/2} + CN^{-3/2} \sum_{n\geq 3} \beta^{n}\leq CN^{-3/2}. 
		\end{split}  \]		
On the other hand, proceeding similarly as in the proof of Lemma \ref{lm:lrg} and using that $ \bE \tanh^4(\beta g_{ij})\leq C \,\bE \,g_{ij}^4\,\leq CN^{-2} $, we readily find that
		\[\begin{split}
		 \big \|R_{ij}^{(7)} \big\|_{L^2(\Omega) }^2 & =   \bE \bigg( \tanh^2(\beta g_{ij})  \sum_{ \substack{  \gamma\in \Gl: \{i,j\}\in \gamma  }}  \frac{w(\gamma)}{\tanh(\beta g_{ij}) } \bigg)^2\\
		 & \leq C N^{-2} \, \bE \bigg(   \sum_{ \substack{  \gamma\in \Gl: \{i,j\}\in \gamma  }} \frac{w(\gamma)}{\tanh(\beta g_{ij}) } \bigg)^2 \leq CN^{-3}.\qedhere
		 \end{split}\]	 
\end{proof}

 
Let us now switch to the analysis of $R_{ij}^{(2)}$. We recall from \eqref{eq:defR2} that it is defined by 
		\[R_{ij}^{(2)}= \sum_{ \substack{   \gamma\in \Gl:\{i,j\}\in\gamma   }}  \frac{w(\gamma)}{\tanh(\beta g_{ij}) }   \mathbf1 _{\{|\cV_\gamma| < \km^2 \}}  \sum_{ \substack{   \tau\in \Gsc : |\cV_{\gamma}\cap\cV_{\tau}|\geq 2}}  \!\!\! w(\tau) \mathbf1 _{\{  |\cV_\tau| < \km^4 \}}. \]
In order to simplify the following discussion, which is based on certain graphical considerations, let us define the related error
		\be\label{eq:defR2t}\widetilde R_{ij}^{(2)} := \sum_{ \substack{   \gamma\in \Gl, \tau \in \Gsc: \\ \{i,j\}\in\gamma, |\cV_{\gamma}\cap\cV_{\tau}|\geq 2   }}   w(\gamma)   w(\tau) \mathbf1_{\{ |\cV_\gamma|  < \km^2,    |\cV_\tau| < \km^4 \} }. \ee
To relate $ \widetilde R_{ij}^{(2)}$ to $R_{ij}^{(2)}$, denote by $\bE_{ij}(\cdot)$ the expectation with regards to $g_{ij}$. Then 
	\[\begin{split}
	\bE_{ij} (R_{ij}^{(2)})^2 
	& = \bE_{ij}\tanh^2(\beta g_{ij}) \bigg(\sum_{ \substack{   \gamma\in \Gl, \tau \in \Gsc:\\ \{i,j\}\in\gamma , \{i,j\}\in \tau,\\ |\cV_{\gamma}\cap\cV_{\tau}|\geq 2 }}  \frac{w(\gamma)}{\tanh(\beta g_{ij}) } w(\tau\setminus \{i,j\} ) \mathbf1 _{\{ |\cV_\gamma|<k_N^2,  |\cV_\tau| < \km^4 \}}\bigg)^2\\
	&\hspace{0.5cm} +  \bigg(\sum_{ \substack{   \gamma\in \Gl, \tau \in \Gsc:\\ \{i,j\}\in\gamma , \{i,j\}\not \in \tau,\\ |\cV_{\gamma}\cap\cV_{\tau}|\geq 2 }} \frac{w(\gamma) }{\tanh(\beta g_{ij}) }     w(\tau  ) \mathbf1 _{\{ |\cV_\gamma|<k_N^2,  |\cV_\tau| < \km^4 \}}\bigg)^2,
	\end{split}\]
where we used that the odd moments of $\tanh(\beta g_{ij})$ vanish. Combining this with Cauchy-Schwarz, which shows that $ \big(\bE_{ij}\tanh^2(\beta g_{ij})\big)^2\leq \bE_{ij} \tanh^4(\beta g_{ij})$, we obtain
 	\[\begin{split}
	  \bE_{ij} (R_{ij}^{(2)})^2
	&  \leq CN \, \bE_{ij}\tanh^2(\beta g_{ij}) \Big( \bE_{ij} (R_{ij}^{(2)})^2\Big)\\
	& \leq CN \, \bE_{ij}\tanh^4(\beta g_{ij})  \bigg(\sum_{ \substack{   \gamma\in \Gl, \tau \in \Gsc:\\ \{i,j\}\in\gamma , \{i,j\}\in \tau,\\ |\cV_{\gamma}\cap\cV_{\tau}|\geq 2 }}  \frac{w(\gamma)w(\tau\setminus \{i,j\} )}{\tanh(\beta g_{ij}) }  \mathbf1 _{\{ |\cV_\gamma|<k_N^2,  |\cV_\tau| < \km^4 \}}\bigg)^2\\
	&\hspace{0.5cm} +  CN \, \bE_{ij}\tanh^2(\beta g_{ij})\bigg(\sum_{ \substack{   \gamma\in \Gl, \tau \in \Gsc:\\ \{i,j\}\in\gamma , \{i,j\}\not \in \tau,\\ |\cV_{\gamma}\cap\cV_{\tau}|\geq 2 }} \frac{w(\gamma) }{\tanh(\beta g_{ij}) }     w(\tau  ) \mathbf1 _{\{ |\cV_\gamma|<k_N^2,  |\cV_\tau| < \km^4 \}}\bigg)^2\\
	& \leq  CN \, \bE_{ij} \tanh^2(\beta g_{ij}) (R_{ij}^{(2)})^2
	\end{split}\]
for some $C>0$ and, as a consequence, that
		\be\label{eq:R2step0}\begin{split}
		\bE  (R_{ij}^{(2)})^2\leq C N\,\bE \big (\wt R_{ij}^{(2)}\big)^2.
		\end{split}\ee
Hence, from now on let us focus on controlling the \rhs in \eqref{eq:R2step0}. Due to the vertex constraint $|\cV_{\gamma}\cap\cV_{\tau}|\geq 2 $ and the edge constraint $\{i,j\} \in \gamma$ in the summation over $(\gamma,\tau) \in \Gl\times \Gsc$, basic combinatorics suggests $\wt  R_{ij}^{(2)}$ to be typically of size $ O(N^{-2})$ and we aim to verify this by controlling its second moment. Compared to most of the previous $L^2(\Omega)$ bounds, however, \eg those in the proof of Lemma \ref{lm:R134567}, we encounter a crucial difference: the graphs $ \gamma\in\Gl $ and $\tau\in \Gsc$ may now share some edges \st the product weights $ w(\gamma)w(\tau)$ are in general not orthogonal in $L^2(\Omega)$ for distinct pairs $ (\gamma,\tau) \in \Gl\times \Gsc$. For this reason, controlling the \rhs in \eqref{eq:R2step0} requires additional care compared to our previous arguments.

In order to estimate the second moment in a systematic way, we proceed in three main steps. Before explaining these steps more precisely, let us introduce some additional notation: motivated by the fact that we sum over products of weights of pairs $(\gamma, \tau) \in \Gl\times \Gsc$ which may have non-empty edge intersection, it is useful to consider from now on multi-graphs $\gamma\circ\tau$ (edges may have multiplicity greater than one) which are built up from cycles $\gamma\in\Gl$ and simple, closed graphs $\tau\in\Gsc$. Given such a multi-graph $\gamma\circ\tau$, we define its graph weight by
		\[w(\gamma\circ\tau):= w(\gamma) w(\tau).\] 
Note that this is consistent with the definition \eqref{eq:gweight} for simple graphs. We then define  
		\[\begin{split} 
		 \psi:\Gl\times \Gsc &\to \Gs^2, \hspace{0.5cm}  (\gamma,\tau)\mapsto \psi(\gamma,\tau):=(\psi_1,\psi_2) , 
		\end{split}\]
where $\psi_1=\psi_1(\gamma\circ\tau)  $ and $\psi_2=\psi_2(\gamma\circ\tau) $ are defined so that
		\begin{equation}\label{eq:defeta}
		\gamma \circ \tau = \psi_1 \circ \psi_2 \circ \psi_2\hspace{0.5cm}\text{ and } \hspace{0.5cm}\psi_1 \cap \psi_2 = \emptyset.
		\end{equation}
In other words, $\psi_1 (\gamma\circ \tau)$ denotes the simple graph that consists of the edges that have multiplicity one in $\gamma \circ \tau$, and $\psi_2(\gamma\circ \tau)$ denotes the simple graph consisting of the edges which have multiplicity two in $\gamma\circ\tau$. Observe here that, since $\gamma, \tau \in \Gsc$, each edge $\{i,j\}$ in $\gamma\circ\tau$ has multiplicity $n_{ij}(\gamma\circ \tau)\in \{0,1,2\}$. In particular, $\psi (\gamma,\tau)$ is well-defined. For a basic illustration of the definition of $\psi$, see Figure \ref{fig:1}. 

In Lemma \ref{lm:peta} below, we summarize basic, but important structural properties of the graphs $ \psi_1(\gamma\circ\tau)$ and $\psi_2(\gamma\circ\tau)$, for $(\gamma,\tau)\in \Gl\times \Gsc$. Among other things, we show that in fact $ \psi_1(\gamma\circ\tau)\in \Gsc$. As a consequence, note that for $  \psi_1(\gamma\circ\tau) \neq \psi_1(\gamma'\circ\tau') $
		\be \label{eq:psi1orth}
		\begin{split}
		\bE w(\gamma\circ \tau) w(\gamma'\circ \tau') &= \bE w(\psi_1(\gamma\circ\tau))  w(\psi_1(\gamma'\circ\tau'))w^2(\psi_2(\gamma\circ\tau)) w^2(\psi_2(\gamma'\circ\tau'))   \\
		& =\big( \bE  w(\psi_1 )  w(\psi_1') \big) \big( \bE w^2(\psi_2) w^2(\psi_2') \big)  = 0.
		\end{split} 
		\ee
Here, we abbreviate $ \psi_i =\psi_i(\gamma\circ\tau) $  and $ \psi_i' =\psi_i(\gamma'\circ\tau') $ for $i\in \{1,2\}$. 

\begin{figure}[t!]

\begin{subfigure}{1\textwidth}
\begin{center}

\begin{tikzpicture}[node distance={12mm}, thick, main/.style = {draw, circle}] 
\node[main, RoyalBlue] (1) {$i$}; 
\node[main, RoyalBlue] (2) [above  of=1] {$j$}; 
\node[main, RoyalBlue] (3) [above right of=2] {$k$}; 
\node[main,fill=RoyalBlue] (4) [right of=3] {$l$}; 
\node[main,,fill=RoyalBlue] (5) [below right of=4] {$m$}; 
\node[main,,fill=RoyalBlue] (6) [below of=5] {$n$}; 
\node[main, RoyalBlue] (7) [below right of=1] {$p$}; 
\node[main, RoyalBlue] (8) [right of=7] {$q$}; 

\node[main] (9) [above right of=4] {$r$}; 
\node[main] (10) [above right=1.5cm and 2cm of 5] {$s$}; 
\node[main] (11) [right=5cm of 5] {$t$}; 
\node[main] (12) [right=3cm of 10] {$u$};
\node[main] (13) [below right=1cm and 3cm of 6] {$v$};
\node[main] (14) [above right=1cm and 1.8cm of 13] {$w$}; 

\path (1) edge [RoyalBlue, bend left =10] node [left,  pos=0.5] { } (2); 
\path (2) edge [RoyalBlue, bend left =10] node [above left,  pos=0.5] { } (3); 
\path (3) edge [RoyalBlue, bend left =10] node [above,  pos=0.5] { } (4);
\path (4) edge [bend left =10] node [above right, pos=0.5] { } (5);
\path (4) edge [RoyalBlue, bend right =10] node [left, pos=0.5] { } (5);
\path (5) edge [bend left =10] node [right, pos=0.5] { } (6);
\path (5) edge [RoyalBlue, bend right =10] node [left , pos=0.5] { } (6);
\path (1) edge [RoyalBlue, bend right =10] node [below left,  pos=0.5] { } (7);
\path (7) edge [RoyalBlue, bend right =10] node [below,  pos=0.5] { } (8);
\path (8) edge [RoyalBlue, bend right =10] node [below right,  pos=0.5] { } (6);

\path (4) edge [bend left =10] node [above left,  pos=0.5] { } (9);
\path (9) edge [bend left =10] node [above right,  pos=0.5] { } (10);
\path (5) edge [bend left =10] node [below right,  pos=0.5] { } (10);
\path (5) edge [bend right =10] node [below,  pos=0.6] { } (11);
\path (10) edge [bend left =10] node [above right,  pos=0.5] { } (11);
\path (10) edge [bend left =10] node [above,  pos=0.5] { } (12);
\path (11) edge [bend right =10] node [below right,  pos=0.5] { } (12);
\path (6) edge [bend right =10] node [above right,  pos=0.5] { } (13);
\path (13) edge [bend right =10] node [below right,  pos=0.5] { } (14);
\path (11) edge [bend left =10] node [above right,  pos=0.7] { } (14);
\end{tikzpicture} 

\subcaption{A multigraph $ \gamma\circ \tau$ with $ \gamma\in\Gl$, $\tau\in\Gsc$ \st $\{i,j\} \in \gamma$ and $ |\cV_{\gamma}\cap\cV_\tau |\geq 2$. The edges of $\gamma$ are colored in blue, the edges of $\tau$ are depicted in black and the joint vertices are filled in blue. }
\end{center}
\end{subfigure}

\vspace{0.5cm}

\begin{subfigure}{1\textwidth}
\begin{center}

\begin{tikzpicture}[node distance={12mm}, thick, main/.style = {draw, circle}] 
\node[main] (1) {$i$}; 
\node[main ] (2) [above  of=1] {$j$}; 
\node[main ] (3) [above right of=2] {$k$}; 
\node[main, fill=RoyalBlue] (4) [right of=3] {$l$}; 
\node[main ,fill=RoyalBlue] (5) [below right of=4] {$m$}; 
\node[main ,fill=RoyalBlue] (6) [below of=5] {$n$}; 
\node[main ] (7) [below right of=1] {$p$}; 
\node[main] (8) [right of=7] {$q$}; 

\node[main ] (9) [above right of=4] {$r$}; 
\node[main ] (10) [above right=1.5cm and 2cm of 5] {$s$}; 
\node[main ] (11) [right=5cm of 5] {$t$}; 
\node[main ] (12) [right=3cm of 10] {$u$};
\node[main ] (13) [below right=1cm and 3cm of 6] {$v$};
\node[main ] (14) [above right=1cm and 1.8cm of 13] {$w$}; 

\path (1) edge [  bend left =10] node [left,  pos=0.5] { } (2); 
\path (2) edge [ bend left =10] node [above left,  pos=0.5] { } (3); 
\path (3) edge [bend left =10] node [above,  pos=0.5] { } (4); 
\path (4) edge [RoyalBlue, bend left =10] node [above right, pos=0.5] { } (5);
\path (4) edge [RoyalBlue,  bend right =10] node [left, pos=0.5] { } (5);
\path (5) edge [RoyalBlue,  bend left =10] node [right, pos=0.5] { } (6);
\path (5) edge [RoyalBlue, bend right =10] node [left , pos=0.5] { } (6);
\path (1) edge [bend right =10] node [below left,  pos=0.5] { } (7);
\path (7) edge [bend right =10] node [below,  pos=0.5] { } (8);
\path (8) edge [bend right =10] node [below right,  pos=0.5] { } (6);

\path (4) edge [ bend left =10] node [above left,  pos=0.5] { } (9);
\path (9) edge [ bend left =10] node [above right,  pos=0.5] { } (10);
\path (5) edge [bend left =10] node [below right,  pos=0.5] { } (10);
\path (5) edge [ bend right =10] node [below,  pos=0.6] { } (11);
\path (10) edge [ bend left =10] node [above right,  pos=0.5] { } (11);
\path (10) edge [ bend left =10] node [above,  pos=0.5] { } (12);
\path (11) edge [ bend right =10] node [below right,  pos=0.5] { } (12);
\path (6) edge [ bend right =10] node [above right,  pos=0.5] { } (13);
\path (13) edge [ bend right =10] node [below right,  pos=0.5] { } (14);
\path (11) edge [ bend left =10] node [above right,  pos=0.7] { } (14);
\end{tikzpicture} 

\subcaption{$\psi(\gamma\circ\tau) = (\psi_1, \psi_2)$: the edges of $\psi_1(\gamma\circ \tau)$ are depicted in black while the two copies of the simple path $\psi_2 (\gamma\circ\tau) = \{ l,m \}\circ\{m,n\}$ are colored in blue. By definition, $\gamma\circ\tau = \psi_1\circ\psi_2\circ\psi_2$.}
\end{center}
\end{subfigure}

\caption{}
\label{fig:1}
\end{figure}
Given this notation and preliminary remarks, our strategy to control the second moment of $ \wt R_{ij}^{(2)}$ can now be summarized by the following three key steps: 

\vspace{0.4cm}
\noindent \textbf{Step 1:} Motivated by the orthogonality \eqref{eq:psi1orth} in $\psi_1$, we make a change of variables 
		\[\begin{split} \widetilde R_{ij}^{(2)} &= \sum_{ \substack{   \gamma\in \Gl, \tau \in \Gsc: \\ \{i,j\}\in\gamma, |\cV_{\gamma}\cap\cV_{\tau}|\geq 2   }}   w(\gamma)   w(\tau) \mathbf1_{\{ |\cV_\gamma| < \km^2 ,    |\cV_\tau| < \km^4 \} }\\
		& = \sum_{\eta_1 \in \Gsc  } \sum_{\substack{  (\gamma,\tau)\in \Gl\times \Gsc: \\\{i,j\}\in\gamma, \psi_1(\gamma\circ \tau) = \eta_1   }}   w(\eta_1) w(\psi_2(\gamma\circ \tau)) \mathbf1_{\{ |\cV_\gamma| < \km^2 ,    |\cV_\tau| < \km^4 \} }
		\end{split} \]
and use \eqref{eq:psi1orth} and the positivity of Gaussian moments to estimate 
		\be\label{eq:R2step1}
		\begin{split}
		\bE \big(\widetilde R_{ij}^{(2)}\big)^2 \leq \sum_{ \substack{ \eta_1 \in \Gsc: \\ |\cV_{\eta_1}|< k_N^2+k_N^4 }} \bE w^2(\eta_1) \, \bE \bigg(  \sum_{ (\gamma, \tau)\in S_{\eta_1}^{ij} }  w^2 (\psi_2(\gamma\circ \tau ))   \bigg)^2,  
		\end{split}
		\ee
where, for $\eta\in \Gsc$ and $i\neq j\in [N]$, the sets $ S_{\eta}^{ij}$ are defined by
		 \be \begin{split} \label{eq:defSetaij} 
		S_{\eta}^{ij}&:= \big\{(\gamma,\tau) \in \Gamma_{\text{loop}}\times \Gamma_{\text{sc}}: \{i,j\} \in \gamma ,|\cV_{\gamma}\cap\cV_{\tau}|\geq 2, |\gamma| < \km^2, \psi_1( \gamma \circ \tau) = \eta\big\}. 
		\end{split}\ee
		 
\noindent \textbf{Step 2:}  In the second step, we would like to replace the sum 
		$$ \sum_{ (\gamma, \tau)\in S_{\eta}^{ij} }  w^2 (\psi_2(\gamma\circ \tau ))  $$
in \eqref{eq:R2step1} by a sum of the form 
		$$ \sum_{ \eta_2\in T_{\eta_1}^{ij} }  w^2 (\eta_2)  $$ 
for a suitable subset $ T_{\eta_1}^{ij}\subset \Gs$. Observe, however, that for $\eta_1\in \Gsc$ there may be many distinct $ (\gamma,\tau), (\gamma',\tau')\in S_{\eta_1}^{ij}$ with the property that $ \psi(\gamma,\tau) = \psi(\gamma',\tau')$; see Figure \ref{fig:2} for a basic illustration. In Lemma \ref{lm:peta}, we therefore collect basic structural properties of the graphs $ \psi_1(\gamma\circ\tau)$ and $ \psi_2(\gamma\circ\tau)$. Among other things, we estimate the size of the pre-image $ \psi^{-1}(\{(\eta_1, \eta_2)\})$ for $ (\eta_1,\eta_2)\in \Gsc\times \Gs$. Based on the key observation that $|  \psi^{-1}(\{(\eta_1, \eta_2)\})|$ is bounded by the number of cycles contained in $ \eta_1\circ \eta_2$, we use a basic counting algorithm that shows that
		\[ |  \psi^{-1}(\{(\eta_1, \eta_2)\})|\leq (1+|\cV_{\eta_1}|) e^{2 (|\cE_{\eta_1}|-|\cV_{\eta_1}|)}. \]
Setting $ T_{\eta}^{ij} := \big\{ \psi_2(\gamma\circ\tau): (\gamma,\tau)\in S_{\eta}^{ij}  \big\}$, this implies the upper bound 
		\be\label{eq:R2step2}
		\begin{split}
		\bE \big(\widetilde R_{ij}^{(2)}\big)^2 &\leq \sum_{ \substack{ \eta_1 \in \Gsc: \\ |\cV_{\eta_1}|< k_N^2+k_N^4 }} (1+|\cV_{\eta_1}|)^2 e^{4 (|\cE_{\eta_1}|-|\cV_{\eta_1}|)}  \bE \,w^2(\eta_1) \, \bE \bigg(  \sum_{ \eta_2\in T_{\eta_1}^{ij} }  w^2 (\eta_2)  \bigg)^2. 
		\end{split}
		\ee

\begin{figure}[t!]

\begin{subfigure}{0.5\textwidth}
\begin{center}

\begin{tikzpicture}[node distance={12mm}, thick, main/.style = {draw, circle}] 
\node[main, RoyalBlue,fill=RoyalBlue] (1) {}; 
\node[main, RoyalBlue,fill=RoyalBlue] (2) [above  of=1] {}; 
\node[main, RoyalBlue,fill=RoyalBlue] (3) [above right of=2] {}; 
\node[main,RoyalBlue, fill=RoyalBlue] (4) [right of=3] {}; 
\node[main,fill=RoyalBlue] (5) [below right of=4] {}; 
\node[main,fill=RoyalBlue] (6) [below of=5] {}; 
\node[main, RoyalBlue,fill=RoyalBlue] (7) [below right of=1] {}; 
\node[main, RoyalBlue,fill=RoyalBlue] (8) [right of=7] {}; 

\node[main,fill=black] (11) [right=2.2cm of 5] {}; 
\node[main,fill=black] (13) [below right=1cm and 1.2cm of 6] {};
\node[main,fill=black] (14) [above right=1cm and 1.4cm of 13] {}; 

\path (1) edge [RoyalBlue, bend left =10] node [left,  pos=0.5] { } (2); 
\path (2) edge [RoyalBlue, bend left =10] node [above left,  pos=0.5] { } (3); 
\path (3) edge [RoyalBlue, bend left =10] node [above,  pos=0.5] { } (4);
\path (4) edge [RoyalBlue,bend left =10] node [above right, pos=0.5] { } (5);
\path (5) edge [bend left =10] node [right, pos=0.6] {\small\{i,j\} } (6);
\path (5) edge [RoyalBlue, bend right =10] node [left , pos=0.5] {\small\{i,j\}  } (6);
\path (1) edge [RoyalBlue, bend right =10] node [below left,  pos=0.5] { } (7);
\path (7) edge [RoyalBlue, bend right =10] node [below,  pos=0.5] { } (8);
\path (8) edge [RoyalBlue, bend right =10] node [below right,  pos=0.5] { } (6);

\path (5) edge [bend left =10] node [below,  pos=0.6] { } (11);
\path (6) edge [bend right =10] node [above right,  pos=0.5] { } (13);
\path (13) edge [bend right =10] node [below right,  pos=0.5] { } (14);
\path (11) edge [bend left =10] node [above right,  pos=0.7] { } (14);
\end{tikzpicture} 

\subcaption{ }
\end{center}
\end{subfigure}
\begin{subfigure}{0.5\textwidth}
\begin{center}

\begin{tikzpicture}[node distance={12mm}, thick, main/.style = {draw, circle}] 
\node[main,fill=black] (1) {}; 
\node[main,fill=black] (2) [above  of=1] {}; 
\node[main,fill=black] (3) [above right of=2] {}; 
\node[main,fill=black] (4) [right of=3] {}; 
\node[main,RoyalBlue,fill=black] (5) [below right of=4] {}; 
\node[main, RoyalBlue,fill=black] (6) [below of=5] {}; 
\node[main, fill=black] (7) [below right of=1] {}; 
\node[main,fill=black] (8) [right of=7] {}; 

\node[main, RoyalBlue,fill=RoyalBlue ] (11) [right=2.2cm of 5] {}; 
\node[main, RoyalBlue,fill=RoyalBlue ] (13) [below right=1cm and 1.2cm of 6] {};
\node[main, RoyalBlue,fill=RoyalBlue] (14) [above right=0.8cm and 1.4cm of 13] {}; 

\path (1) edge [ bend left =10] node [left,  pos=0.5] { } (2); 
\path (2) edge [ bend left =10] node [above left,  pos=0.5] { } (3); 
\path (3) edge [ bend left =10] node [above,  pos=0.5] { } (4);
\path (4) edge [bend left =10] node [above right, pos=0.5] { } (5);
\path (5) edge [RoyalBlue, bend left =10] node [right, pos=0.6] {\small\{i,j\} } (6);
\path (5) edge [bend right =10] node [left , pos=0.5] {\small\{i,j\}  } (6);
\path (1) edge [  bend right =10] node [below left,  pos=0.5] { } (7);
\path (7) edge [ bend right =10] node [below,  pos=0.5] { } (8);
\path (8) edge [ bend right =10] node [below right,  pos=0.5] { } (6);

\path (5) edge [RoyalBlue,bend left =10] node [below,  pos=0.6] { } (11);
\path (6) edge [RoyalBlue,bend right =10] node [above right,  pos=0.5] { } (13);
\path (13) edge [RoyalBlue,bend right =10] node [below right,  pos=0.5] { } (14);
\path (11) edge [RoyalBlue,bend left =10] node [above right,  pos=0.7] { } (14);
\end{tikzpicture} 

\subcaption{}
\end{center}
\end{subfigure}

\caption{Example of two distinct pairs $ (\gamma,\tau), (\gamma', \tau')\in \Gl\times \Gsc \in S_{\eta}^{ij}$ for some $\eta \in \Gl$ such that $ \psi_2(\gamma\circ\tau) = \psi_2(\gamma'\circ\tau')$. Note that (b) is obtained from (a) by simply switching the roles of $\gamma$ and $\tau$.}
\label{fig:2}
\end{figure}
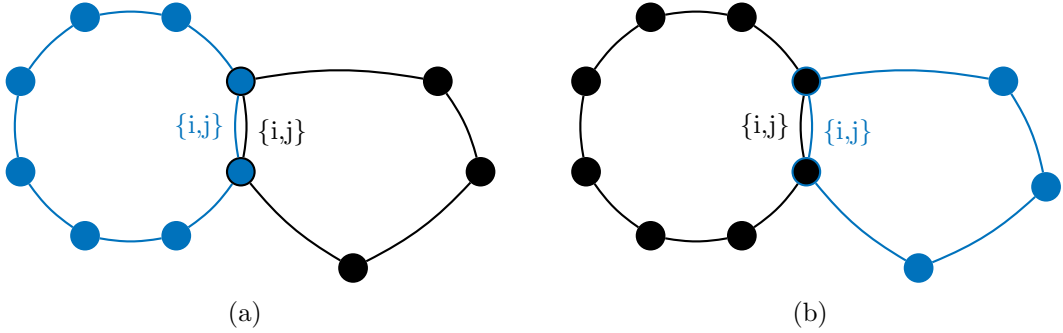

\noindent \textbf{Step 3:} In the last step, we control the \rhs in \eqref{eq:R2step2} by combining the bounds from Lemma \ref{lm:Akl} with Lemma \ref{lm:Tetambd} discussed below. In Lemma \ref{lm:Tetambd}, we estimate the size of
		\be\label{eq:R2step3a} \bE \bigg(  \sum_{ \eta_2\in T_{\eta_1}^{ij} }  w^2 (\eta_2)  \bigg)^2 = \bE  \sum_{ \eta_2, \eta_2'\in T_{\eta_1}^{ij} }  w^2 (\eta_2) w^2 (\eta'_2) \ee
under a mild constraint on $|\cV_{\eta_1}|$. Loosely speaking, the proof of Lemma \ref{lm:Tetambd} shows that the expectation in \eqref{eq:R2step3a} is, up to certain combinatorial corrections (see the bound \eqref{eq:4thmom0} for the exact statement), comparable to  
		\[ \bigg( \bE   \sum_{ \eta_2\in T_{\eta_1}^{ij} }  w^2 (\eta_2)  \bigg) \bigg( \bE   \sum_{ \eta'_2\in T_{\eta_1}^{ij} }  w^2 (\eta_2')  \bigg). \]
Combined with Lemma \ref{lm:peta}, which shows that $ \eta_2$ is either a cycle or an edge-disjoint union of self-avoiding paths with endpoints in $\eta_1$, basic combinatorial considerations then suggest that the expectation in \eqref{eq:R2step3a} is bounded uniformly in $N$, that it is of size $ O(N^{-2})$ if $ T_{\eta_1}^{ij}$ only contains graphs $\eta_2$ with at least one edge and that it is of size $  O(N^{-4})$ if $ T_{\eta_1}^{ij}$ only contains graphs $\eta_2$ that contain the edge $\{i,j\}$. All this is made precise in Lemma \ref{lm:Tetambd} below.

Now, keeping in mind that we expect that $ \bE (\widetilde R_{ij}^{(2)})^2$ is essentially of order $ O(N^{-4})$ (up to corrections of size $ O(N^{\epsilon})$, for any small $\epsilon>0$), Lemma \ref{lm:Akl} and the previous observations suggest to upper bound the \rhs in \eqref{eq:R2step2} by 
		\be\label{eq:R2step3b}
		\begin{split}
		\bE \big(\widetilde R_{ij}^{(2)}\big)^2 
		& \leq \sum_{l\leq 1}  \sum_{ \substack{ \eta_1 \in \Gsc: \{i,j\}\subset \cV_{\eta_1}, \\ |\cV_{\eta_1}|< k_N^2+k_N^4, \\ | \cE_{\eta_1}|-|\cV_{\eta_1}|=l }} (1+|\cV_{\eta_1}|)^2 e^{4l} \, \bE \,w^2(\eta_1) \, \bE \bigg(  \sum_{ \eta_2\in T_{\eta_1}^{ij} }  w^2 (\eta_2)  \bigg)^2 \\
		&\hspace{0.5cm} + \sum_{l\geq 2}  \sum_{ \substack{ \eta_1 \in \Gsc: \{i,j\}\subset \cV_{\eta_1}, \\ |\cV_{\eta_1}|< k_N^2+k_N^4, \\ | \cE_{\eta_1}|-|\cV_{\eta_1}|=l }} (1+|\cV_{\eta_1}|)^2 e^{4l} \, \bE \,w^2(\eta_1) \, \bE \bigg(  \sum_{ \eta_2\in T_{\eta_1}^{ij} }  w^2 (\eta_2)  \bigg)^2\\
		&\hspace{0.5cm} +   \sum_{ \substack{ \eta_1 \in \Gsc:\\ \{i,j \}\cap \cV_{\eta_1}^c\neq \emptyset,\\ |\cV_{\eta_1}|< k_N^2+k_N^4  }} (1+|\cV_{\eta_1}|)^2 e^{4 (|\cE_{\eta_1}|-|\cV_{\eta_1}|)} \, \bE \,w^2(\eta_1) \, \bE \bigg(  \sum_{ \eta_2\in T_{\eta_1}^{ij} }  w^2 (\eta_2)  \bigg)^2. 
		\end{split}
		\ee
Since $ 0\leq  |\cE_{\eta_1}| -  |\cV_{\eta_1}|\leq 1 $ implies $ \eta_2 \neq \emptyset$, which is explained in Lemma \ref{lm:peta}, we deduce that the \rhs of \eqref{eq:R2step3b} is of size $ O(N^{-4})$ from the following three key scenarios that correspond to the three terms on the \rhs of \eqref{eq:R2step3b}, illustrated in Figure \ref{fig:3}: 
\begin{itemize} 
\item $ \eta_1$ is such that $ \{ i,j\}\subset  \cV_{\eta_1} $ and $ |\cE_{\eta_1}| -  |\cV_{\eta_1}|\leq 1$ \st in particular $ \eta_2 \neq \emptyset$: the fact that $ \{ i,j\}\subset  \cV_{\eta_1} $ provides a decay factor $N^{-2}$ (by Lemma \ref{lm:Akl}) and the fact that $ \eta_2 \neq \emptyset$ provides another decay factor $N^{-2}$ (by Lemma \ref{lm:Tetambd}). 
\item $ \eta_1$ is such that $ \{ i,j\}\subset  \cV_{\eta_1} $ and $  |\cE_{\eta_1}| -  |\cV_{\eta_1}|\geq 2$: the fact that $ \{ i,j\}\subset  \cV_{\eta_1} $ provides a decay factor $N^{-2}$ (by Lemma \ref{lm:Akl}) and the fact that $  |\cE_{\eta_1}| -  |\cV_{\eta_1}|\geq 2$ provides an additional decay factor $ N^{-2}$ (by Lemma \ref{lm:Akl}). In other words, in this scenario the overall size $ O(N^{-4})$ is already implied by properties of the graphs $\eta_1$ alone. 
\item Either $ i \not \in \cV_{\eta_1}$ or $ j \not \in \cV_{\eta_1} $: this implies that $ \{i,j\}\in \eta_2$ so that the $N^{-4}$ decay is a consequence of the bounds on the term in \eqref{eq:R2step3a} (by Lemma \ref{lm:Tetambd}). 
\end{itemize}
Combining the key steps, our conclusion on $  R_{ij}^{(2)} $ is summarized in Corollary \ref{cor:R2} below. 
\vspace{0.4cm}

\begin{figure}[t!]

\begin{subfigure}[b]{0.32\textwidth}
\begin{center}
\begin{tikzpicture}[node distance={10mm}, thick, main/.style = {draw, circle}] 
\node[main, ForestGreen,fill=ForestGreen] (1) {}; 
\node[main, ForestGreen,fill=ForestGreen] (3) [above right of=1] {}; 
\node[main,fill=RoyalBlue] (5) [ below right=0.1cm and 0.5cm of 3] {}; 
\node[main,fill=RoyalBlue] (6) [below of=5] {}; 
\node[main, RoyalBlue,fill=RoyalBlue] (8) [below right of=1] {}; 

\node[main,fill=black] (11) [right=1cm of 5] {}; 
\node[main,fill=black] (13) [below right=0.6cm and 0.6cm of 6] {};
\node[main,fill=black] (14) [above right=0.6cm and 0.6cm of 13] {}; 

\path (1) edge [ForestGreen, bend left =15] node [left,  pos=0.5] { } (3); 
\path (3) edge [RoyalBlue, bend left =15] node [above,  pos=0.5] { } (5);
\path (5) edge [RoyalBlue,bend right =15] node [right, pos=0.6] { } (6);
\path (5) edge [bend left =15] node [right, pos=0.6] { } (6);
\path (6) edge [RoyalBlue, bend left =15] node [left , pos=0.5] {  } (8);
\path (8) edge [RoyalBlue, bend left =15] node [below right,  pos=0.5] { } (1);

\path (5) edge [bend left =10] node [below,  pos=0.6] { } (11);
\path (6) edge [bend right =10] node [above right,  pos=0.5] { } (13);
\path (13) edge [bend right =10] node [below right,  pos=0.5] { } (14);
\path (11) edge [bend left =10] node [above right,  pos=0.7] { } (14);
\end{tikzpicture} 
\subcaption{$i,j\in\cV_{\psi_1} $, $|\cE_{\psi_1}| = |\cV_{\psi_1}| $.}
\end{center}
\end{subfigure}
\begin{subfigure}[b]{0.32\textwidth}
\begin{center}
\begin{tikzpicture}[node distance={10mm}, thick, main/.style = {draw, circle}] 
\node[main, ForestGreen,fill=ForestGreen] (1) {}; 
\node[main, ForestGreen,fill=ForestGreen] (3) [above right of=1] {}; 
\node[main,fill=RoyalBlue] (5) [  right of=3] {}; 
\node[main,fill=RoyalBlue] (6) [below of=5] {}; 
\node[main, RoyalBlue,fill=RoyalBlue] (8) [below right of=1] {}; 
\node[main,fill=] (2) [above left=0.8cm and 0.2cm of 5] {}; 
\node[main,fill=] (4) [above right=1cm and 0.4cm of 5] {}; 
\node[main,fill=] (7) [below right=0.1cm and 0.8cm of 4] {}; 

\node[main,fill=black] (11) [right=1cm of 5] {}; 
\node[main,fill=black] (13) [below right=0.6cm and 0.6cm of 6] {};
\node[main,fill=black] (14) [above right=0.6cm and 0.6cm of 13] {}; 

\path (1) edge [ForestGreen, bend left =15] node [left,  pos=0.5] { } (3); 
\path (3) edge [RoyalBlue, bend left =15] node [above,  pos=0.5] { } (5);
\path (5) edge [RoyalBlue,bend right =15] node [right, pos=0.6] { } (6);
\path (5) edge [bend left =15] node [right, pos=0.6] { } (6);
\path (5) edge [bend left =15] node [right, pos=0.6] { } (2);
\path (5) edge [bend right =15] node [right, pos=0.6] { } (4);
\path (2) edge [bend left =15] node [right, pos=0.6] { } (4);
\path (11) edge [bend left =15] node [right, pos=0.6] { } (4);
\path (11) edge [bend right =15] node [right, pos=0.6] { } (7);
\path (4) edge [bend left =15] node [right, pos=0.6] { } (7);
\path (5) edge [RoyalBlue,bend right =15] node [right, pos=0.6] { } (6);
\path (5) edge [RoyalBlue,bend right =15] node [right, pos=0.6] { } (6);
\path (6) edge [RoyalBlue, bend left =15] node [left , pos=0.5] {  } (8);
\path (8) edge [RoyalBlue, bend left =15] node [below right,  pos=0.5] { } (1);

\path (5) edge [bend left =10] node [below,  pos=0.6] { } (11);
\path (6) edge [bend right =10] node [above right,  pos=0.5] { } (13);
\path (13) edge [bend right =10] node [below right,  pos=0.5] { } (14);
\path (11) edge [bend left =10] node [above right,  pos=0.7] { } (14);
\end{tikzpicture} 
\subcaption{$i,j\in \cV_{\psi_1} $, $|\cE_{\psi_1}| \geq |\cV_{\psi_1}|+ 2 $.}
\end{center}
\end{subfigure}
\begin{subfigure}[b]{0.32\textwidth}
\begin{center}
\begin{tikzpicture}[node distance={10mm}, thick, main/.style = {draw, circle}] 
\node[main, RoyalBlue,fill=RoyalBlue] (1) {}; 
\node[main, RoyalBlue,fill=RoyalBlue] (3) [above right=0.6cm and 0.6cm of 1] {}; 
\node[main, ForestGreen,fill=ForestGreen] (5) [  right of=3] {}; 
\node[main,ForestGreen ,fill=ForestGreen] (6) [below=0.8cm of 5] {}; 
\node[main, RoyalBlue,fill=RoyalBlue] (8) [below right=0.8cm and 0.15cm of 1] {}; 

\node[main,fill=RoyalBlue] (9) [below left =0.8cm and 0cm of 6] {}; 
\node[main,fill=black] (11) [right=0.8cm of 5] {}; 
\node[main,fill=black] (13) [below right=0.8cm and 0.8cm of 6] {};
\node[main,fill=black] (14) [above right=0.7cm and 0.4cm of 13] {}; 

\path (1) edge [RoyalBlue, bend left =15] node [left,  pos=0.5] { } (3); 
\path (3) edge [RoyalBlue, bend left =15] node [above,  pos=0.5] { } (5);
\path (5) edge [ForestGreen,bend right =15] node [right, pos=0.6] { } (6);
\path (5) edge [ForestGreen,bend left =15] node [right, pos=0.6] { } (6);
\path (6) edge [RoyalBlue,bend right =15] node [right, pos=0.6] { } (9);
\path (6) edge [ bend left =15] node [right, pos=0.6] { } (9);
\path (9) edge [RoyalBlue, bend left =15] node [left , pos=0.5] {  } (8);
\path (8) edge [RoyalBlue, bend left =15] node [below right,  pos=0.5] { } (1);

\path (5) edge [bend left =10] node [below,  pos=0.6] { } (11);
\path (9) edge [bend right =10] node [above right,  pos=0.5] { } (13);
\path (13) edge [bend right =10] node [below right,  pos=0.5] { } (14);
\path (11) edge [bend left =10] node [above right,  pos=0.7] { } (14);
\end{tikzpicture} 
\subcaption{$ i\not \in   \cV_{\psi_1}$ \st $ \{i,j\}\in \psi_2$. }
\end{center}
\end{subfigure}

\caption{Example graphs related to \textbf{Step 3}: $ \gamma\in\Gl$ is colored in blue, $ \tau\in\Gsc$ is depicted in black, $\{i,j\}$ is colored in green and we denote $ \psi_1 = \psi_1(\gamma\circ\tau), \psi_2 = \psi_2(\gamma\circ\tau)$.}
\label{fig:3}
\end{figure}
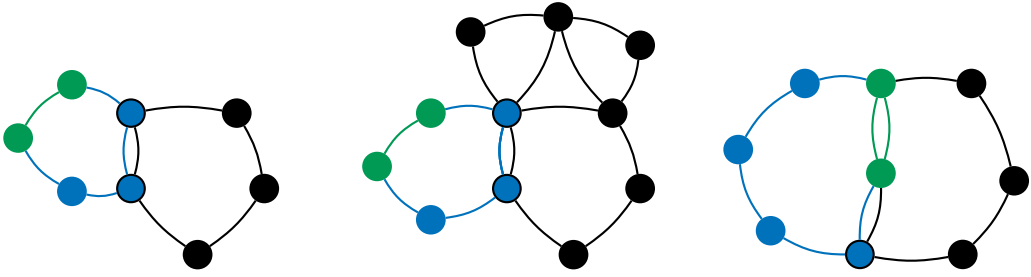
Let us now carry out the above key steps in detail. We start with the following lemma that describes basic properties of the map $\psi$ defined above before \textbf{Step 1}. 

\begin{lemma}
\label{lm:peta} Let $ \gamma\in \Gamma_{\emph{ loop}}, \tau \in \Gamma_{\emph{sc}}$ and let $\psi=(\psi_1 , \psi_2 ):\Gamma_{\emph{ loop}}\times \Gamma_{\emph{sc}}\to \Gamma_{\emph s}^2$ be defined as in \eqref{eq:defeta}. Then the following properties hold true for every $ (\gamma,\tau) \in \Gamma_{\emph{ loop}}\times \Gamma_{\emph{sc}}$:
	\begin{enumerate}
		\item[(1)] We have that $\psi_1=\psi_1 (\gamma\circ\tau) \in \Gamma_{\emph{sc}}$ and $\psi_2=\psi_2(\gamma\circ\tau)= \gamma\cap \tau $ is either equal to $ \psi_2=\gamma \in \Gamma_{\emph{loop}}$ or it is equal to an edge-disjoint union of paths in $\Gp$ whose end points lie in $\cV_{\psi_1}$. If $\psi_2=\gamma$, then $\tau = \psi_1\circ\gamma$. If $ |\cV_{\psi_{1}} \cap \cV_{\psi_{2}}| \geq 2$, then $\psi_{2} $ can always be written as an edge-disjoint union of paths in $\Gp$ with end points in $\cV_{\psi_{1}}$ (including both cases $\psi_{2}=\gamma$ or $\psi_{2}\not=\gamma$) and, on the other hand, if $ |\cV_{\psi_{1}} \cap \cV_{\psi_{2}}| \leq 1$, then we have necessarily $ \psi_{2} = \gamma $. 
		
		\item [(2)] If $|\cV_{\gamma}\cap\cV_{\tau}|\geq 2$ and $|\cE_{\psi_{1}(\gamma\circ\tau)}| - |\cV_{\psi_{1}(\gamma\circ\tau) }|\leq 1$, then $\psi_{2}(\gamma\circ\tau) \not = \emptyset$.
		
		\item[(3)] For every $\eta_{1} \in \Gamma_{\emph{sc}}$ and $\eta_{2}\in \Gamma_{\emph{s}}$ with $\eta_{1}\cap \eta_{2}=\emptyset$, we have that
				\be\label{eq:prebnd}|\psi^{-1} (\{ (\eta_{1},\eta_{2})\} )|\leq (1+|\cV_{\eta_{1}}| ) e^{2( |\cE_{\eta_{1}}|-|\cV_{\eta_{1}}|)}. \ee
		
	\end{enumerate}
\end{lemma}
\begin{proof}
To prove (1), note that for any $i\in[N]$, the degree $n_i(\psi_{1}) = n_i(\gamma)+n_i(\tau)-2n_i(\psi_{2})$ is even so that $\psi_{1} \in \Gsc$. Since $\psi_{2} = \gamma \cap \tau$, the remaining statements follow from the fact that every vertex $ i \in \cV_{\psi_{2}}$ has degree either one or two (with regards to $\psi_{2}$). In fact, since $ \psi_{2}\subset \gamma $, every vertex $i\in \cV_{\psi_{2}}$ has degree two with respect to $\gamma \in\Gl$ so either both edges $ \{i, j_1\}, \{i,j_2\}\in \gamma$ that contain $i \in \cV_{\psi_{2}}$ are contained in $\psi_{2}$ or only one of them. Note also that $n_i(\psi_{2})=1$ implies $ i \in \cV_{\psi_{1}}\cap \cV_{\psi_{2}}$. Indeed, with the previous notation, we can assume \Wlog that $ \{i, j_1\} \in \psi_{2}$ while $   \{i,j_2\} \in \psi_{1}  $, i.e. $ i \in \cV_{\psi_{1}}\cap \cV_{\psi_{2}}$.

Now, suppose first that $\psi_{2}$ contains a cycle $\gamma' \in\Gl$. Then $\gamma'\subset\psi_{2}\subset \gamma$, so in fact $\gamma' = \psi_{2} = \gamma$ and, as a consequence of $\gamma\circ\tau = \psi_{1}\circ\psi_{2}\circ\psi_{2}$, we obtain that $\tau =\psi_{1}\circ\psi_{2}$. On the other hand, assume that $ \psi_{2}\neq \gamma$ (in particular, $\psi_{2}$ contains no cycles). Then, $\psi_{2}$ must contain a vertex $i\in \cV_{\psi_{2}}$ with degree $n_i(\psi_{2})=1$ (otherwise $\psi_{2}$ would be an even graph that contains a cycle, by \eqref{eq:gsccycle}). Consider the path connected component $\psi_i\subset \psi_{2}$ that contains $i\in \cV_{\psi_{2}}$. Then, we can write $\psi_i$ as a finite union
		\[ \mu_i = \{ i, j_1\} \circ \{j_1,j_2\} \circ\ldots \circ \{j_{k-1}, j_k\} \]   
for vertices $j_1, \ldots, j_k \in \cV_{\psi_{2}}$ with $j_l\neq j_{l'}$ and $j_l \neq i$ for all $l\neq l'$. The fact that $j_k\neq i$ follows from the assumption that $n_i(\psi_{2})=1$. Similarly, the fact that $j_l\neq j_{l'}$ for all $l\neq l'$ follows from the assumption that $ \psi_{2} $ contains no cycle, and it is also clear that $n_{j_k}(\psi_{2}) =1$, otherwise we could extend the component $\psi_i$ by another edge that contains $j_k\in \cV_{\psi_{2}}$. By the previous remarks, this implies that $j_k\in \cV_{\psi_{1}}\cap \cV_{\psi_{2}}$. Now, checking whether $j_l \in \cV_{\psi_{1}}\cap \cV_{\psi_{2}}$ for each $l$, starting with $j_1\in\cV_{\psi_{2}}$, it is clear that we can write $\psi_i$ as an edge disjoint union of paths in $\Gp$ with endpoints in $\cV_{\psi_{1}}$. 

The previous arguments show that $\psi_{2} = \gamma $ (in which case $\tau = \psi_{1}\circ\psi_{2}$) or $\psi_{2}$ can be written as an edge disjoint union of paths in $\Gp$ with endpoints in $\cV_{\psi_{1}}$. Now, suppose $ |\cV_{\psi_{1}}\cap \cV_{\psi_{2}}|\geq 2$ and suppose that $\psi_{2} = \gamma \in \Gl$. Then, writing   
		\[\psi_{2} =   \{j_1,j_2\} \circ\ldots \circ \{j_{k-1}, j_k\} \circ \{j_k,j_1\},\]
we find at least two vertices $ j_{l_1}, j_{l_2} \in \cV_{\psi_{1}}\cap \cV_{\psi_{2}}$ so that $\psi_{2}$ equals an edge disjoint union of paths in $\Gp$ with endpoints in $\cV_{\psi_{1}}$.  Finally, if $ |\cV_{\psi_{1}}\cap \cV_{\psi_{2}}|\leq 1$, we claim that every vertex $i\in \cV_{\psi_{2}}$ has degree $n_i(\psi_{2}) =2$ so that $\psi_{2}$ contains a cycle which implies $\psi_{2} = \gamma$. Indeed, if $|\cV_{\psi_{1}}\cap \cV_{\psi_{2}}| =0 $, $\psi_{2}$ cannot contain a vertex $i\in \cV_{\psi_{2}}$ with degree $n_i(\psi_{2})=1$, by the previous remarks. On the other hand, if there is a vertex $i\in \cV_{\psi_{2}}$ with degree $n_i(\psi_{2})=1$, then consider the path connected component $\psi_i$ that contains $i\in \cV_{\psi_{2}}$. Arguing as above, we see that $\psi_i$ is a disjoint union of paths in $\Gp$ with endpoints in $\cV_{\psi_{1}}$ and that at least one additional vertex $j_k\in \cV_{\psi_{2}}$ has degree $n_{j_k}(\psi_{2})=1$, contradicting the assumption $ |\cV_{\psi_{1}}\cap \cV_{\psi_{2}}|\leq 1$. This concludes the proof of (1). 

Let us now switch to the proof of (2). Here, we use the identities
		\begin{align*}
		|\cE_{\gamma}| + |\cE_{\tau}|&=|\cE_{\psi_{1}}| + 2|\cE_{\psi_{2}}|\,, \hspace{0.5cm}|\cV_{\gamma} \cup \cV_{\tau}| = |\cV_{\psi_{1}}| + |\cV_{\psi_{2}}| -|\cV_{\psi_{1}}\cap\cV_{\psi_{2}}|\,,
		\end{align*}
which implies (2), because if we assume $\psi_{2} = \emptyset$, then 
	\begin{equation*}
	|\cE_{\psi_{1}}| -  |\cV_{\psi_{1}}|  = |\cE_{\gamma}| + |\cE_{\tau}| - |\cV_{\gamma} \cup \cV_{\tau}|  = |\cE_{\gamma}| + |\cE_{\tau}| - |\cV_{\gamma}| - |\cV_{\tau}|+ |\cV_{\gamma}	\cap\cV_{\tau}| \geq |\cV_{\gamma}\cap\cV_{\tau}|. 
	\end{equation*}
 
To prove (3), suppose first $\eta_{1}=\emptyset$. Then $|\psi^{-1} (\{ (\eta_{1},\eta_{2})\} )|\neq 0 $  if and only if $ \eta_{2}\in \Gl $, because if $\eta_{2}$ is not a loop, then $\eta_{2}\circ\eta_{2}$ does not contain a simple loop, while $\gamma\circ\tau$ always contains simple loops, for $\gamma\in \Gl$ and $\tau \in \Gsc$. If $ \eta_{2}\in\Gl$, on the other hand, the identity $\gamma\circ\tau = \eta_{2}\circ\eta_{2}$ for $(\gamma,\tau) \in \Gl\times \Gsc$ implies $\gamma = \eta_{2}$ (since $\eta_{2}\subset\gamma$) and thus $\tau = \eta_{2}$. This shows in particular that \eqref{eq:prebnd} is true if $\eta_{1}=\emptyset$.

So, assume that $\eta_{1} \not =\emptyset$. Since every $(\gamma,\tau)\in \psi^{-1}( \{ (\eta_{1},\eta_{2})\})\subset \Gl\times \Gsc $ contains a cycle $\gamma \subset \eta_{1}\circ\eta_{2}$, which in fact determines $\tau$ through $\gamma\circ\tau = \eta_{1}\circ\eta_{2}\circ\eta_{2}$, notice that
		\[ |\psi^{-1}( \{ (\eta_{1},\eta_{2})\})|\leq  | \{  \gamma \in\Gl: \eta_{2} \subset \gamma\subset \eta_{1} \circ\eta_{2} \}|,   \]
so it is enough to count the number of cycles contained in $\eta_{1} \circ\eta_{2}$. Without loss of generality, we can assume that $\psi^{-1}( \{ (\eta_{1},\eta_{2})\})\not =\emptyset$, which implies that $ \eta_{1}\in \Gsc$ and that $\eta_{2}$ either consists of disjoint paths in $\Gp$ with endpoints in $\cV_{\eta_{1}}$ or $\eta_{2}$ is equal to some loop $\gamma\in\Gl$. In the latter case, $| \psi^{-1}( \{ (\eta_{1},\eta_{2})\}) | = 1$, by the proof of (1). So, assume that $ \eta_{1}\in \Gsc$, that $\eta_{2} \not \in \Gl$ (so it also does not contain a cycle) and that $\eta_{2}$ consists of edge disjoint paths in $\Gp$ with endpoints in $\cV_{\eta_{1}}$. 

Here, we start from an arbitrary vertex $ i_1 \in   \cV_{\eta_{1}}$ and pick one of $i_1$'s adjacent vertices in $\eta_{1}$. Then, the adjacent vertex $i_2\in \cV_{\eta_{1}}$ is either not in $\cV_{\eta_{2}}$ or it is an end point $ i_2\in \cV_{\eta_{1}}\cap \cV_{\eta_{2}}$ of a path $\pi \in \Gp$ with $\pi \subset \eta_{2}$ having end points in $\cV_{\eta_{1}}$. In the first case, we continue our walk, going at the next step to an adjacent vertex $ i_3 \in \cV_{\eta_{1}}$ (going backwards is not allowed at each step; in particular $i_3\neq i_1$). In the second case, we also walk to a new vertex in $\cV_{\eta_{1}}$, but here there are two different options: either $ i_2 \in \cV_{\eta_{1}}\cap \cV_{\eta_{2}}$ has an adjacent vertex in $\cV_{\eta_{1}}$ (which is not equal to $i_1$) or we can follow the path $\pi$ that starts at $i_2$ and arrives at $i_3\in \cV_{\eta_{1}}$ (since $\pi\in\Gp$ is simple, we can only walk along one direction on $\pi$ and we can not miss any cycles in $\eta_{1}\circ\eta_{2}$ proceeding this way). Altogether, we arrive at $i_3 \in \cV_{\eta_{1}}$ and now we can repeat the procedure. This way, we arrive after finitely many steps at some vertex along the way twice. This means we have detected a cycle in $\eta_{1}\circ\eta_{2}$. Counting all possible directions that we can follow at each step, we have $ |\cV_{\eta_{1}}| $ vertices from which we can choose our starting point, and then we have $n_{i_k} (\eta_{1})-1$ many directions we can choose at each subsequent step. Since every loop $\gamma \in \eta_{1}\circ\eta_{2}$ will certainly be detected, we find at most
 		\[ |\cV_{\eta_{1}}|  \prod_{k \in \cV_{\eta_{1}}} (n_{k} (\eta_{1})-1) \leq |\eta_{1}| e^{\sum_{k  \in \cV_{\eta_{1}}} \log (n_{k} (\eta_{1})-1)} \leq  |\eta_{1}| e^{2(|\cE_{\eta_{1}}| - |\cV_{\eta_{1}}|)}\]
cycles, where we used $\log (x-1) \leq x -2 $ for $x \geq 2$ and $ |\cE_{\gamma}| = \frac12\sum_{i=1}^N n_i(\gamma) $ for $\gamma \in \Gsc$.
\end{proof}
 
Next, let us turn to the bounds on the expectation in \eqref{eq:R2step3a} that were motivated above in $\textbf{Step 3}$. We have the following lemma. 
\begin{lemma}
\label{lm:Tetambd}
Let $ \eta_1 \in \Gsc$ be such that $|\cV_{\eta_1}| \leq k_N^{n} $ for some $n\in\bN$, let $ S_{\eta_1}^{ij}$ be defined as in \eqref{eq:defSetaij} and define $ T_{\eta_1}^{ij}$ by 
		\be \label{eq:defTetaij} T_{\eta_1}^{ij}:= \big\{ \psi_2(\gamma\circ\tau): (\gamma,\tau)\in S_{\eta}^{ij}  \big\}. \ee
Then, there exists a constant $C>0$ that is independent of $\eta_1$ such that
		\be\label{eq:lm261}\begin{split}
		\bE  \sum_{ \eta_2, \eta_2'\in T_{\eta_1}^{ij} }  w^2 (\eta_2) w^2 (\eta'_2)  \leq C. 
		\end{split}\ee
Assuming that $\eta_1$ is such that $ \eta_2 \neq \emptyset$ for every $ \eta_2 \in T_{\eta_1}^{ij}$, we have that
		\be\label{eq:lm262}\begin{split}
		\bE  \sum_{ \eta_2, \eta_2'\in T_{\eta_1}^{ij} }  w^2 (\eta_2) w^2 (\eta'_2) \leq C ( |\cV_{\eta_1}|^2+1)^2   N^{-2}.
		\end{split}\ee
Finally, if we assume that $ i \not \in \cV_{\eta_1}$ or $j\not \in \cV_{\eta_1}$, then 
		\be\label{eq:lm263}\begin{split}
		\bE  \sum_{ \eta_2, \eta_2'\in T_{\eta_1}^{ij} }  w^2 (\eta_2) w^2 (\eta'_2)  \leq C  ( |\cV_{\eta_1}|+1)^2 N^{-4} . 
		\end{split}\ee
\end{lemma}
\begin{proof} We start with a general bound on $\bE  \sum_{ \eta_2, \eta_2'\in T_{\eta_1}^{ij} }  w^2 (\eta_2) w^2 (\eta'_2)  $
and then consider the different cases in the statement of the lemma separately. In the first step, we use that $ \bE[Z^4]= 3 \sigma^4$ for centered Gaussian random variables $Z \sim \cN(0,\sigma^2)$, which implies
	 \[ \bE[w^2(\eta_2)w^2(\eta_2')]= 3^{|\eta_2 \cap \eta_2'|}\bE\big( w^2(\eta_2)\big) \bE\big( w^2(\eta_2')\big). 
	 \]
Define $\cV^\circ_{\eta_2}: = \cV_{\eta_2} \setminus \cV_{\eta_1}$ as the set of interior points of the paths in $  \Gp$ that are contained in $\eta_2$ and which have endpoints in $\cV_{\eta_1}$. Moreover, denote by $\eta_2^\partial$ the set
		\[ \eta_2^{\partial} := \big\{ e \in \eta_2: e = \{j_1,j_2\} \text{ for } j_1,j_2\in \cV_{\eta_1}\cap \cV_{\eta_2} \big\} \subset \eta_2. \]
That is, $\eta_2^\partial$ corresponds to the edges in $\eta_2$ that connect two vertices in $\cV_{\eta_1} \cap \cV_{\eta_2}$. Then 
	\be \label{eq:eta22'bnd}  |\eta_2\cap \eta_2'| \leq 2| \cV^\circ_{\eta_2}\cap \cV^\circ_{\eta_2'}| + |\eta_2^\partial \cap (\eta_2')^\partial| \leq 2| \cV^\circ_{\eta_2}\cap \cV^\circ_{\eta_2'}| + |\eta_2^\partial| + |(\eta_2')^\partial|. \ee
Given an edge $e \in \eta_2\cap \eta_2' \in \Gs$, the bound in \eqref{eq:eta22'bnd} follows from the fact that $e$ either connects two vertices in $\cV_{\eta_1}\cap \cV_{\eta_2} $ (and thus in $\cV_{\eta_1}\cap \cV_{\eta_2'}$ as well, since $e \in \eta_2'$), or it contains at least one vertex in $ \cV^\circ_{\eta_2} = \cV_{\eta_2} \setminus \cV_{\eta_1}$ (which lies also in $ \cV^\circ_{\eta_2'}$, because $e\in \eta_2'$). Notice that every such vertex in $ \cV^\circ_{\eta_2} \cap \cV^\circ_{\eta_2'}$ is an element of exactly two edges ($ \eta_2$ is a edge-disjoint union of paths in $\Gp$ with endpoints in $\eta_1$), justifying the factor two in front of the first term on the \rhs in \eqref{eq:eta22'bnd}. Hence

	\[\begin{split}
		\bE  \sum_{ \eta_2, \eta_2'\in T_{\eta_1}^{ij} }  w^2 (\eta_2) w^2 (\eta'_2)     
		&\leq  \sum_{ \eta_2, \eta_2'\in T_{\eta_1}^{ij} }  9^{|\cV^\circ_{\eta_2} \cap \cV^\circ_{\eta_2'}|}\bE \big( 3^{|\eta_2^\partial|}w^2(\eta_2)\big)\bE\big(3^{|(\eta_2')^\partial|}w^2(\eta_2')\big).
	\end{split}\]
We now prove that the major contribution to the \rhs of the last bound comes from the case when $|\cV^\circ_{\eta_2} \cap \cV^\circ_{\eta_2'}| = 0$. To this end, recall first that $\bE\,w^2(\eta_2)$ is uniquely determined by the edge occupation numbers $n_{ij} (\eta_2)$ for $1\leq i<j\leq N$. Considering two graphs $\eta_2 \sim \wt \eta_2$ to be equivalent if there exists a bijection $\pi:[N]\to[N]$ such that $\pi( \cV_{\eta_1}) = \cV_{\eta_1} $ and $ n_{ij} (\wt \eta_2) = n_{\pi(i)\pi(j)}(\eta_2)$ for all $i,j\in[N]$, we can rewrite $T_{\eta_1}^{ij} = \cup_{\iota} T(\iota)$ (abbreviating $T (\iota):= T_{\eta_1}^{ij}(\iota)$) as the disjoint union over equivalence classes $\iota$ of unlabeled graphs. Thus
	\begin{equation}\begin{split}
	\label{eq:3ee}
	\bE  \sum_{ \eta_2, \eta_2'\in T_{\eta_1}^{ij} }  w^2 (\eta_2) w^2 (\eta'_2)    &\leq	\sum_{\iota,\iota'} \sum_{\eta_2 \in T(\iota),\eta_2' \in T(\iota')}9^{|\cV^\circ_{\eta_2} \cap \cV^\circ_{\eta_2'}|}\bE\big( 3^{|\eta_2^\partial|}w^2(\eta_2)\big) \bE\big( 3^{|(\eta_2')^\partial|}w^2(\eta_2')\big) \\
	 &=\sum_{\iota,\iota'}  \bE\big( 3^{|\iota^\partial|}w^2(\iota)\big) \bE \big( 3^{|(\iota')^\partial|}w^2(\iota')\big)\sum_{\eta_2 \in T(\iota),\eta_2' \in T(\iota')}9^{|\cV^\circ_{\eta_2} \cap \cV^\circ_{\eta_2'}|}. 
	 \end{split}\end{equation} 
Now, consider any two vertex sets $V, V' \subset [N]\setminus V_{\eta_1} $ \st $ |V|, |V'| <k_N^2$. This applies to $\cV^\circ_{\eta_2} , \cV^\circ_{\eta_2'} $ for every $\eta_2 \in T(\iota) , \eta_2'\in T(\iota')$, because $ \eta_2\subset \gamma$ and $|\gamma|<k_N^2$. We rewrite
	 \[\begin{split}
	 &\sum_{\eta_2 \in T(\iota),\eta_2' \in T(\iota')}9^{|\cV^\circ_{\eta_2} \cap \cV^\circ_{\eta_2'}|} = \!\!\!\!\!\sum_{\substack{ V, V' \subset [N]\setminus V_{\eta_1}: \\|V|, |V'| < k_N^2}}\!\!\!\!\!\!\!\!9^{|V\cap V'|}  |\{\eta_2 \in T(\iota):\cV^\circ_{\eta_2} = V\}|   |\{\eta_2' \in T(\iota'): \cV^\circ_{\eta_2'}= V'\}|, \\
	 \end{split}\]
and notice that $ |\{\eta_2 \in T(\iota):\cV^\circ_{\eta_2}  = V_1\}| =  |\{\eta_2 \in T(\iota):\cV^\circ_{\eta_2}  =  V_2\}|$ whenever $|V_1|  = |V_2|$. 
Given $l\in \bN$, let us therefore abbreviate 
		\[ \lambda_{\iota, l} := |\{\eta_2 \in T(\iota):\cV^\circ_{\eta_2} = V\}|\]
for some and hence all $V\subset [N]\setminus \cV_{\eta_1}$ with $|V|=l$. Then, for $l=|V|, l'= |V'| < k_N^2$, we have for large $N$ ($\gg k_N^2$ and $\gg k_N^n\geq |\cV_{\eta_1}|$ ) that
	\[ \begin{split}
	 	\frac{|\{(V,V'):  |V|=l, |V'|=l',  |V \cap V'| = k \}|}{|\{(V,V'): |V|=l, |V'|=l', |V \cap V'| = 0\}|} 
		&= \frac{\binom{N -|\cV_{\eta_1}|}{k}\binom{N-|\cV_{\eta_1}|-k}{l-k}\binom{N-|\cV_{\eta_1}|-l-k}{l'-k}}{\binom{N-|\cV_{\eta_1}|}{l} \binom{N-|\cV_{\eta_1}|-l}{l'}}\\ 
		& = \frac{l! l'! }{k! (l-k)! (l'-k)! }\frac{ (N -|\cV_{\eta_1}|- l- k)!}{  (N -|\cV_{\eta_1}|- l)!}\\
		&\leq( 2\km^4)^{k}   N^{-k}
	 \end{split}\]
so that  
	\[\begin{split}
	 \sum_{\eta_2 \in T(\iota),\eta_2' \in T(\iota')}9^{|\cV^\circ_{\eta_2} \cap \cV^\circ_{\eta_2'}|} &\leq \sum_{l, l'=0    }^{k_N^2} \sum_{k=0}^{k_N^2} \lambda_{\iota, l}  \lambda_{\iota', l'} \sum_{\substack{ V, V' \subset [N]:|V\cap V'|=k, \\ |V| =l, |V'|=l'  }  }9^{k }   \\
	 &\leq \sum_{l, l'=0    }^{k_N^2} \sum_{k=0}^{k_N^2} \lambda_{\iota, l}  \lambda_{\iota', l'} \sum_{\substack{ V, V' \subset [N]:|V\cap V'|=0, \\ |V| =l, |V'|=l'   }  } (18  k_N^{4})^k N^{-k}    \\
	  &\leq \sum_{l, l'=0    }^{k_N^2}   \sum_{\substack{ V, V' \subset [N]: \\ |V| =l, |V'|=l'   }  } \lambda_{\iota, l}  \lambda_{\iota, l'} \bigg(\sum_{k\geq 0}(18  k_N^{4})^k N^{-k}\bigg)    \\
	&\leq C |  T(\iota)||  T(\iota')| 
	 \end{split}\]
and therefore, by \eqref{eq:3ee}, that
	\be \label{eq:4thmom0} \begin{split}
	\bE  \sum_{ \eta_2, \eta_2'\in T_{\eta_1}^{ij} }  w^2 (\eta_2) w^2 (\eta'_2)   &\leq C \sum_{\iota,\iota'}  |  T(\iota)||  T(\iota')|  \bE\big( 3^{|\iota^\partial|}w^2(\iota)\big) \bE \big( 3^{|(\iota')^\partial|}w^2(\iota')\big)\\
	&=  C\sum_{\iota,\iota'} \sum_{\eta_2 \in T(\iota),\eta_2' \in T(\iota')}\bE\big( 3^{|\eta_2^\partial|}w^2(\eta_2)\big)\bE\big( 3^{|(\eta_2')^\partial|}w^2(\eta_2')\big) \\
	&= C \bigg( \sum_{\eta_2  \in T_{\eta_1}^{ij}} 3^{|\eta_2^\partial|} \bE \,w^2(\eta_2) \bigg)^2 \,.
	 \end{split}\ee	
	 
Based on 	\eqref{eq:4thmom0}, let us now consider the different cases considered in the statement of the lemma. First of all, without any further assumptions on $\eta_1$ (except the constraint on the size of its vertex set), we note that the constant $C>0$ on the \rhs in \eqref{eq:4thmom0} is independent of $\eta_1$ and, based on the structural results from Lemma \ref{lm:peta} (1), we get 
		\[\begin{split}
		 \sum_{\eta_2  \in T_{\eta_1}^{ij}} 3^{|\eta_2^\partial|}\bE w^2(\eta_2) &=\sum_{l  \geq 0}\sum_{\substack{  \eta_2 = \nu_1\circ\ldots\circ\nu_l\in T_{\eta_1}^{ij}:\\ \nu_m \in \Gp, \nu_m \cap \nu_n=\emptyset \,(\forall \, m\neq n), \\|\cV_{\eta_1} \cap \cV_{\eta_2}| \geq 2 }} 3^{|\eta_2^\partial|} \bE \,w^2(\eta_2)  + \sum_{ \substack{\eta_2\in T_{\eta_1}^{ij}\cap \Gl: \\ |\cV_{\eta_1} \cap \cV_{\eta_2}| \leq 1  }} 3^{|\eta_2^\partial|} \bE \,w^2(\eta_2)\\
		 & \leq \sum_{l \geq 0} 3^l\bigg(\sum_{u,v \in \cV_{\eta_1} } \sum_{\nu \in \Gp^{uv} } \bE w^2(\nu) \bigg)^l + \sum_{\substack{ \eta_2\in   \Gl:\\ \{i,j\}\in \eta_2} }\bE \,w^2(\eta_2)\\
		 & \leq C \sum_{l \geq 0} 3^l |\cV_{\eta_1}|^{2l} N^{-l} + CN^{-2}  \leq C (|\cV_{\eta_1}|^2 +1)N^{-1}\leq C.
		\end{split}\]
Notice that we used that if $\eta_2\in T_{\eta_1}^{ij}\cap \Gl$ with $ |\cV_{\eta_1} \cap \cV_{\eta_2}| \leq 1$, then $\{i,j\}\in\eta_2$ by definition of $T_{\eta_1}^{ij}$ and $ | \eta_2^\partial|=0$. Similarly, $ |\eta_2^\partial|\leq l$ if $\eta_2$ equals an edge disjoint union of $l$ paths in $\Gp$ with endpoints in $\cV_{\eta_1}$. This proves \eqref{eq:lm261}. The bound \eqref{eq:lm262} follows in the same way. Here, we simply use that if $\eta_1$ is such that $\eta_2\neq \emptyset$ for all $\eta_2\in T_{\eta_1}^{ij}$, then 
		\[\begin{split}
		 \sum_{\eta_2  \in T_{\eta_1}^{ij}} 3^{|\eta_2^\partial|}\bE w^2(\eta_2) &=  \sum_{l  \geq 1}\sum_{\substack{  \eta_2 = \nu_1\circ\ldots\circ\nu_l\in T_{\eta_1}^{ij}:\\ \nu_m \in \Gp, \nu_m \cap \nu_n=\emptyset \,(\forall \, m\neq n), \\|\cV_{\eta_1} \cap \cV_{\eta_2}| \geq 2 }} 3^{|\eta_2^\partial|} \bE \,w^2(\eta_2)  + \sum_{ \substack{\eta_2\in T_{\eta_1}^{ij}\cap \Gl: \\ |\cV_{\eta_1} \cap \cV_{\eta_2}| \leq 1  }} 3^{|\eta_2^\partial|} \bE \,w^2(\eta_2)\\
		 & \leq  C \sum_{l \geq 1} 3^l |\cV_{\eta_1}|^{2l} N^{-l} + CN^{-2}  \leq C (|\cV_{\eta_1}|^2 +1)N^{-1}. 
		\end{split}\]

It remains to prove the last bound \eqref{eq:lm263}. If $i \not \in \cV_{\eta_1}$ or $j \not \in \cV_{\eta_1}$, then $ \{i,j\}\not \in \eta_1$ so that $\{i,j\} \in \eta_2$ for all $\eta_2 \in T_{\eta_1}^{ij}$. Here, we obtain additional decay. If $i \in \cV_{\eta_1}$ and $ j \not \in \cV_{\eta_1}$, we can proceed as above, but use instead    
		\[ \begin{split}
		 \sum_{\eta_2  \in T_{\eta_1}^{ij}} \bE w^2(\eta_2) &\leq  \sum_{l  \geq 1}\,\sum_{\substack{  \eta_2 = \nu_1\circ\ldots\circ\nu_l\in T_{\eta_1}^{ij}: \\ \nu_m \in \Gp, \nu_m \cap \nu_n=\emptyset \,(\forall \, m\neq n), \\ \{i,j\}\in \eta_2, |\cV_{\eta_1} \cap \cV_{\eta_2}| \geq 2  }} 3^{|\eta_2^\partial|}\bE \,w^2(\eta_2)  + \sum_{\substack{ \eta_2\in T_{\eta_1}^{ij}\cap  \Gl : \\ |\cV_{\eta_1} \cap \cV_{\eta_2}| \leq 1 }} 3^{|\eta_2^\partial|} \bE \,w^2(\eta_2)\\
		 &\leq \sum_{l \geq 1} 3^{l}\bigg(\sum_{u \in \cV_{\eta_1}} \sum_{\nu \in \Gp^{ju} } \bE w^2(\{i,j\}\circ \nu) \bigg)\bigg(\sum_{u,v \in \cV_{\eta_1} } \sum_{\nu \in \Gp^{uv} } \bE w(\nu)^2 \bigg)^{l-1} +CN^{-2} \\
		&\leq C ( |\cV_{\eta_1}| +1)N^{-2}. 
		\end{split} \]
For the cases $i \not \in \cV_{\eta_1}$, $ j  \in \cV_{\eta_1}$ and both $i\not \in \cV_{\eta_1},j \not \in \cV_{\eta_1}$, we proceed analogously.
\end{proof}


Collecting the previous results and combining them with the key \textbf{Steps 1, 2, 3} outlined above, we are now ready to control the remaining error term $ R_{ij}^{(2)}$.  
\begin{cor}\label{cor:R2}
Let $\beta <1$ and let $R_{ij}^{(2)} $ be defined as in \eqref{eq:defR2}. Then, for every $\epsilon>0$ there exists a constant $C>0$ such that
		\[\max_{i,j \in [N], i\neq j} \big \|R_{ij}^{(2)} \big\|_{L^2(\Omega) } \leq C N^{-3/2+\epsilon}. \]
\end{cor}
\begin{proof} 
Starting with the bound \eqref{eq:R2step0}, we combine \eqref{eq:R2step1} with Lemma \ref{lm:peta} (\textbf{Step 1}) to conclude \eqref{eq:R2step2} (\textbf{Step 2}) and split the latter as in \eqref{eq:R2step3b} (\textbf{Step 3}). Thus, we arrive at
		\be\label{c}\begin{split}
		\bE \big( R_{ij}^{(2)}\big)^2 
		& \leq CN \sum_{l\leq 1}  \sum_{ \substack{ \eta_1 \in \Gsc: \{i,j\}\subset \cV_{\eta_1}, \\ |\cV_{\eta_1}|< k_N^2+k_N^4, \\ | \cE_{\eta_1}|-|\cV_{\eta_1}|=l }} (1+|\cV_{\eta_1}|)^2  \, \bE \,w^2(\eta_1) \, \bE \bigg(  \sum_{ \eta_2\in T_{\eta_1}^{ij} }  w^2 (\eta_2)  \bigg)^2 \\
		&\hspace{0.5cm} + CN \sum_{l\geq 2}  \sum_{ \substack{ \eta_1 \in \Gsc: \{i,j\}\subset \cV_{\eta_1}, \\ |\cV_{\eta_1}|< k_N^2+k_N^4, \\ | \cE_{\eta_1}|-|\cV_{\eta_1}|=l }} (1+|\cV_{\eta_1}|)^2 e^{4l} \, \bE \,w^2(\eta_1) \, \bE \bigg(  \sum_{ \eta_2\in T_{\eta_1}^{ij} }  w^2 (\eta_2)  \bigg)^2\\
		&\hspace{0.5cm} + CN   \sum_{ \substack{ \eta_1 \in \Gsc:\\ \{i,j \}\cap \cV_{\eta_1}^c\neq \emptyset,\\ |\cV_{\eta_1}|< k_N^2+k_N^4  }} (1+|\cV_{\eta_1}|)^2 e^{4 (|\cE_{\eta_1}|-|\cV_{\eta_1}|)} \, \bE \,w^2(\eta_1) \, \bE \bigg(  \sum_{ \eta_2\in T_{\eta_1}^{ij} }  w^2 (\eta_2)  \bigg)^2. 
		\end{split}\ee
By Lemma \ref{lm:peta}, we know that $| \cE_{\eta_1}|-|\cV_{\eta_1}|\leq 1 $ implies that $\eta_2\neq \emptyset $ for all $\eta_2 \in T_{\eta_1}^{ij}$. Applying Lemma \ref{lm:Akl} and Lemma \ref{lm:Tetambd}, we thus conclude that 
		\[\begin{split}
		 CN \sum_{l\leq 1}  \sum_{ \substack{ \eta_1 \in \Gsc: \{i,j\}\subset \cV_{\eta_1}, \\ |\cV_{\eta_1}|< k_N^2+k_N^4, \\ | \cE_{\eta_1}|-|\cV_{\eta_1}|=l }} (1+|\cV_{\eta_1}|)^2  \, \bE \,w^2(\eta_1) \, \bE \bigg(  \sum_{ \eta_2\in T_{\eta_1}^{ij} }  w^2 (\eta_2)  \bigg)^2&\leq CN^{-3+\epsilon}
		\end{split}\]
for every small $\epsilon>0$ (bounding $ k_N^n \ll N^{\epsilon}$ for every fixed $n\in\bN$) and $N$ large enough. Based on the same lemmas, we obtain that
		\[\begin{split}
		&CN \sum_{l\geq 2}  \sum_{ \substack{ \eta_1 \in \Gsc: \{i,j\}\subset \cV_{\eta_1}, \\ |\cV_{\eta_1}|< k_N^2+k_N^4, \\ | \cE_{\eta_1}|-|\cV_{\eta_1}|=l }} (1+|\cV_{\eta_1}|)^2 e^{4l} \, \bE \,w^2(\eta_1) \, \bE \bigg(  \sum_{ \eta_2\in T_{\eta_1}^{ij} }  w^2 (\eta_2)  \bigg)^2 \\
		&\leq C N^{-1 } k_N^{8} \sum_{l=2}^{2k_N^8}    e^{4l} N^{-l(1-\epsilon')}  \leq CN^{-3+\epsilon}
		\end{split}\]
for every small $\epsilon', \epsilon> 0$ and, analogously, that 
		\[CN \!\!\!\!\!\!  \sum_{ \substack{ \eta_1 \in \Gsc:\\ \{i,j \}\cap \cV_{\eta_1}^c\neq \emptyset,\\ |\cV_{\eta_1}|< k_N^2+k_N^4  }} (1+|\cV_{\eta_1}|)^2 e^{4 (|\cE_{\eta_1}|-|\cV_{\eta_1}|)} \, \bE \,w^2(\eta_1) \, \bE \bigg(  \sum_{ \eta_2\in T_{\eta_1}^{ij} }  w^2 (\eta_2)  \bigg)^2\leq CN^{-3+\epsilon}. \]		
\end{proof}
 
As a corollary of Lemma \ref{lm:R134567} and Corollary \ref{cor:R2}, we can reduce the comparison of $\tbf{M}$ with the resolvent $(1+\beta^2- \beta \tbf{G})^{-1}$ to that of comparing its main contribution $\tbf{P} = (p_{ij})_{1\leq i,j\leq N}\in\bR^{N\times N}$, defined in \eqref{eq:defP},
with $(1+\beta^2- \beta \tbf{G})^{-1}$. 
The approximation of $\tbf{M}$ by $\tbf{P}$ is summarized in the next corollary which gives precise meaning to \eqref{eq:heuristics2} and which proves the first half of Theorem \ref{thm:main}. In the proof of the corollary, we make use of the main result of \cite{ALR}, which tells us that the fluctuation term $\wh Z_N$, defined in \eqref{eq:ZNgraph}, converges in distribution to a log-normal variable. 

\begin{prop}(\cite[Prop. 2.2]{ALR}) \label{prop:ALR}
Let $\beta <1$ and set 
		\[ \sigma^2 = \sum_{k\geq 3} \frac{\beta^{2k}}{2k}. \]
Then $\wh Z_N$, defined in \eqref{eq:ZNgraph}, converges in distribution to a log-normal variable
		\[  \lim_{N\to \infty}\wh Z_N \stackrel{\emph d }{=} \exp\Big (Y - \frac{\sigma^2}2 \Big),  \]
where $Y \sim \cN(0,\sigma^2)$ denotes a centered, Gaussian random variable of variance $\sigma^2$. 
\end{prop}

\begin{cor}\label{cor:Mmain} Let $\beta <1 $, let $ \tbf{M}\in \bR^{N\times N}$ be defined as in \eqref{eq:mij} and let $ \tbf{P}\in \bR^{N\times N}$ be defined as in \eqref{eq:defP}. Then, we have in the sense of probability that 
		\[ \lim_{N\to \infty} \| \tbf{M} - \tbf{P} \|_{\emph{F}} =0. \]
\end{cor}

\begin{proof} 
Since $ \langle \sigma_i \sigma_i\rangle = 1 = p_{ii}$, it is enough to compare the off-diagonal elements of $\tbf M$ and $\tbf P$. Recalling the identity \eqref{eq:coruseR1to7}, we have for $i\neq j$ that
		\[\begin{split}
		(\tbf{M} - \tbf{P})_{ij} & =  (1-\tanh^2(\beta g_{ij}) ) \wh Z_N^{-1} \Big( R_{ij}^{(1)}-R_{ij}^{(2)} -R_{ij}^{(3)}-R_{ij}^{(4)} + R_{ij}^{(5)}    \Big) -  R_{ij}^{(6)} - R_{ij}^{(7)}. 
		\end{split}\]
Combining the fact that $(\sum_{k=1}^{7} R_{ij}^{(k)})^2 \leq C \sum_{k=1}^{7}( R_{ij}^{(k)})^2 $ with $ \tanh^2(.)\leq 1$, we obtain  
\[\begin{split}
		&\bP \big(   \| \tbf{M} - \tbf{P} \|_{\text{F}}^2 >\delta^2    \big) \\
	        & \leq   \bP \bigg(  C  \sum_{1\leq i\neq j \leq N}  \wh Z_N^{-2} \sum_{k=1}^5 \big(R_{ij}^{(k)}\big)^2+\sum_{k=6}^{7}\big(R_{ij}^{(k)}\big)^2 >  \delta^2  \bigg) \\
		& \leq  \sum_{k=1}^5 \bP \bigg(    \sum_{1\leq i\neq j \leq N} \big( R_{ij}^{(k)} \big)^2 >  \frac{\wh Z_N^2 \delta^2}{C'}    \bigg) + \sum_{k=6}^7\bP \bigg(    \sum_{1\leq i\neq j \leq N} \big( R_{ij}^{(k)} \big)^2 >  \frac{\delta^2}{C'}    \bigg).
		\end{split}\]
Applying Markov's inequality on the \rhs of the last bound yields		
		\[\begin{split}
		&\bP \big(  \| \tbf{M} - \tbf{P} \|_{\text{F}} >\delta     \big))\\
		&\leq  \bP\big(  \wh Z_N \leq \epsilon   \big)+  \sum_{\substack{ k=1,\ldots, 5, \\ k\neq 4 }} \bP \bigg(   \sum_{1\leq i\neq j \leq N} \big( R_{ij}^{(k)} \big)^2 >   \frac{\delta^2 \epsilon^2}{C'}    \bigg)  +  C' \delta^{-2}      \sum_{1\leq i\neq j \leq N}     \| \wh Z_N^{-1} R_{ij}^{(4)}  \|_{L^2(\Omega)}^2  \\
		& \hspace{0.5cm}+  C' \delta^{-2}      \sum_{1\leq i\neq j \leq N}  \sum_{k=6}^7  \| R_{ij}^{(k)}  \|_{L^2(\Omega)}^2  \\
		& \leq \bP\big(   \wh Z_N \leq \epsilon  \big) + C' \delta^{-2}N^{2}   (1+ \epsilon^{-2} \big) \max_{k\in \{1,2,4,5,6,7\}} \max_{i,j\in [N]:i\neq j}  \| R_{ij}^{(k)}  \|_{L^2(\Omega)}^2  \\
		&\hspace{0.5cm} +   \bP\Big(   \max_{i,j\in [N]:i\neq j} | R_{ij}^{(3)}|^2 >  N^{-2} \delta^2 \epsilon^2/ C' \Big)
		\end{split}\]	
for every $\delta, \epsilon > 0$. 	By Lemma \ref{lm:R134567} and Corollary \ref{cor:R2}, we conclude that 
		\[\begin{split}
		\limsup_{N\to\infty }\bP \big(  \| \tbf{M} - \tbf{P} \|_{\text{F}} >\delta    \big)&\leq \limsup_{N\to \infty}\Big(  \bP\big( \big \{ \wh Z_N \leq \epsilon \big\} \big)  + C' \delta^{-2}   (1+ \epsilon^{-2} \big) CN^{-1+\epsilon'} +o(1)\Big)\\
		& = \bP\big(    Y  \leq \sigma^2/2+ \log \epsilon \big), 
		\end{split}\]
for every $\delta>0$. Here $o(1)$ denotes an error with $o(1)\to 0$ as $N\to \infty$ and $ Y \sim \cN(0, \sigma^2)$ denotes a Gaussian random variable with variance as in Proposition \ref{prop:ALR}. Since $\epsilon>0$ was arbitrary, the claim follows by sending $\epsilon\to 0$.
\end{proof}
\section{Proof of Theorem \ref{thm:main}}\label{sec:proofmain}
In this section we conclude the proof of Theorem \ref{thm:main}. Throughout this section, we assume that for some sufficiently small $\epsilon>0$, we have that $$ \|\beta \tbf G\|_{\text {op}} \leq 2\beta +(1-\beta)^2 \epsilon/2.$$ 
It is well-known (see \eg \cite{BY}) that this holds true on a set of probability at least $1-o(1)$, for some error $o(1) $ (that depends on $\epsilon$) such that $\lim_{N\to \infty}o(1)=0$.
 
As a consequence of Corollary \ref{cor:Mmain}, Theorem \ref{thm:main} follows if we prove that $ \tbf P\in \bR^{N\times N}$, defined in \eqref{eq:defP}, converges in norm to $ (1+\beta^2- \beta \tbf G)^{-1}$. The assumption on $ \|  \tbf G\|_{\text {op}} $ implies $ \| (1+\beta^2-\beta \tbf G)^{-1}\|_{\text {op}}\leq (1-\beta)^{-2}( 1+\epsilon )$ \st it is enough to prove that the matrix 
		\be\label{eq:defQ}\tbf{Q}= (q_{ij})_{1\leq i,j\leq N}:=\tbf P (1+\beta^2- \beta \tbf G) - \text{id}_{\bR^N} \in \bR^{N\times N} \ee
converges in norm to zero. Before providing the details on this, let us point out two crucial cancellation mechanisms that readily indicate the smallness of $\tbf Q$. To ease the notation, here and from now on all summations over matrix indices run over $[N]$ (recall also that $g_{ii} = 0$). Using the graphical representation of $p_{ij}$ for $i \neq j$, we have that  
		\be\label{eq:qijheuristic}\begin{split}
		 q_{ij} &=  p_{ij}  - \delta_{ij} -\beta g_{ij} + \beta^2 p_{ij}   - \beta \sum_{k: k \neq i} p_{ik} g_{kj}  \\
		 & = \bigg(  \sum_{\substack{ \gamma \in \Gp^{ij}: |\gamma|\geq 2 }} \,\,\prod_{e\in \gamma} \beta g_{e}\bigg) \textbf{1}_{i\neq j}+ \beta^2 p_{ij}- \beta \sum_{k: k \neq i} p_{ik} g_{kj}
		\end{split}\ee
and our main task is to show that the main contributions to the first two terms on the \rhs of \eqref{eq:qijheuristic} are cancelled by the last term on the \rhs of \eqref{eq:qijheuristic}, that is, by the sum
		\be\label{eq:qijheuristic2} \beta \sum_{k: k \neq i} p_{ik} g_{kj}.\ee
These cancellations are based on two mechanisms: for $i\neq k$, the matrix element $p_{ik}$ consists of a sum over the weights $\prod_{e\in \gamma} \beta g_{e}$ of the self-avoiding paths $\gamma\in \Gp^{ik}$ from vertex $i$ to $k$. Such paths $\gamma$ either contain the edge $\{j,k\}$ or they do not contain $\{j,k\}$. In the latter case, if additionally $ j\not \in \cV_\gamma$, multiplying the weight  
 $\prod_{e\in \gamma} \beta g_{e}$ by $\beta g_{kj}$ corresponds to the weight of a self-avoiding path from vertex $i$ to vertex $j$ of length greater than two, indicating a cancellation of such contributions with the first term on the \rhs of \eqref{eq:qijheuristic} (the graphs that do not contain $ \{j,k\}$, but that do contain the vertex $j$ turn out to lead to negligible errors). On the other hand, if $\gamma\in \Gp^{ik}$ already contains $\{j, k\}$, then $\gamma$ is equal to a union $\gamma = \gamma'\circ \{j,k\}$ for some self-avoiding path $\gamma'\in \Gp$ from vertex $i$ to $j$. If we factor out $\beta g_{kj}$ from $\prod_{e\in \gamma} \beta g_{e}$, the law of large numbers 
 		\[ \beta^2 \sum_{k:k\neq i} g_{kj}^2 \approx \beta^2  \]
indicates a cancellation of such contributions with the term $\beta^2 p_{ij}$ on the \rhs of \eqref{eq:qijheuristic}.

Despite the simplicity of the cancellation effects outlined above, keeping track of the precise errors that remain to be controlled after taking into account the cancellations involves a number of rather tedious calculations. We summarize this in the next lemma. For its statement, we define matrices $ \tbf{Q}^{(k)} = (q_{ij}^{(k)})_{1\leq i,j\leq N} \in \bR^{N\times N}$, $k=1,\ldots,5$, by 
		\be \label{eq:defQ1to5}\begin{split}
		q_{ij}^{(1)}&:=  \begin{cases} - 2 \sum_{ \substack{    \gamma_1  \in \Gp^{ij}    }}  \Big(\prod_{\substack { e\in \gamma_1  }} \beta g_e\Big) \sum_{ \substack{    \gamma_2 \in \Gl : \\ \cV_{\gamma_1}\cap \cV_{\gamma_2}=\{j\} }}   \prod_{\substack { e'\in \gamma_2   }} \beta g_{e'} &:i\neq j \\ 0& :i=j, \end{cases} \\
		q_{ij}^{(2)}&:= \begin{cases} 0  &:i\neq j\\  -  2 \sum_{ \substack{    \gamma\in \Gl:\, i\in \cV_\gamma  }}  \prod_{\substack { e\in \gamma   }} \beta  g_e& :i=j. \end{cases} \\
		q_{ij}^{(3)}&:= \begin{cases} \beta^2 \sum_{k:k\neq i} g^2_{kj} \sum_{ \substack{   \gamma' \in \Gl : \\ \{i,j\}\in\gamma', \, k\in \cV_{\gamma'}    }}  \prod_{\substack { e\in \gamma':e\neq \{i,j\}   }} \beta g_e   &: i\neq j\\ 0& :i=j,\end{cases} \\
		q_{ij}^{(4)}&:= \begin{cases} 0 &:i\neq j \\  - \beta^2  \sum_{k}  \big( g_{ik}^2 - \bE  g_{ik}^2\big) & :i=j, \end{cases} \\
		q_{ij}^{(5)}&:= \begin{cases}  \beta^2 g_{ij}^2 p_{ij}&: i\neq j\\ 0& :i=j. \end{cases}
		\end{split}\ee

\begin{lemma}\label{lm:decqij} For $ \tbf{P}$ defined as in \eqref{eq:defP} and the $ \tbf{Q}^{(k)}$ defined as in \eqref{eq:defQ1to5}, we have that
		\be\label{eq:Qdec}\begin{split} 
		\tbf{Q} =  \frac{\beta^2}{N} \, \tbf{P}  + \tbf{Q}^{(1)} + \tbf{Q}^{(2)} + \tbf{Q}^{(3)}+ \tbf{P} \tbf{Q}^{(4)}+ \tbf{Q}^{(5)}. 
		\end{split}\ee
  
\end{lemma}
\begin{proof}
We start with the diagonal elements. A direct calculation yields
		\be\label{eq:qii}\begin{split}
		q_{ii} &= \beta^2 \bigg(1 -   \sum_k  g_{ik}^2\bigg) - \sum_{k}   \sum_{ \substack{   \{i,k\}\in \gamma\in \Gl   }}\beta g_{ik} \prod_{\substack { e\in \gamma:e\neq \{i,k\}  }}  \beta g_e \\
		&= \frac{\beta^2}N - \beta^2  \sum_{k }  \big( g_{ik}^2 - \bE  g_{ik}^2\big)  - \sum_{k}   \sum_{ \substack{   \{i,k\}\in \gamma\in \Gl   }}  \prod_{\substack { e\in \gamma   }} \beta  g_e \\
		& =\frac{\beta^2}N p_{ii} - \beta^2  \sum_{k}  \big( g_{ik}^2 - \bE  g_{ik}^2\big)p_{ii} -  2 \sum_{ \substack{    \gamma\in \Gl:\, i\in \cV_\gamma  }}  \prod_{\substack { e\in \gamma   }} \beta  g_e \\
		& = \Big( \frac{\beta^2}N \tbf P     +   \tbf{Q}^{(2)} +   \tbf{P} \tbf{Q}^{(4)}\Big)_{ii}
		\end{split}\ee
for all $i\in [N]$. For $i,j\in[N]$ with $i\neq j$, on the other hand, we use \eqref{eq:defP} and rewrite
		\be\label{eq:qid1}\begin{split}
		q_{ij}  &=  \beta^3 g_{ij}-   \sum_{k:k\neq i} \sum_{ \substack{   \{i,k\}\in \gamma\in \Gl   }} \Big(\prod_{\substack { e\in \gamma:e\neq \{i,k\}  }} \beta g_e\Big)\beta g_{kj}  \\
		&\hspace{0.5cm} + (1+\beta^2)  \sum_{ \substack{   \{i,j\}\in \gamma\in \Gl   }} \prod_{\substack { e\in \gamma:e\neq \{i,j\}  }} \beta g_e -   \sum_k \beta g_{ik} \beta g_{kj}  \\
		&=  \beta^3 g_{ij}-   \sum_{k:k\neq i} \sum_{ \substack{   \gamma\in \Gl: \\ \{i,k\}\in \gamma, |\gamma| \geq 4  }} \Big(\prod_{\substack { e\in \gamma:e\neq \{i,k\}  }} \beta g_e\Big)\beta g_{kj}    \\
		&\hspace{0.5cm} + \sum_{ \substack{   \gamma\in \Gl: \\ \{i,j\}\in \gamma, |\gamma| = 3   }} \prod_{\substack { e\in \gamma:e\neq \{i,j\}  }} \beta g_e-   \sum_k \beta g_{ik} \beta g_{kj}  \\
		&\hspace{0.5cm} +  \sum_{ \substack{   \gamma\in \Gl: \\ \{i,j\}\in \gamma, |\gamma| = 4   }} \prod_{\substack { e\in \gamma:e\neq \{i,j\}  }} \beta g_e -   \sum_{k:k\neq i} \sum_{ \substack{    \gamma\in \Gl: \\ \{i,k\}\in \gamma, |\gamma|  =3  }} \Big(\prod_{\substack { e\in \gamma:e\neq \{i,k\}  }} \beta g_e\Big)\beta g_{kj}  \\
		&\hspace{0.5cm}  + \beta^2   \sum_{ \substack{   \{i,j\}\in \gamma\in \Gl   }} \prod_{\substack { e\in \gamma:e\neq \{i,j\}  }} \beta g_e + \sum_{ \substack{    \gamma\in \Gl: \\ \{i,j\}\in \gamma, |\gamma| \geq 5   }} \prod_{\substack { e\in \gamma:e\neq \{i,j\}  }} \beta g_e  .
		\end{split}\ee
Inserting the representations
		\[\begin{split}  
		 \sum_{ \substack{    \gamma\in \Gl : \\ \{i,k\}\in \gamma, |\gamma| = 3  }} \!\!\! \!\!\! \Big(\prod_{\substack { e\in \gamma:e\neq \{i,k\}  }}\!\!\!  \beta g_e\Big)  
		&=   \sum_{l}  \beta  g_{il} \beta g_{lk} ,  \sum_{ \substack{    \gamma\in \Gl:\\\{i,j\}\in\gamma, |\gamma| =4   }}\!\!\!  \prod_{\substack { e\in \gamma:e\neq \{i,j\}  }} \!\!\! \beta g_e  =\!\!\! \sum_{k,l :k\neq i, l\neq j}   \!\!\!  \beta  g_{il} \beta g_{lk} \beta g_{kj} 
		\end{split}\]
into \eqref{eq:qid1}, we observe that the second line on the \rhs of \eqref{eq:qid1} vanishes and we get	 
		\be\label{eq:qid2} 
		\begin{split}
		q_{ij} &=  \beta^3 g_{ij}-   \sum_{k:k\neq i} \sum_{ \substack{   \gamma\in \Gl: \\ \{i,k\}\in \gamma, |\gamma| \geq 4  }} \Big(\prod_{\substack { e\in \gamma:e\neq \{i,k\}  }} \beta g_e\Big)\beta g_{kj}    \\
		&\hspace{0.5cm} + \sum_{k,l :k\neq i, l\neq j}    \beta  g_{il} \beta g_{lk} \beta g_{kj} -   \sum_{k:k\neq i}  \sum_{l}  \beta  g_{il} \beta g_{lk}\beta g_{kj}  \\
		&\hspace{0.5cm}  + \beta^2   \sum_{ \substack{   \{i,j\}\in \gamma\in \Gl   }} \prod_{\substack { e\in \gamma:e\neq \{i,j\}  }} \beta g_e + \sum_{ \substack{    \gamma\in \Gl: \\ \{i,j\}\in \gamma, |\gamma| \geq 5   }} \prod_{\substack { e\in \gamma:e\neq \{i,j\}  }} \beta g_e \\
		& = \beta^2 \bigg( 1 - \sum_{k:k\neq i} g_{kj}^2\bigg) \beta g_{ij}   -   \sum_{k:k\neq i} \sum_{ \substack{    \gamma\in \Gl : \\ \{i,k\}\in \gamma, |\gamma|\geq 4  }} \Big(\prod_{\substack { e\in \gamma:e\neq \{i,k\}  }} \beta g_e\Big)\beta g_{kj}\\
		&\hspace{0.5cm}  + \beta^2\sum_{ \substack{   \{i,j\}\in \gamma\in \Gl   }} \prod_{\substack { e\in \gamma:e\neq \{i,j\}  }} \beta g_e +    \sum_{ \substack{    \gamma\in \Gl:\\\{i,j\}\in\gamma, |\gamma|\geq 5   }} \prod_{\substack { e\in \gamma:e\neq \{i,j\}  }} \beta g_e \\
		& =: \beta^2 \bigg( 1 - \sum_{k:k\neq i} g_{kj}^2\bigg) \beta g_{ij}    + \beta^2\sum_{ \substack{   \{i,j\}\in \gamma\in \Gl   }} \prod_{\substack { e\in \gamma:e\neq \{i,j\}  }} \beta g_e +   \Sigma_{\geq 5} -  \Sigma_{\geq 4}.
		\end{split}\ee		 
Here, we defined (neglecting for notational simplicity the dependence on $i$ and $j$) 
		\be\label{eq:qid3} \begin{split} 
		\Sigma_{\geq 4}&= \sum_{k:k\neq i} \sum_{ \substack{    \gamma\in \Gl : \\ \{i,k\}\in \gamma, |\gamma|\geq 4  }} \Big(\prod_{\substack { e\in \gamma:e\neq \{i,k\}  }} \beta g_e\Big)\beta g_{kj},\hspace{0.5cm} \Sigma_{\geq 5}  =  \!\!\!\sum_{ \substack{    \gamma\in \Gl:\\\{i,j\}\in\gamma, |\gamma|\geq 5   }} \prod_{\substack { e\in \gamma:e\neq \{i,j\}  }} \beta g_e. 
		\end{split} \ee 
		
Next, let us decompose $\Sigma_{\geq 4}$, keeping in mind the cancellation mechanisms outlined at the beginning of this section. Recalling that $g_{kj} = 0$ for $k=j$, we split
		\[\begin{split}
		\Sigma_{\geq 4}& = \beta^2\sum_{k:k\neq i} g^2_{kj}\;\sum_{ \substack{    \gamma\in \Gl :|\gamma|\geq 4, \\ \{i,k\}\in \gamma, \{k,j\}\in \gamma   }} \Big(\prod_{\substack { e\in \gamma:e\neq \{i,k\}, \{k,j\}  }} \beta g_e\Big) \\
		& \hspace{0.5cm} + \sum_{k:k\neq i} \sum_{ \substack{    \gamma\in \Gl :|\gamma|\geq 4, \\ \{i,k\}\in \gamma, \{k,j\}\not \in \gamma   }} \Big(\prod_{\substack { e\in \gamma:e\neq \{i,k\}  }} \beta g_e\Big)\beta g_{kj}.
		\end{split}\]
Since a cycle $\tau \in \Gl$ that contains the edges $ \{i,j\}, \{i,k\}, \{k,j\}$ must in fact equal the cycle $ \tau =\{i,j\} \circ \{i,k\}\circ \{k,j\} $, which is of length $|\tau| =3$, we have that 
		\[\begin{split}
		&\sum_{ \substack{    \gamma\in \Gl :|\gamma|\geq 4, \\ \{i,k\}\in \gamma, \{k,j\}\in \gamma   }} \Big(\prod_{\substack { e\in \gamma:e\neq \{i,k\}, \{k,j\}  }} \beta g_e\Big) = \sum_{ \substack{    \gamma\in \Gl :|\gamma|\geq 4 \\  \{i,k\} , \{k,j\}\in \gamma, \{i,j\}\not \in \gamma   }} \Big(\prod_{\substack { e\in \gamma:e\neq \{i,k\}, \{k,j\}  }} \beta g_e\Big).
		\end{split} \]
Now, observe that every cycle $\gamma\in\Gl$ that contains $\{i,k\} , \{k,j\}\in \gamma$ and that does not contain $\{i,j\}\not \in \gamma$, can be identified uniquely with a cycle $\gamma' $ of degree $|\gamma'| = |\gamma|- 1$ with $ \{i,j\}\in \gamma'$ and $k\not \in \cV_{\gamma'}$. In fact, given a cycle $\gamma$ with these properties, the edges $\cE_{\gamma}\setminus\{\{i,k\} , \{k,j\} \} $ determine a unique, self-avoiding path from vertex $i$ to vertex $j$ that avoids vertex $k$, and closing this path through the edge $\{i,j\}$ yields $\gamma'$. Thus 
		\[\begin{split}
		\sum_{ \substack{    \gamma\in \Gl :|\gamma|\geq 4, \\ \{i,k\}\in \gamma, \{k,j\}\in \gamma,   }} \Big(\prod_{\substack { e\in \gamma:e\neq \{i,k\}, \{k,j\}  }} \beta g_e\Big) &= \sum_{ \substack{  \{i,j\}\in  \gamma' \in \Gl      }} \Big(\prod_{\substack { e\in \gamma':e\neq \{i,j\}   }} \beta g_e\Big) \\
		&\hspace{0.5cm} - \sum_{ \substack{   \gamma' \in \Gl : \\ \{i,j\}\in \gamma', \, k\in \cV_{\gamma'}    }} \Big(\prod_{\substack { e\in \gamma':e\neq \{i,j\}   }} \beta g_e\Big)
		\end{split} \]	
so that $\Sigma_{\geq 4}$ can be written as
		\[\begin{split}
		\Sigma_{\geq 4}& = \beta^2\sum_{k:k\neq i} g^2_{kj}\; \!\!\!\sum_{ \substack{  \{i,j\}\in  \gamma' \in \Gl      }}\!\!\! \Big(\prod_{\substack { e\in \gamma':e\neq \{i,j\}   }} \beta g_e\Big)   - \beta^2\sum_{k:k\neq i} g^2_{kj}\!\!\! \!\! \sum_{ \substack{   \gamma' \in \Gl : \\ \{i,j\}\in \gamma', \, k\in \cV_{\gamma'}    }} \!\!\!\Big(\prod_{\substack { e\in \gamma':e\neq \{i,j\}   }} \beta g_e\Big) \\
		& \hspace{0.5cm} + \sum_{k:k\neq i} \sum_{ \substack{    \gamma\in \Gl :|\gamma|\geq 4, \\ \{i,k\}\in \gamma, \{k,j\}\not \in \gamma   }} \Big(\prod_{\substack { e\in \gamma:e\neq \{i,k\}  }} \beta g_e\Big)\beta g_{kj}.
		\end{split}\]
		
Similarly, suppose that $\gamma\in\Gl$ with $|\gamma|\geq 4$ is a cycle that contains the edge $\{i,k\}\in \gamma$, but such that $j\not \in \cV_{\gamma}$ (in particular $ \{k,j\} \not \in \gamma$). Such a loop can be identified uniquely with a self-avoiding path from vertex $i$ to vertex $k$, avoiding vertex $j$. Such a path, on the other hand, can be identified uniquely with a cycle $\gamma'$ that contains the edges $\{i,j\}$ and $\{j,k\}$, with $|\gamma'| = |\gamma| +1$. What this implies is that 
 		\[\begin{split}
		&\sum_{k:k\neq i} \sum_{ \substack{    \gamma\in \Gl :|\gamma|\geq 4, \\ \{i,k\}\in \gamma, \{k,j\}\not \in \gamma   }} \Big(\prod_{\substack { e\in \gamma:e\neq \{i,k\}  }} \beta g_e\Big)\beta g_{kj} \\
		&=   \sum_{k:k\neq i} \,\sum_{ \substack{    \gamma'\in \Gl :|\gamma'|\geq 5 \\ \{i,j\} ,\{ j,k\} \in \gamma'  }}  \prod_{\substack { e\in \gamma':e\neq \{i,j\}  }} \beta g_e  + \sum_{k:k\neq i} \sum_{ \substack{    \gamma\in \Gl :|\gamma|\geq 4, \\ \{i,k\}\in \gamma, \{k,j\}\not \in \gamma, j\in \cV_{\gamma}   }} \Big(\prod_{\substack { e\in \gamma:e\neq \{i,k\}  }} \beta g_e\Big)\beta g_{kj} \\
		&=   \sum_{ \substack{    \gamma'\in \Gl :\\ |\gamma'|\geq 5,  \{i,j\}\in \gamma'   }} \prod_{\substack { e\in \gamma':e\neq \{i,j\}  }} \beta g_e +\sum_{k:k\neq i} \sum_{ \substack{    \gamma\in \Gl :|\gamma|\geq 4, \\ \{i,k\}\in \gamma, \{k,j\}\not \in \gamma, j\in \cV_{\gamma}  }} \Big(\prod_{\substack { e\in \gamma:e\neq \{i,k\}  }} \beta g_e\Big)\beta g_{kj}\\
		&=   \Sigma_{\geq 5} +\sum_{k:k\neq i} \sum_{ \substack{    \gamma\in \Gl :|\gamma|\geq 4, \\ \{i,k\}\in \gamma, \{k,j\}\not \in \gamma, j\in \cV_{\gamma}  }} \Big(\prod_{\substack { e\in \gamma:e\neq \{i,k\}  }} \beta g_e\Big)\beta g_{kj},
		\end{split} \]
where in the second step, we used the definition of $\Sigma_{\geq 5}$ in \eqref{eq:qid3} and the fact that the set of cycles that contain $ \{i,j\}, \{j,k\} $ is disjoint from the set of cycles that contain $ \{i,j\}, \{j,k'\} $ for $k\neq k'$. Combining the previous two identities, we infer that
		\be\label{eq:qid4}\begin{split}
		\Sigma_{\geq 5}-\Sigma_{\geq 4} & = - \beta^2\sum_{k:k\neq i} g^2_{kj}\; \!\!\!\sum_{ \substack{  \{i,j\}\in  \gamma' \in \Gl      }}\!\!\! \Big(\prod_{\substack { e\in \gamma':e\neq \{i,j\}   }} \beta g_e\Big)\\
		&\hspace{0.5cm} +\beta^2\sum_{k:k\neq i} g^2_{kj}\!\!\! \!\! \sum_{ \substack{   \gamma' \in \Gl : \\ \{i,j\}\in \gamma', \, k\in \cV_{\gamma'}    }} \!\!\!\Big(\prod_{\substack { e\in \gamma':e\neq \{i,j\}   }} \beta g_e\Big) \\
		&\hspace{0.5cm} - \sum_{k:k\neq i} \sum_{ \substack{    \gamma\in \Gl :|\gamma|\geq 4, \\ \{i,k\}\in \gamma, \{k,j\}\not \in \gamma, j\in \cV_{\gamma}  }} \Big(\prod_{\substack { e\in \gamma:e\neq \{i,k\}  }} \beta g_e\Big)\beta g_{kj}.
		\end{split}\ee
Now, collecting the identities \eqref{eq:qid2}, \eqref{eq:qid3} and \eqref{eq:qid4} and recalling the definition of $p_{ij}$ in \eqref{eq:defP}, we obtain for $i\neq j$ the decomposition
		\be\label{eq:qij5}\begin{split}
		q_{ij}  & =   \beta^2\bigg(1-\sum_{k:k\neq i}g^2_{kj}\bigg) \bigg( \beta   g_{ij} + \sum_{ \substack{   \{i,j\}\in \gamma\in \Gl   }} \prod_{\substack { e\in \gamma:e\neq \{i,j\}  }} \beta g_e  \bigg) \\
		&\hspace{0.5cm}  + \beta^2 \sum_{k:k\neq i} g^2_{kj} \sum_{ \substack{   \gamma' \in \Gl : \\ \{i,j\}\in\gamma', \, k\in \cV_{\gamma'}    }}  \prod_{\substack { e\in \gamma':e\neq \{i,j\}   }} \beta g_e\\
		&\hspace{0.5cm} - \sum_{k:k\neq i} \sum_{ \substack{    \gamma\in \Gl :|\gamma|\geq 4, \\ \{i,k\}\in \gamma, \{k,j\}\not \in \gamma, j\in \cV_{\gamma}  }} \Big(\prod_{\substack { e\in \gamma:e\neq \{i,k\}  }} \beta g_e\Big)\beta g_{kj}\\ 
		& =  \frac{\beta^2}N p_{ij} - \beta^2 \bigg(\sum_{k }\big( g^2_{kj} - \bE g^2_{kj}\big) \bigg) p_{ij}  + \beta^2 \sum_{k:k\neq i} g^2_{kj} \sum_{ \substack{   \gamma' \in \Gl : \\ \{i,j\}\in\gamma', \, k\in \cV_{\gamma'}    }}  \prod_{\substack { e\in \gamma':e\neq \{i,j\}   }} \beta g_e \\
		&\hspace{0.5cm}+ \beta^2 g_{ij}^2 p_{ij}  - \sum_{k:k\neq i} \sum_{ \substack{    \gamma\in \Gl :|\gamma|\geq 4, \\ \{i,k\}\in \gamma, \{k,j\}\not \in \gamma, j\in \cV_{\gamma}  }} \Big(\prod_{\substack { e\in \gamma:e\neq \{i,k\}  }} \beta g_e\Big)\beta g_{kj} \\
		& = \Big( \frac{\beta^2}N  \tbf P + \tbf Q^{(3)}+  \tbf P \tbf Q^{(4)} + \tbf Q^{(5)}\Big)_{ij} - \sum_{k:k\neq i} \sum_{ \substack{    \gamma\in \Gl :|\gamma|\geq 4, \\ \{i,k\}\in \gamma, \{k,j\}\not \in \gamma, j\in \cV_{\gamma}  }} \Big(\prod_{\substack { e\in \gamma:e\neq \{i,k\}  }} \beta g_e\Big)\beta g_{kj}. 
		\end{split}\ee
Finally, to connect the last term on the \rhs in \eqref{eq:qij5} with $\tbf Q^{(1)}$, notice that
		\be\label{eq:loopxcyc}\begin{split}
		& \sum_{k:k\neq i} \sum_{ \substack{    \gamma\in \Gl :|\gamma|\geq 4, \\ \{i,k\}\in \gamma, \{k,j\}\not \in \gamma, j\in \cV_{\gamma}  }} \Big(\prod_{\substack { e\in \gamma:e\neq \{i,k\}  }} \beta g_e\Big)\beta g_{kj}\\
		& =   \sum_{k:k\neq i,j} \sum_{ \substack{    (\gamma_1, \gamma_2) \in \Gp\times \Gl : \\ n_i(\gamma_1)=n_j(\gamma_1)=1,\\ \{j,k\}\in  \gamma_2, \cV_{\gamma_1}\cap \cV_{\gamma_2}=\{j\} }} \Big(\prod_{\substack { e\in \gamma_1  }} \beta g_e\Big) \Big(\prod_{\substack { e'\in \gamma_2   }} \beta g_{e'}\Big)\\
		&=  2  \sum_{ \substack{    \gamma_1  \in \Gp : \\ n_i(\gamma_1)=n_j(\gamma_1)=1     }}  \Big(\prod_{\substack { e\in \gamma_1  }} \beta g_e\Big) \sum_{ \substack{    \gamma_2 \in \Gl : \\ \cV_{\gamma_1}\cap \cV_{\gamma_2}=\{j\} }}   \prod_{\substack { e'\in \gamma_2   }} \beta g_{e'}.
		\end{split}\ee
Indeed, a loop $\gamma\in\Gl$ that contains the edge $\{i,k\}\in \gamma $ and the vertex $j\in \cV_\gamma$, but not the edge $\{j,k\}\not \in \gamma$ can be uniquely identified with a pair of self-avoiding walks, one going from vertex $i$ to $j$ avoiding vertex $k$, and the other going from vertex $j$ to vertex $k$ having length at least two (because $\{j,k\}\not\in \gamma$). In particular, the intersection of the two walks is given by vertex $j$ only. Since adding the edge $\{j,k\}$ to the second walk turns it uniquely into a cycle $\gamma_2$, we can identify $\gamma \circ \{j,k\}$ with $ \gamma_1\circ\gamma_2 $, where $   \gamma_1\in \Gp$ with endpoints $ i$ and $j$, $ \{j,k\}\in \gamma_2$ and $\cV_{\gamma_1}\cap \cV_{\gamma_2}=\{j\}$. Finally, summing over all $k\in[N]\setminus\{i,j\}$, we obtain the \rhs in \eqref{eq:loopxcyc} and inserting this into \eqref{eq:qij5}, we see that
		\be\label{eq:qij}\begin{split}
		q_{ij}  & =    \Big( \frac{\beta^2}N  \tbf P  +\tbf Q^{(1)}+ \tbf Q^{(3)}+  \tbf P \tbf Q^{(4)} + \tbf Q^{(5)}\Big)_{ij} 
		\end{split}\ee
for $i\neq j$. Comparing \eqref{eq:qii} and \eqref{eq:qij} with \eqref{eq:defQ1to5}, this proves the decomposition \eqref{eq:Qdec}. 
\end{proof}
Recalling the strategy outlined at the end of Section \ref{sec:intro}, our next task is to show that the error matrices $ \tbf Q^{(k)}$ have small operator norms. Looking at \eqref{eq:defQ1to5}, it is in fact not difficult to see that $ \tbf Q^{(2)}, \tbf Q^{(3)}, \tbf Q^{(4)}$ and $\tbf Q^{(5)}$ have small operator norms: 
\begin{itemize}
\item Due to the vertex constraint $i\in \cV_{\gamma}$ in the summation over cycles $\gamma\in\Gl$ in the definition of $ q^{(2)}_{ii}$, it is clear that the entries $q^{(2)}_{ii}$ are typically of order $N^{-1/2}$. In order to see that $ \| \tbf Q^{(2)}\|_{\text{op}} = \max_{i\in [N]} |q_{ii}^{(2)}| $ is small, we find it convenient to derive a fourth moment bound on the small graph contributions to $ q^{(2)}_{ii}$ and to combine this with a second moment bound on the large graph contributions.
\item Based on simple second moment bounds, the matrices $\tbf Q^{(3)}$ and $\tbf Q^{(5)}$ turn out to have small Frobenius (and thus operator) norms. 
\item The entries of $\tbf Q^{(4)}$ vanish by the law of large numbers and it is straightforward to quantify this through standard exponential moment bounds. In particular, this implies the smallness of $ \| \tbf Q^{(4)}\|_{\text{op}} = \max_{i\in [N]} |q_{ii}^{(4)}| $.  
\end{itemize} 

Treating the remaining off-diagonal error $ \tbf Q^{(1)}$ requires a more involved strategy: the (expectation of the) Frobenius norm of $ \tbf Q^{(1)}$ is of order one so that simple moment bounds are not sufficient to derive the smallness of its operator norm. Interestingly, however, it turns out that we can factor out $\tbf P$ from $ \tbf Q^{(1)}$, up to errors that have asymptotically vanishing operator norm. Anticipating to treat the small and large graph contributions separately, as before, let us define for fixed $k,k_1,k_2\geq 3$ the matrices $\tbf{Q}_{>k_1,k_2}^{(1)}, \tbf{Q}_{\leq k_1,k_2}^{(1)}, \tbf{Q}_{>k}^{(2)}$, $ \tbf{Q}_{\leq k}^{(2)}  $, $ \tbf{Q}_{> k_1,k_2}^{(6)}  $ and $ \tbf{Q}_{\leq k_1,k_2}^{(6)}  $  via 
		\be \label{eq:defQ126k}\begin{split}
		q_{\leq k_1,k_2,ij}^{(1)}&:= \begin{cases} - 2 \sum_{ \substack{    \gamma_1  \in \Gp^{ij} : |\gamma_1|\leq k_1    }} \! \big(\prod_{\substack { e\in \gamma_1:e  }} \beta g_e\big) \!\sum_{ \substack{    \gamma_2 \in \Gl :|\gamma_2|\leq k_2, \\ \cV_{\gamma_1} \cap \cV_{\gamma_2}=\{j\} }}   \prod_{\substack { e'\in \gamma_2   }} \beta g_{e'} &:i\neq j \\ 0& :i=j, \end{cases} \\
		q_{>k_1,k_2,ij}^{(1)}&:= \begin{cases}  q_{ij}^{(1)}-q_{\leq k_1,k_2,ij}^{(1)}  &:i\neq j \\  0 & :i=j, \end{cases} \\
		q_{\leq k, ij}^{(2)}&:= \begin{cases} 0  &:i\neq j\\  -  2 \sum_{ \substack{    \gamma\in \Gl:\\ | \gamma|\leq k , i\in \cV_\gamma  }}  \prod_{\substack { e\in \gamma   }} \beta  g_e& :i=j, \end{cases} \hspace{1.2cm} q_{> k,ij}^{(2)}:= \begin{cases} 0  &:i\neq j\\ q_{ij}^{(2)}-q_{\leq k, ij}^{(2)}& :i=j, \end{cases}\\
		q_{> k_1,k_2,ij}^{(6)}&:= \begin{cases} - 2 \sum_{ \substack{    \gamma_1  \in \Gp : |\gamma_1| > k_1\\n_i(\gamma_1)=n_j(\gamma_2)=1    }} \!\!\!\! \big(\prod_{\substack { e\in \gamma_1:e  }} \beta g_e\big) \sum_{ \substack{    \gamma_2 \in \Gl : \\|\gamma_2|\leq k_2, j\in\cV_{\gamma_2}  }}   \prod_{\substack { e'\in \gamma_2   }} \beta g_{e'}&:i\neq j \\  0 & :i=j. \end{cases} \\ 
		q_{\leq k_1,k_2,ij}^{(6)} &:= \begin{cases}- 2\sum_{ \substack{    \gamma_1  \in \Gp : |\gamma_1|\leq k_1\\n_i(\gamma_1)=n_j(\gamma_1)=1    }} \!\!\!\! \big(\prod_{\substack { e\in \gamma_1:e  }} \beta g_e\big) \sum_{ \substack{    \gamma_2 \in \Gl :|\gamma_2|\leq k_2, \\ j\in \cV_{\gamma_2},|\cV_{\gamma_1} \cap \cV_{\gamma_2}|\geq 2 }} \!\!  \prod_{\substack { e'\in \gamma_2   }} \beta g_{e'}&:i\neq j \\  0 & :i=j. \end{cases}\\ 
		\end{split}\ee	
Then, by \eqref{eq:defQ1to5} and \eqref{eq:defQ126k}, we have that
		\[ \tbf{Q}^{(1)} =  \tbf{Q}_{\leq k_1,k_2}^{(1)} + \tbf{Q}_{>k_1,k_2}^{(1)} =  \tbf{Q}_{>k_1,k_2}^{(1)}+ \big( \tbf{P} - \text{id}_{\bR^{N}} \big) \tbf{Q}_{\leq k_2}^{(2)}  - \tbf{Q}_{>k_1,k_2}^{(6)} - \tbf{Q}_{\leq k_1,k_2}^{(6)}. \] 
Combined with the decomposition \eqref{eq:Qdec}, this implies that for every $k_1,k_2\geq 3$, we have
		\be\label{eq:Qdec2}\begin{split} 
		\tbf{Q} &=   \tbf{P} \Big(\frac{ \beta^2}N \text{id}_{\bR^N}   +   \tbf{Q}_{\leq k_2}^{(2)}+ \tbf{Q}^{(4)}\Big) + \tbf{Q}^{(3)} + \tbf{Q}^{(5)}  \\
		&\hspace{0.5cm}+  \tbf{Q}_{>k_1,k_2}^{(1)}  + \tbf{Q}_{> k_2}^{(2)}  - \tbf{Q}_{>k_1,k_2}^{(6)} - \tbf{Q}_{\leq k_1,k_2}^{(6)}.  
		\end{split}\ee	
Instead of controlling $\tbf Q^{(1)}$ directly, the last identity shows that it is sufficient to prove the smallness of $\tbf{Q}_{>k_1,k_2}^{(1)}, \tbf{Q}_{>k_1,k_2}^{(6)}$ and $\tbf{Q}_{\leq k_1,k_2}^{(6)}$ (in addition to the previously mentioned errors). While the matrices $\tbf{Q}_{>k_1,k_2}^{(1)}, \tbf{Q}_{>k_1,k_2}^{(6)}$ involve large graphs, so that second moment bounds are enough to conclude their smallness, the entries of $\tbf{Q}_{\leq k_1,k_2}^{(6)}$ are structurally quite similar to the error term $R^{(2)}_{ij}$, defined in \eqref{eq:defR2} and analyzed in detail in Section \ref{sec:notation}. Adapting the latter analysis to the analysis of $\big(\tbf{Q}_{\leq k_1,k_2}^{(6)}\big)_{ij}$, we show below that $\tbf{Q}_{\leq k_1,k_2}^{(6)}$ has small Frobenius (and thus operator) norm. 

\begin{lemma} Let $\beta <1$, $ \tbf P $ as in \eqref{eq:defP}, $\tbf{Q}$ as in \eqref{eq:defQ} and let $\tbf{Q}^{(k)}$, $k\in\{1,\dots,5\}$, be defined as in \eqref{eq:defQ1to5}. Then, there exists a constant $C = C_{\beta}>0$ such that 
		\be \label{eq:qijbnds1}
		\begin{split}
		 \max_{i,j \in [N], i\neq j} \big \|q_{ij}^{(3)} \big\|_{L^2(\Omega) }\leq  CN^{-3/2}, \; \max_{i,j \in [N], i\neq j}\big \|q_{ij}^{(5)}  \big\|_{L^2(\Omega) } &\leq CN^{-3/2}. 
		 \end{split}
		 \ee
Moreover, for fixed $k,k_1,k_2\geq 3$, let $\tbf{Q}_{>k_1,k_2}^{(1)}, \tbf{Q}_{\leq k_1,k_2}^{(1)}, \tbf{Q}_{>k}^{(2)},  \tbf{Q}_{\leq k}^{(2)},\tbf{Q}_{>k_1,k_2}^{(6)}  $ and $\tbf{Q}_{>k_1,k_2}^{(2)}$ be defined as in \eqref{eq:defQ126k}. 
Then, there exists $C>0$, independent of $N $ and $k$, as well as some constant $ C_k' >0$ that depends on $k $, but that is independent of $N$ so that 
		\be \label{eq:qijbnds2}\begin{split}
		\max_{i,j\in [N] }\big\| q_{>k_1,k_2,ij}^{(1)} \big\|_{L^2(\Omega)} &\leq C e^{- \min(k_1,k_2)\log \beta^{-2} } N^{-1}, \\
		\max_{i\in [N] }\big\| q_{>k,ii}^{(2)} \big\|_{L^2(\Omega)} &\leq C e^{- k\log \beta^{-2} } N^{-1/2}, \hspace{0.5cm}\max_{i \in [N] }\big\| q_{\leq k,ii}^{(2)} \big\|_{L^4(\Omega)} \leq C'_k N^{-1/2}.
		\end{split}\ee	 
Moreover, for $k_N $ as defined in \eqref{eq:defkN}, we have for every $\epsilon>0$ that
		\be \label{eq:qijbnds3}\begin{split}
		\max_{i,j\in [N]: i\neq j} \| q_{\leq k_N^4,k_N^2,  ij}^{(6)} \|_{L^2(\Omega)} \leq N^{-3/2 + \epsilon}\,, \; \lim_{N \to \infty}\bP\Big(\max_{i,j\in [N]: i\neq j}|q_{>k_N^4,k_N^2,  ij}^{(6)}| >N^{-\log N}\Big) = 0.
		\end{split}\ee	
\end{lemma} 
\begin{proof}
By symmetry, it is enough to prove \eqref{eq:qijbnds1}, \eqref{eq:qijbnds2} and \eqref{eq:qijbnds3} for fixed $i,j\in[N]$. 

Let us start with \eqref{eq:qijbnds1}. Neglecting the prefactor $\beta^2$ in $ q_{ij}^{(3)}$, we compute 
		\[\begin{split}
		\bE \bigg(\sum_{k:k\neq i} g^2_{kj} \sum_{ \substack{   \gamma \in \Gl : \\ \{i,j\}\in \gamma, \, k\in \cV_{\gamma}    }}  \prod_{\substack { e\in \gamma:e\neq \{i,j\}   }} \beta g_e\bigg)^2& \leq 2 \bE \bigg(\sum_{k:k\neq i} \beta g^3_{kj} \!\!\sum_{ \substack{   \gamma \in \Gl :  \\ \{i,j\}, \{j,k\}\in \gamma   }} \!\!\! \prod_{\substack { e\in \gamma:e\neq \{i,j\},\{j,k\}   }} \beta g_e\bigg)^2\\
		&\hspace{0.5cm} + 2\bE \bigg(\sum_{k:k\neq i} g^2_{kj} \!\!\!\sum_{ \substack{   \gamma \in \Gl : k\in \cV_{\gamma}, \\ \{i,j\}\in \gamma, \{j,k\}\not\in\gamma     }}  \prod_{\substack { e\in \gamma:e\neq \{i,j\}   }} \beta g_e\bigg)^2.
		\end{split}\]
Now, using that the odd moments of each weight $g_{kj} $ vanish and that the summation runs over $k\in [N]$ with $k\neq i$, such that a graph that contains $ \{i,j\}, \{k,j\}\in \gamma\in \Gl$ can not contain another edge $\{k', j\}$ for $k'\not \in \{ i,k\}$, we obtain with 
		$$ \bE g_{st}^{2p} = (2p-1)!! \big(\bE g_{st}^2\big)^p = (2p-1)!! N^{-p}$$
for $p\in\bN$ that
		\[\begin{split}
		\bE \bigg(\sum_{k:k\neq i} g^3_{kj} \!\!\sum_{ \substack{   \gamma \in \Gl :  \\ \{i,j\}, \{j,k\}\in \gamma   }}  \prod_{\substack { e\in \gamma:\\e\neq \{i,j\},\{j,k\}   }} \beta g_e\bigg)^2&= 15N^{-3}\sum_{k:k\neq i}  \bE \bigg(\!\!\sum_{ \substack{   \gamma \in \Gl :  \\ \{i,j\}, \{j,k\}\in \gamma   }} \, \prod_{\substack { e\in \gamma:\\e\neq \{i,j\},\{j,k\}   }} \beta g_e\bigg)^2\\
		&= 15N^{-3}\sum_{k:k\neq i} \!\!\sum_{ \substack{   \gamma \in \Gl :  \\ \{i,j\}, \{j,k\}\in \gamma   }}  (\beta^2)^{|\gamma|-2}N^{-|\gamma|+2}   \\
		&\leq 15N^{-3}\sum_{k:k\neq i} \sum_{l\geq 3}(\beta^2)^{l-2} N^{l-3} N^{-l+2}\leq CN^{-3}, 
		\end{split}\]	
where we used that there are $ (N- 3)(N-4)\ldots (N - l+1)\leq N^{l-3}$ cycles of length $l \geq 4$ with edges $\{i,j\}, \{j,k\}$, for fixed $i,j, k\in [N]$ (if $l=3$, there is a unique such cycle). Similarly, we can use that $ g_{kj}$ and the weight of any self-avoiding path from vertex $i$ to vertex $j$ that does not contain the edge $\{k,j\}$ are independent, and that the weights of any two different self-avoiding paths are orthogonal in $L^2(\Omega)$. This yields 
		\[\begin{split}
		&\bE \bigg(\sum_{k:k\neq i,j} g^2_{kj} \!\!\!\sum_{ \substack{   \gamma \in \Gl : k\in \cV_{\gamma}, \\ \{i,j\}\in \gamma, \{j,k\}\not\in\gamma     }}  \prod_{\substack { e\in \gamma:e\neq \{i,j\}   }} \beta g_e\bigg)^2\\
		& =   \sum_{k:k\neq i,j} \big( \bE\,g^4_{kj}\big) \,  \bE\, \bigg(\sum_{ \substack{   \gamma \in \Gl : k\in \cV_{\gamma}, \\ \{i,j\}\in \gamma, \{j,k\}\not\in\gamma     }}  \prod_{\substack { e\in \gamma:e\neq \{i,j\}   }} \beta g_e\bigg)^2   \\
		&\hspace{0.5cm} + \bE\!\!\! \sum_{\substack{ k,k':k\neq i,j;\\ k'\neq i,j; k\neq k' }} g^2_{kj} g^2_{k'j} \ \bigg(\sum_{ \substack{   \gamma \in \Gl : k\in \cV_{\gamma}, \\ \{i,j\}\in \gamma, \{j,k\}\not\in\gamma     }}  \prod_{\substack { e\in \gamma:e\neq \{i,j\}   }} \beta g_e\bigg) \bigg(\sum_{ \substack{   \gamma \in \Gl : k'\in \cV_{\gamma'}, \\ \{i,j\}\in \gamma', \{j,k'\}\not\in\gamma'     }}  \prod_{\substack { e\in \gamma':e\neq \{i,j\}   }} \beta g_e\bigg)\\
		& =  3N^{-2} \sum_{k:k\neq i,j}  \sum_{ \substack{   \gamma \in \Gl : k\in \cV_{\gamma}, \\ \{i,j\}\in \gamma, \{j,k\}\not\in\gamma     }}  (\beta^2)^{|\gamma|-1}N^{-|\gamma|+1}   \\
		&\hspace{0.5cm} + \bE\!\!\! \sum_{\substack{ k,k':k\neq i,j;\\ k'\neq i,j; k\neq k' }} g^2_{kj} g^2_{k'j} \ \bigg(\sum_{ \substack{   \gamma \in \Gl : k\in \cV_{\gamma}, \\ \{i,j\}\in \gamma, \{j,k\}\not\in\gamma     }}  \prod_{\substack { e\in \gamma:e\neq \{i,j\}   }} \beta g_e\bigg) \bigg(\sum_{ \substack{   \gamma \in \Gl : k'\in \cV_{\gamma'}, \\ \{i,j\}\in \gamma', \{j,k'\}\not\in\gamma'     }}  \prod_{\substack { e\in \gamma':e\neq \{i,j\}   }} \beta g_e\bigg). 
		\end{split}\]
 Using once again that the odd moments of the edge weights vanish, one can also factorize the expectation in the last term on the \rhs of the previous equation to get
 		\[\begin{split}
		&\bE \bigg(\sum_{k:k\neq i} g^2_{kj} \!\!\!\sum_{ \substack{   \gamma \in \Gl : k\in \cV_{\gamma}, \\ \{i,j\}\in \gamma, \{j,k\}\not\in\gamma     }}  \prod_{\substack { e\in \gamma:e\neq \{i,j\}   }} \beta g_e\bigg)^2\\
		& = 3N^{-2} \!\!\sum_{k:k\neq i,j}  \sum_{ \substack{   \gamma \in \Gl : k\in \cV_{\gamma}, \\ \{i,j\}\in \gamma, \{j,k\}\not\in\gamma     }} \!\!\!\!  (\beta^2)^{|\gamma|-1}N^{-|\gamma|+1}+N^{-2}\!\! \!\!\!\sum_{\substack{ k,k':k\neq i,j;\\ k'\neq i,j; k\neq k' }}   \sum_{ \substack{   \gamma \in \Gl : k, k'\in \cV_{\gamma}, \\ \{i,j\}\in \gamma, \{j,k\}, \{j,k'\}\not\in\gamma     }}  \hspace{-0.6cm} (\beta^2)^{|\gamma|-1}N^{-|\gamma|+1}\\
		&\leq C N^{-2} \sum_{k:k\neq i,j}\sum_{l\geq 3} l (\beta^2)^{l-1} N^{l-3} N^{-l+1} + N^{-2} \!\!\!\sum_{\substack{ k,k':k\neq i,j;\\ k'\neq i,j; k\neq k' }} l^2 (\beta^2)^{l-1} N^{l-4} N^{-l+1} \leq C N^{-3}.
		\end{split}\]
Here, we used that there are not more than $ l N^{l-3} $ cycles of length $l$ with edge $\{i,j\}$ and some fixed vertex $k\neq i,j$, and not more than $l^2 N^{l-4}$ cycles with edge $\{i,j\}$ and two different vertices $k,k'\neq i,j$. In summary, this proves $ \| q_{ij}^{(3)}\|_{L^2(\Omega)} \leq CN^{-3/2}$. 
		 
Next, using once again the orthogonality of the weights of two different self-avoiding paths in $L^2(\Omega)$, in particular the orthogonality of $g_{ij}$ to every self-avoiding path from vertex $i$ to vertex $j$ of length greater than two, the bound on $q_{ij}^{(5)} $ is obtained as
		\[\begin{split}
		\bE \, g_{ij}^4 p_{ij}^2 &= 15\beta^2 N^{-3}  +  3 N^{-2} \sum_{ \substack{   \{i,j\}\in \gamma\in \Gl   }}  (\beta^2)^{|\gamma|-1}     N^{-|\gamma|+1}  \\
		& \leq 15\beta^2 N^{-3} +3 N^{-2}\sum_{l\geq 3} (\beta^2)^{l-1}  (N-2)(N-3)\ldots (N-l+1)     N^{-l+1} \leq C N^{-3}.
		\end{split}\]		

Let us now switch to \eqref{eq:qijbnds2}, fixing $k,k_1,k_2 \geq 3$. For $q_{>k_1,k_2,ij}^{(1)} $, we use that a self-avoiding path $\gamma_1$ from vertex $i$ to vertex $j$ which has length greater than two can be identified uniquely with a cycle that contains $\{i,j\}$ and which has length $ |\gamma_1|+1$ (the loop is given by $\{i,j\}\circ\gamma_1 $). The large graph contributions are then controlled by using the $L^2(\Omega)$ orthogonality of different cycles. With the representation
		\[\begin{split}
		q_{>k_1,k_2,ij}^{(1)} &= - 2 \sum_{ \substack{    \gamma_1  \in \Gl : \\ \{i,j\}\in\gamma_1, |\gamma_1|>k_1   }}  \Big(\prod_{\substack { e\in \gamma_1:e\neq \{i,j\}  }} \beta g_e\Big) \sum_{ \substack{    \gamma_2 \in \Gl : \\ \cV_{\gamma_1}\cap \cV_{\gamma_2}=\{j\} }}   \prod_{\substack { e'\in \gamma_2   }} \beta g_{e'}\\
		&\hspace{0.5cm} - 2 \sum_{ \substack{    \gamma_1  \in \Gl : \\\{i,j\}\in\gamma_1,|\gamma_1|\leq k_1     }}  \Big(\prod_{\substack { e\in \gamma_1:e\neq \{i,j\}  }} \beta g_e\Big) \sum_{ \substack{    \gamma_2 \in \Gl : |\gamma_2|>k_2 , \\ \cV_{\gamma_1}\cap \cV_{\gamma_2}=\{j\} }}   \prod_{\substack { e'\in \gamma_2   }} \beta g_{e'},
		\end{split}\]
we find that
		\[\begin{split}
		\big\| q_{>k_1,k_2,ij}^{(1)} \big\|^2_{L^2(\Omega)}& \leq  C \sum_{ \substack{    \gamma_1  \in \Gl : \\ \{i,j\}\in\gamma_1, |\gamma_1|>k_1   }}(\beta^2)^{|\gamma_1|-1} N^{-|\gamma_1| +1} \sum_{ \substack{    \gamma_2 \in \Gl : \\ \cV_{\gamma_1}\cap \cV_{\gamma_2}=\{j\} }} (\beta^2)^{|\gamma_2|} N^{-|\gamma_2|}\\
		& \hspace{0.5cm} +  C \sum_{ \substack{    \gamma_1  \in \Gl : \\ \{i,j\}\in\gamma_1, |\gamma_1|\leq k_1   }}(\beta^2)^{|\gamma_1|-1} N^{-|\gamma_1| +1} \sum_{ \substack{    \gamma_2 \in \Gl : |\gamma_2| > k_2 \\ \cV_{\gamma_1}\cap \cV_{\gamma_2}=\{j\} }} (\beta^2)^{|\gamma_2|} N^{-|\gamma_2|} \\
		&\leq C \sum_{ \substack{     \gamma_1    \in \Gl :\\\{i,j\}\in\gamma_1, |\gamma_1|>k_1 }}(\beta^2)^{|\gamma_1| -1} N^{-|\gamma_1|+1} \sum_{ \substack{     \gamma_2   \in \Gl: j\in \cV_{\gamma_2} }}(\beta^2)^{|\gamma_2| } N^{-|\gamma_2|} \\
		&\hspace{0.5cm} + C   \sum_{ \substack{     \gamma_1    \in \Gl :\\\{i,j\}\in\gamma_1  }}(\beta^2)^{|\gamma_1| -1} N^{-|\gamma_1|+1}  \sum_{ \substack{     \gamma_2   \in \Gl:|\gamma_2|>k_2,  j\in \cV_{\gamma_2} }}(\beta^2)^{|\gamma_2| } N^{-|\gamma_2|}  \\
		&\leq C \sum_{l > k_1} (\beta^2)^{ l -1} N^{l-2}N^{-l +1}   \sum_{l \geq 3} (\beta^2)^{ l } N^{l-1}N^{-l } \\
		&\hspace{0.5cm} +  C \sum_{l \geq 3} (\beta^2)^{ l -1} N^{l-2}N^{-l +1} \sum_{l > k_2} (\beta^2)^{ l } N^{l-1}N^{-l }\leq CN^{-2} (\beta^2)^{-\min( k_1,k_2) }.
		\end{split}\]  
Analogously, we obtain
		\[\begin{split}
		\big\| q_{>k,ii}^{(2)} \big\|^2_{L^2(\Omega)}  \leq C \sum_{ \substack{     \gamma\in \Gl:\\ | \gamma| > k , i\in \cV_\gamma  }}(\beta^2)^{|\gamma_1| } N^{-|\gamma_1|}&\leq C \sum_{l > k } (\beta^2)^{l } N^{l-1} N^{-l} \leq CN^{-1} e^{-k \log \beta^{-2}}. 
		\end{split}\]  
Consider now the $ L^4(\Omega)$ estimate on $q_{>k,ii}^{(2)}$. Neglecting constant prefactors, we have that
		\be \label{eq:4thmom1}\begin{split}
		\bE \big(q_{\leq k,ii}^{(2)}\big)^4 & \leq \sum_{ \substack{    \gamma_1, \gamma_2,\gamma_3, \gamma_4 \in \Gl:\\ i\in \cV_{\gamma_t}, | \gamma_t|\leq k \,\forall \,t   }} \bE \prod_{\substack { e\in \gamma_1, \dots, \gamma_4   }} \beta  g_e.
		\end{split}\ee
To estimate the right hand side further, we proceed similarly as in \cite[Lemma 3.1]{ALR} and interpret the sum to range over multigraphs $ \gamma = \gamma_1\circ \gamma_2\circ \gamma_3\circ\gamma_4$ with $\gamma_t\in \Gl$, $i\in \cV_{\gamma_t}$ and $|\gamma_t|\leq k$, for each $t\in \{1,2,3,4\}$. By Wick's rule, the expectation of the weight
		\[  \bE \prod_{\substack { e\in \gamma= \gamma_1\circ \gamma_2\circ\gamma_3\circ \gamma_4   }} \beta  g_e\]
does only contribute to the fourth moment if the multiplicity $n_{uv}(\gamma) $ of each edge is even, and since $\gamma$ consists of four cycles, we have in this case in fact $n_{uv}(\gamma)\in \{0,2,4\} $ for each $u,v\in [N]$ with $u<v$. Moreover, we can bound the contribution of such a $\gamma$ trivially by 
		\[ \bE \prod_{\substack { e\in \gamma   }} \beta  g_e = (\beta^2)^{|\gamma|/2}   \prod_{1\leq u<v \leq N}\bE \, g_{uv}^{n_{uv}(\gamma)} \leq  (3\beta^2)^{2k} N^{-|\gamma|/2}. \]
Now, notice that the contribution of each multi-graph $\gamma = \gamma_1\circ \gamma_2\circ \gamma_3\circ\gamma_4$ to the sum on the \rhs in \eqref{eq:4thmom1} depends only on the occupation numbers of the edges $n_{uv}(\gamma) $. In other words, if there exists a permutation $ \pi :[N]\to [N] $ such that $ n_{uv}(\gamma) = n_{\pi(u)\pi(v)}(\gamma')$ for all $u<v$ (i.e., $\gamma$ and $\gamma'$ are isomorphic multi-graphs), then 
		\[\bE \prod_{\substack { e\in \gamma   }} \beta  g_e = \bE \prod_{\substack { e\in \gamma'   }} \beta  g_e. \]
Defining $ \gamma\sim \gamma'$ to be equivalent if and only if there exists a permutation $ \pi :[N]\to [N] $ such that $ n_{uv}(\gamma) = n_{\pi(u)\pi(v)}(\gamma')$ for all $u<v$, and denoting by $[\gamma]$ the corresponding equivalence classes that partition the set 
		\[S_4 :=\big \{\gamma=\gamma_1\circ\gamma_2\circ\gamma_3\circ\gamma_4: i\in \cV_{\gamma_t}, | \gamma_t|\leq k,  \;\forall \,t =1,2,3,4\big \}, \] 
into a disjoint union of equivalence classes, we obtain  
		\[ \bE \big(q_{\leq k,ii}^{(2)}\big)^4\leq (3\beta^2)^{2k} \sum_{ [\gamma] } | [\gamma]| N^{-|[\gamma]|/2} =  (3\beta^2)^{2k}   \sum_{ [\gamma] }  |[\gamma]| e^{-\frac14\sum_{i=1}^N n_i([\gamma])  }.   \]
Now, the number of elements in an equivalence class is bounded by $ N^{ |\cV_{[\gamma] } |-1 }$, because all multi-graphs have vertices in $[N]$ and because each graph contains $i\in[N]$, by assumption. This assumption also implies that $ n_i([\gamma]) = 8$ for each $[\gamma]$. Together with the fact that the degree $n_i([\gamma]) \geq 4$ for each vertex in $[\gamma]$ (because $\gamma$ is the product of cycles and every edge must occur at least twice to give a non-zero contribution under $\bE$), we find
		\[\begin{split}\bE \big(q_{\leq k,ii}^{(2)}\big)^4 \leq    (3\beta^2)^{2k}N^{-1}  \sum_{ [\gamma] }  N^{|\cV_{[\gamma] } |-\frac14\sum_{i=1}^N n_i([\gamma])  } \leq (3\beta^2)^{2k} N^{-2}  \sum_{ [\gamma] }   \leq (3\beta^2)^{2k} C'_k N^{-2}. \end{split}\]
Here, we used that the number of multi-graphs with $l \leq 4k$ edges is bounded by some $C'_k>0$ determined uniquely by $k$ (in particular independent of $N$). This proves \eqref{eq:qijbnds2}.

Finally, let us explain the bounds in \eqref{eq:qijbnds3}. Let us note first the upper bound
\begin{align*}
	\bE\big( (q_{\leq k_N^4,k_N^2,ij}^{(6)})^2\big)&\leq 4 \,\bE  \bigg(\sum_{ \substack{    \gamma_1  \in \Gl : \\\{i,j\}\in\gamma_1, |\gamma_1|\leq k_N^2  }} \!\! \Big(\prod_{\substack { e\in \gamma_1:e\neq \{i,j\}  }} \beta g_e\Big) \sum_{ \substack{    \gamma_2 \in \Gsc :|\gamma_2|\leq k_N^4, \\ |\cV_{\gamma_1} \cap \cV_{\gamma_2}|\geq 2 }} \!\!  \prod_{\substack { e'\in \gamma_2   }} \beta g_{e'},  \bigg)^2,\end{align*}
where we removed the restrictions $j\in \cV_{\gamma_2}$, $\gamma_2 \in \Gl$ and where we relaxed the large graph constraint to $ | \gamma_2|\leq k_N^4$, as all cross terms give positive contributions under taking the expectation, by Wick's rule. Now, it is clear that the \rhs in the last bound can be treated in the same way as $\bE (R^{(2)}_{ij})^2 $ in Corollary \ref{cor:R2} (the differences to $R^{(2)}_{ij}$ are only the graph size constraints and the fact that the edge weights $\tanh(\beta g_e)$ are replaced by $\beta g_e$; both differences do not affect the arguments in the proof of Corollary \ref{cor:R2}). Thus 
		\[\bE\big( \big(q_{\leq k_N^4,k_N^2,ij}^{(6)}\big)^2\big) \leq CN^{-3+\epsilon}\]
for every $\epsilon>0$ small enough and $N$ large enough.

Similarly, the contribution $q_{> k_N^4,k_N^2,ij}^{(6)}$ can be controlled like the term $ R^{(3)}_{ij}$, defined in \eqref{eq:defR2}, as in the proof of Lemma \ref{lm:R134567}. This yields the probability estimate in \eqref{eq:qijbnds3}.
\end{proof}
 
We are now ready to conclude Theorem \ref{thm:main}.
\begin{proof}[Proof of Theorem \ref{thm:main}]
By Corollary \ref{cor:Mmain}, it is enough to show that 
		\[ \| \tbf{P} - (1+\beta^2- \beta \tbf G)^{-1}\|_{\text{op}}\to 0\]
as $N\to\infty$, in the sense of probability. Using $ \|\beta \tbf G\|_{\text {op}} < 2\beta + (1-\beta^2) \epsilon/2$, we first bound
		\[\begin{split}
		\| \tbf P - (1+\beta^2- \beta \tbf G)^{-1}\|_{\text{op}} &\leq  \|  (1+\beta^2- \beta \tbf G)^{-1}\|_{\text{op}} \| \tbf Q \|_{\text{op}}\leq (1-\beta)^{-2}(1+ \epsilon) \| \tbf Q \|_{\text{op}} 
		\end{split}\]
so that, by \eqref{eq:Qdec2} and $ \|\tbf A \tbf B\|_{\text{op}}\leq \|\tbf A\|_{\text{op}}\| \tbf B\|_{\text{op}}$ for all $ \tbf A, \tbf B\in\bR^{N\times N}$, we obtain that
		\be\label{eq:pfmbnd1}\begin{split}
		&\frac{ \| \tbf P - (1+\beta^2- \beta \tbf G)^{-1}\|_{\text{op}}}{ (1-\beta)^{-2}(1+ \epsilon)} \\
		& \leq   \big(  \| \tbf P - (1+\beta^2- \beta \tbf G)^{-1}\|_{\text{op}}  +C \big) \big( \beta^2 N^{-1}   + \| \tbf{Q}^{(2)}\|_{\text{op}} + \|\tbf{Q}_{> k_2}^{(2)}\|_{\text{op}}+ \| \tbf{Q}^{(4)}\|_{\text{op}}\big)  \\
		&\hspace{0.5cm}  + \|\tbf{Q}^{(3)} \|_{\text{op}}+\| \tbf{Q}^{(5)}\|_{\text{op}}  +  \|\tbf{Q}_{>k_1,k_2}^{(1)}\|_{\text{op}} + \|\tbf{Q}_{> k_2}^{(2)}\|_{\text{op}} + \|\tbf{Q}_{>k_1,k_2}^{(6)}\|_{\text{op}} +\|\tbf{Q}_{\leq k_1,k_2}^{(6)}\|_{\text{op}} \\
		\end{split}\ee
for $C=(1-\beta)^{-2}(1+ \epsilon) $ and for fixed $k_1,k_2\geq 3$. Choosing $k_1=k_N^4, k_2=k_N^2 $, defined in \eqref{eq:defkN}, \st $k_N\to \infty$ as $N\to\infty$, we make sure that the $k_1$- and $k_2$-dependent matrices on the \rhs converge to zero, in the sense of probability. Indeed, by \eqref{eq:qijbnds3}, we get 
		\[\begin{split}
		 \bP\Big( \|\tbf{Q}_{>k_N^4,k_N^2}^{(6)}\|_{\text{op}}> \delta  \Big) \leq \bP\Big( \|\tbf{Q}_{>k_N^4,k_N^2}^{(6)}\|_{\text{F}}> \delta  \Big) &\leq \bP\Big( \max_{i,j \in [N]:i\neq j}| q_{>k_N^4,k_N^2,ij}^{(6)}| > \delta N^{-1}  \Big) \to 0
		 \end{split}\]
as $N\to \infty$, for every $\delta >0$. Similarly, we obtain from \eqref{eq:qijbnds2} and \eqref{eq:qijbnds3} that
		\[\begin{split}
		&\bP\Big(  \|\tbf{Q}_{>k_N^4,k_N^2}^{(1)}\|_{\text{op}} > \delta  \Big) +\bP\Big(  \|\tbf{Q}_{>k_N^2}^{(2)}\|_{\text{op}} > \delta  \Big) +\bP\Big(  \|\tbf{Q}_{\leq k_N^4,k_N^2}^{(6)}\|_{\text{op}} > \delta  \Big)\\
		& \leq \bP\Big(  \|\tbf{Q}_{>k_N^4,k_N^2}^{(1)}\|_{\text{F}} > \delta  \Big) +\bP\Big(  \|\tbf{Q}_{>k_N^2}^{(2)}\|_{\text{F}} > \delta  \Big) +\bP\Big(  \|\tbf{Q}_{\leq k_N^4,k_N^2}^{(6)}\|_{\text{F}} > \delta  \Big)\\
		&\leq  \delta^{-2} N^2  \max_{i,j\in[N]: i\neq j}\Big(  \| q_{>k_N^4,k_N^2, ij}^{(1)}\|^2_{L^2(\Omega)}+\| q_{\leq k_N^4,k_N^2, ij}^{(6)}\|^2_{L^2(\Omega)}\Big) + \delta^{-2} N  \max_{i\in[N] } \| q_{>k_N^2, ii}^{(2)}\|^2_{L^2(\Omega)} \\
		&\leq \delta^{-2} N^{-1+\epsilon}
		\end{split}\]
for fixed $\epsilon >0$ small enough, so that 
		\[ \lim_{N\to\infty} \|\tbf{Q}_{>k_N^4,k_N^2}^{(1)}\|_{\text{op}} =  \lim_{N\to\infty} \|\tbf{Q}_{>k_N^2}^{(2)}\|_{\text{op}}=\lim_{N\to\infty} \|\tbf{Q}_{\leq k_N^4,k_N^2}^{(6)}\|_{\text{op}}=\lim_{N\to\infty} \|\tbf{Q}_{>k_N^4,k_N^2}^{(6)}\|_{\text{op}}=0. \]	
Inserting this into \eqref{eq:pfmbnd1}, we conclude that
		\be\label{eq:pfmbnd2}\begin{split}
		&\frac{ \| \tbf P - (1+\beta^2- \beta \tbf G)^{-1}\|_{\text{op}}}{ (1-\beta)^{-2}(1+ \epsilon)} \\
		& \leq   \big(  \| \tbf P - (1+\beta^2- \beta \tbf G)^{-1}\|_{\text{op}}  +C \big) \big(   \| \tbf{Q}^{(2)}\|_{\text{op}}  + \| \tbf{Q}^{(4)}\|_{\text{op}} +o_1(1)\big)  \\
		&\hspace{0.5cm}  + \|\tbf{Q}^{(3)} \|_{\text{op}}+\| \tbf{Q}^{(5)}\|_{\text{op}}  +  o_2(1) \\
		\end{split}\ee
for two errors $o_1(1)$ and $o_2(1)$ that converge $\lim_{N\to\infty} o_1(1)=\lim_{N\to\infty} o_2(1)=0 $ in probability. Now, consider the remaining matrices $\tbf{Q}^{(2)}$, $\tbf{Q}^{(3)}$, $\tbf{Q}^{(4)}$ and $\tbf{Q}^{(5)}$. Using that $\| \tbf Q^{(3)} \|_{\text{op}}\leq \| \tbf Q^{(3)} \|_{\text{F}}$, we obtain with the bounds \eqref{eq:qijbnds1} for every $\delta >0$ that
		\[\begin{split} 
		 \bP\big(  \| \tbf Q^{(3)} \|_{\text{op}}  > \delta  \big) & \leq \bP\big(  \| \tbf Q^{(3)} \|_{\text{F}}  > \delta  \big)\leq \delta^{-2} N^2 \max_{i,j\in [N]: i\neq j} \| q_{ij}^{(3)}\|_{L^2(\Omega)}^2 \leq C \delta^{-2}  N^{-1}
		\end{split}\] 
so that $ \lim_{N\to\infty} \| \tbf Q^{(3)} \|_{\text{op}}=0$. The same argument implies $\lim_{N\to\infty} \| \tbf Q^{(5)} \|_{\text{op}}  =0$ and for $\tbf{Q}^{(4)}$, we notice that, by a standard concentration and union bound, we have 
		\[   \bP\Big(   \max_{i\in [N]}  \Big|  \sum_{u}\big(g_{iu}^2- \bE\,g_{iu}^2 \big) \Big| >  \delta      \Big) \leq  C e^{-c N\delta^{2} } \]
for suitable $ C, c >0$ that are independent of $N\in \bN$ and $\delta >0$. Recalling the definition of $q_{ii}^{(4)}$ in \eqref{eq:defQ1to5}, this trivially implies that in the sense of probability we have
		\[ \lim_{N\to\infty} \| \tbf{Q}^{(4)} \|_{\text{op}} = \beta^2\lim_{N\to\infty}\max_{i\in[N]}  \Big|  \sum_{u}\big(g_{iu}^2- \bE\,g_{iu}^2 \big) \Big| = 0. \]
Finally, consider $ \tbf{Q}^{(2)}$. For $\delta>0$ and fixed $ k  \geq 3$, we use \eqref{eq:qijbnds2} to estimate
		\[\begin{split}
		\bP\big(  \| \tbf Q^{(2)} \|_{\text{op}}  > \delta  \big) & \leq \bP\big(  \| \tbf Q_{> k}^{(2)} \|_{\text{op}}  > \delta/2  \big) + \bP\big(  \| \tbf Q_{\leq k}^{(2)} \|_{\text{op}}  > \delta/2  \big)\\
		&\leq \delta^{-2} N \max_{i\in [N]}\| q_{>k,ii}^{(2)}\|_{L^2(\Omega)}^2 + 2^4\delta^{-4} N \max_{i\in [N]}\| q_{>k,ii}^{(2)}\|_{L^4(\Omega)}^4\\
		&\leq  C \delta^{-2} e^{-k \log \beta^{-2}} + C'_k \delta^{-4} N^{-1}.
		\end{split}\]
 Sending first $N\to\infty$ and then $k\to\infty$, we conclude that $\lim_{N\to\infty}  \| \tbf Q^{(2)} \|_{\text{op}}=0$. 
 
Collecting the bounds from above and inserting them into \eqref{eq:pfmbnd2}, we conclude that 
		\[\begin{split}
		 \| \tbf P - (1+\beta^2- \beta \tbf G)^{-1}\|_{\text{op}}   &\leq \frac{o_1'(1)}{(1-\beta)^{2}(1+ \epsilon)^{-1}- o_2'(1)}
		\end{split}\]
for two errors $o_1'(1)$ and $ o_2'(1)$ that satisfy $\lim_{N\to\infty} o_1'(1) = \lim_{N\to\infty} o_2'(1) =0$. In particular, we conclude that 
$\lim_{N\to\infty } \| \tbf P - (1+\beta^2- \beta \tbf G)^{-1}\|_{\text{op}} =0 $ in probability. 
\end{proof}

\vspace{0.25cm}
\noindent\textbf{Acknowledgements.} We thank two anonymous referees for providing several suggestions to improve our manuscript and to simplify some of our arguments. C. B. thanks G. Genovese for helpful discussions and acknowledges support by the Deutsche Forschungsgemeinschaft (DFG, German Research Foundation) under Germany’s Excellence Strategy – GZ 2047/1, Project-ID 390685813. A. S. acknowledges support by the Deutsche Forschungsgemeinschaft (DFG, German Research Foundation) under Germany’s Excellence Strategy – GZ 2047/1, Project-ID 390685813 and Project-ID 211504053 - SFB1060. The work of C. X. is partially funded by a Simons Investigator award. The work of H.-T. Y. is partially supported by the NSF grant DMS-1855509 and DMS-2153335 and a Simons Investigator award.


\appendix
\section{Proof of the Identity \eqref{eq:heur2}}\label{appx}

In this appendix, we outline the key steps of the computation for
		\be \label{eq:app1} \begin{split}
		\bE   \| \tbf P( 1+\beta^2 - \beta \tbf G) - \text{id}_{\bR^N}  \|_{\text F}^2 = O(1).
		\end{split} \ee
Denoting by $ \langle \tbf A ,\tbf B \rangle_{\text F} = \text{tr}\, \tbf A^T \tbf B =  \sum_{i,j=1}^N a_{ji} b_{ji} $ the standard Hilbert-Schmidt inner product for matrices $\tbf A, \tbf B \in \bR^{N\times N}$, we have that 
		\be\label{eq:form1}\begin{split}
		\bE   \| \tbf P( 1+\beta^2 - \beta \tbf G) - \text{id}_{\bR^N}  \|_{\text F}^2 
		& =  ( 1+\beta^2)^2\, \bE\,   \| \tbf P  \|_{\text F}^2 + \beta^2 \bE   \| \tbf P \tbf G  \|_{\text F}^2 + N - 2  ( 1+\beta^2)\, \bE \,\text{tr}\, \tbf P  \\
		& \hspace{0.5cm} - 2 ( 1+\beta^2)\beta \, \bE\,\langle\tbf P , \tbf P \tbf G \rangle_{\text F}+2 \beta \, \bE \,\text{tr}\, \tbf P \tbf G. 
		\end{split}\ee
We compute each expectation on the \rhs of the previous equation separately. 

First, using that $p_{ii} =1$ and $\bE \, p_{ij}=0$ for all $i\neq j\in [N]$, we trivially have $\bE \,\text{tr}\, \tbf P  =  N. $
Using that the weights of distinct paths are orthogonal in $L^2(\Omega)$, we find next that
		\[ \bE \,\text{tr}\, \tbf P \tbf G= \sum_{i,j: i\neq j} \bE \,p_{ij} g_{ij} = \beta (N-1).  \]
Similarly, we obtain that
		\[\begin{split}
		\bE\,   \| \tbf P  \|_{\text F}^2 = N + \sum_{i,j: i\neq j} \sum_{\gamma\in \Gamma_p^{ij}} \bE \bigg(\prod_{e\in\gamma} \beta g_e\bigg)^2& = N+  \sum_{i,j: i\neq j} \sum_{\gamma\in \Gamma_p^{ij}} \frac{\beta^{2|\gamma|} }{N^{|\gamma|}} = \frac{N}{(1-\beta^2) } + O(1).
		\end{split}\]
 Note that, in the last step, we used that
 		$  | \{ \gamma\in  \Gamma_p^{ij}: |\gamma| = k  \}|  = \frac{(N-2)!}{(N-k-1)!} $
 for every $k\geq 1$. Computing the remaining two expectations that appear on the \rhs of \eqref{eq:form1} is slightly more involved. Using Gaussian integration by parts and the fact that $\bE g_{uv}^{2l+1} =0$ for all $\l\in\bN_0$, we find that
		\[\begin{split}
		\beta\, \bE\,\langle\tbf P , \tbf P \tbf G \rangle_{\text F} & = \frac{\beta^2}N\sum_{i,j,k:i\neq k} \,  \bE \bigg(p_{ij} \sum_{\substack{ \gamma\in \Gp^{jk}:\\ \{i,k\}\in \gamma }} \, \prod_{ \substack{ e \in \gamma, e\not= \{i,k\}  } } \!\! \beta  g_{e}+ \,p_{jk} \sum_{\substack{ \gamma\in \Gp^{ij}:\\ \{i,k\}\in \gamma }} \, \prod_{ \substack{ e \in \gamma, e\not= \{i,k\}  } } \!\! \beta  g_{e} \bigg)\\
		& = 2\beta^2 (N-1) + \frac{2\beta^2}N\sum_{i,j,k: i\neq j, j\neq k,k\neq i}  \bE \sum_{\substack{ \gamma\in \Gp^{ij}:\\ \{i,k\}\in \gamma }} \, \sum_{\substack{ \gamma'\in \Gp^{jk}}}  \prod_{ \substack{ e \in \gamma\setminus \{i,k\},  e'\in\gamma' } } \!\! \beta^2  g_{e}g_{e'}\\
		& = 2\beta^2 N +  2 \sum_{\substack{ i,j,k: i\neq j,\\  j\neq k,k\neq i }} \sum_{\substack{ \gamma\in \Gp^{ij}: \\ \{i,k\}\in \gamma }}  \frac{\beta^{2|\gamma|}}{N^{|\gamma|}}+ O(1)\\
		& = 2\beta^2 N + \frac{2\beta^4}{(1-\beta^2)}N+ O(1). 
		\end{split}\]					
Similarly, we find that $ \beta^2\, \bE   \| \tbf P \tbf G  \|_{\text F}^2 = \sum_{u=1}^5\Delta_u$, where
		\[\begin{split} 
		\Delta_1 & = \beta^2 \big(1-N^{-1}\big) \bE \|\tbf P\|_{\text{F}}^2 = \frac{\beta^2}{(1-\beta^2)}N + O(1), \\
		\Delta_2 & = \frac{\beta^4}{N^2} \sum_{\substack{ i,j,k,l:\\  i\neq j, i\neq l,\\ j\neq k, k\neq l   }}  \sum_{ \substack{ \gamma\in \Gp^{jk}: \\ \{i,j\}\in\gamma }} \sum_{ \substack{ \gamma'\in \Gp^{kl}: \\\{i,l\}\in\gamma' }}  \bE \prod_{\substack{ e\in \gamma \setminus \{i,j\}, \\ e'\in \gamma'\setminus \{i,l \} } } \beta^2 g_e g_{e'}, \\
		\Delta_3 & = \frac{\beta^4}{N^2} \sum_{\substack{ i,j,k,l:\\  i\neq j, i\neq l,\\ j\neq k, k\neq l   }}  \sum_{ \substack{ \gamma\in \Gp^{kl}: \\ \{i,j\}\in\gamma }} \sum_{ \substack{ \gamma'\in \Gp^{jk}: \\\{i,l\}\in\gamma' }}  \bE \prod_{\substack{ e\in \gamma \setminus \{i,j\}, \\ e'\in \gamma'\setminus \{i,l \} } } \beta^2 g_e g_{e'}, \\
		\Delta_4 & = \frac{2\beta^4}{N^2} \sum_{\substack{ i,j,k: \\ i\neq j, j\neq k,\\i\neq k    }}  \sum_{ \substack{ \gamma\in \Gp^{jk}: \\ \{i,j\}, \{i,k\}\in\gamma }}   \bE \prod_{\substack{ e\in \gamma \setminus \{i,k\}\circ \{i,j\}} } \beta g_{e} =  2\beta^4 N  + O(1) , \\
		\Delta_5 & =\frac{2\beta^4}{N^2} \sum_{\substack{ i,j,k,l:\\  i\neq j, i\neq l,i\neq k\\ j\neq k, k\neq l, l\neq j   }}  \sum_{ \substack{ \gamma\in \Gp^{kl} }} \sum_{ \substack{ \gamma'\in \Gp^{jk}: \\ \{i,j\}, \{i,l\}\in\gamma' }}  \bE \prod_{\substack{ e\in \gamma  , \\ e'\in \gamma'\setminus \{i,j\}\circ\{i,l \} } } \beta^2 g_e g_{e'}.\\
		\end{split}\]
Arguing as in the previous steps, we obtain that
		\[\begin{split}
		\Delta_2 & = \beta^4 N  + \frac{\beta^4}{N^2} \sum_{\substack{ i,j,k:\\  i\neq j, j\neq k,\\ i\neq k   }} \sum_{\gamma\in\Gp^{ik}: j\not \in \cV_{\gamma}} \frac{\beta^{2|\gamma|}}{N^{|\gamma|}} + \frac{\beta^4}{N^2} \sum_{\substack{ i,j,k,l:\\  i\neq j, i\neq l, j\neq l \\ j\neq k, k\neq l, i\neq k   }} \sum_{\substack{ \gamma\in\Gp^{ik}:\\ j, l \not \in \cV_{\gamma}}} \frac{\beta^{2|\gamma|}}{N^{|\gamma|}}+ O(1)\\
		& = \beta^4 N+ \frac{\beta^6}{(1-\beta^2)} N+ O(1)
		\end{split}\]
and that
		\[\begin{split} 
		\Delta_3 & =   \frac{\beta^4}{N^2} \sum_{\substack{ i,j ,l:\\  i\neq j, i\neq l, j\neq l    }}  \sum_{ \substack{ \gamma\in \Gp^{il}: \\ \{i,j\}\in\gamma }} \sum_{ \substack{ \gamma'\in \Gp^{ij}: \\\{i,l\}\in\gamma' }}  \bE \prod_{\substack{ e\in \gamma \setminus \{i,j\}, \\ e'\in \gamma'\setminus \{i,l \} } } \beta^2 g_e g_{e'}  \\
		&\hspace{0.5cm} +  \frac{\beta^4}{N^2} \sum_{\substack{ i,j,k:\\  i\neq j, i\neq k, j\neq k    }}  \sum_{ \substack{ \gamma\in \Gp^{jk}: \\ \{i,j\}\in\gamma }} \sum_{ \substack{ \gamma'\in \Gp^{jk}: \\\{i,j\}\in\gamma' }}  \bE \prod_{\substack{ e\in \gamma \setminus \{i,j\}, \\ e'\in \gamma'\setminus \{i,j \} } } \beta^2 g_e g_{e'} \\
		&\hspace{0.5cm} +  \frac{\beta^4}{N^2} \sum_{\substack{ i,j,k,l:\\  i\neq j, i\neq k, i\neq l,\\ j\neq k, j\neq l, k\neq l   }}  \sum_{ \substack{ \gamma\in \Gp^{kl}: \\ \{i,j\}\in\gamma }} \sum_{ \substack{ \gamma'\in \Gp^{jk}: \\\{i,l\}\in\gamma' }}  \bE \prod_{\substack{ e\in \gamma \setminus \{i,j\}, \\ e'\in \gamma'\setminus \{i,l \} } } \beta^2 g_e g_{e'} + O(1)\\
		& =    \frac{\beta^4}{N^2} \sum_{\substack{ i,j,k,l:\\  i\neq j, i\neq k, i\neq l,\\ j\neq k, j\neq l, k\neq l   }}  \sum_{ \substack{ \gamma\in \Gp^{kl}: \\ \{i,j\}\in\gamma }} \sum_{ \substack{ \gamma'\in \Gp^{jk}: \\\{i,l\}\in\gamma' }}  \bE \prod_{\substack{ e\in \gamma \setminus \{i,j\}, \\ e'\in \gamma'\setminus \{i,l \} } } \beta^2 g_e g_{e'} + O(1). 
		\end{split}\] 
For the third term on the \rhs of the previous equation, we notice that 
		\[\begin{split} 
		&\frac{\beta^4}{N^2} \sum_{\substack{ i,j,k,l:\\  i\neq j, i\neq k, i\neq l,\\ j\neq k, j\neq l, k\neq l   }}  \sum_{ \substack{ \gamma\in \Gp^{kl}: \\ \{i,j\}\in\gamma }} \sum_{ \substack{ \gamma'\in \Gp^{jk}: \\\{i,l\}\in\gamma' }}  \bE \prod_{\substack{ e\in \gamma \setminus \{i,j\}, \\ e'\in \gamma'\setminus \{i,l \} } } \beta^2 g_e g_{e'} \\
		& =\frac{\beta^4}{N^2} \sum_{\substack{ i,j,k,l:\\  i\neq j, i\neq k, i\neq l,\\ j\neq k, j\neq l, k\neq l   }}  \bE \bigg(  \sum_{ \substack{ \gamma_1\in \Gp^{ik}, \gamma_2 \in \Gp^{jl}: \\ \cV_{\gamma_1}\cap \cV_{\gamma_2}=\emptyset }} \prod_{\substack{ e_1\in \gamma_1,  e_2\in \gamma_2
		 } } \beta^2 g_{e_1} g_{e_2}  +  \sum_{ \substack{ \gamma_1\in \Gp^{jk}, \gamma_2 \in \Gp^{il}: \\ \cV_{\gamma_1}\cap \cV_{\gamma_2}=\emptyset }} \prod_{\substack{ e_1\in \gamma_1,  e_2\in \gamma_2
		 } } \beta^2 g_{e_1} g_{e_2}\bigg) \\
		&\hspace{2.8cm}\times\bigg(  \sum_{ \substack{ \gamma'_1\in \Gp^{ij}, \gamma'_2 \in \Gp^{kl}: \\\cV_{\gamma'_1}\cap \cV_{\gamma'_2}=\emptyset }}\prod_{\substack{ e'_1\in \gamma'_1,  e'_2\in \gamma'_2
		 } } \beta^2 g_{e'_1} g_{e'_2} +  \sum_{ \substack{ \gamma'_1\in \Gp^{jl}, \gamma'_2 \in \Gp^{ik}: \\ \cV_{\gamma'_1}\cap \cV_{\gamma'_2}=\emptyset }} \prod_{\substack{ e'_1\in \gamma'_1,  e'_2\in \gamma'_2
		 } } \beta^2 g_{e'_1} g_{e'_2} \bigg) \\
		 & = \frac{\beta^4}{N^2}   \sum_{\substack{ i,j,k,l:\\  i\neq j, i\neq k, i\neq l,\\ j\neq k, j\neq l, k\neq l   }}  \sum_{ \substack{ \gamma\in \Gp^{ik}: j,l\not \in \cV_\gamma }} \frac{ \beta^{2|\gamma|} }{N^{|\gamma|}}   \sum_{ \substack{ \gamma'\in \Gp^{jl}: \cV_{\gamma}\cap \cV_{\gamma'}=\emptyset  }} \frac{ \beta^{2|\gamma'|} }{N^{|\gamma'|}} \\
		 & =O(1),
		 \end{split}  \] 	
so that altogether $ \Delta_3 = O(1)$. Finally, similar computations as before yield that
		\[\begin{split}
		\Delta_5 & = \frac{2\beta^4}{N^2} \sum_{\substack{ i,j,k,l:\\  i\neq j, i\neq l,i\neq k\\ j\neq k, k\neq l, l\neq j   }} \sum_{\gamma\in\Gp^{kl}: i,j\not \in \cV_{\gamma}} \frac{\beta^{2|\gamma|}} {N^{|\gamma|}} = \frac{2\beta^6}{(1-\beta^2)}N  + O(1). 
		\end{split}\]
		
Collecting the above identities and plugging them into \eqref{eq:form1}, we conclude that
		\[  \bE   \| \tbf P( 1+\beta^2 - \beta \tbf G) - \text{id}_{\bR^N}  \|_{\text F}^2 = c_\beta N + O(1) \] 
for a constant $c_\beta$ that in fact turns out to vanish: we find that
		\[\begin{split}
		c_\beta & = (1-\beta^2)^{-1}(1+\beta^2)^2   +  (1-\beta^2)^{-1}\beta^2  + 3\beta^4 +  3(1-\beta^2)^{-1}\beta^6 + 1   \\
		& \hspace{0.5cm} -2(1+\beta^2) -2(1+\beta^2) \big(2\beta^2+2(1-\beta^2)^{-1}\beta^4 \big) + 2\beta^2\\
		& = (1-\beta^2)^{-1} \big ( (1+1-2) + ( 2\beta^2+\beta^2-\beta^2-4\beta^2+2\beta^2 ) \\
		&\hspace{2.5cm}  + (\beta^4+3\beta^4+2\beta^4-4\beta^4-2\beta^4)   \big)  \\
		&= 0.
		\end{split} \]   

It is clear from the above identities that with some additional computional effort one can deduce the $O(1)$ contribution to $  \bE   \| \tbf P( 1+\beta^2 - \beta \tbf G) - \text{id}_{\bR^N}  \|_{\text F}^2$ (up to errors of size $O(N^{-1})$) . Since such an asymptotics has already been derived explicitly in \cite{AG} (for the quantity $\bE   \| \tbf M( 1+\beta^2 - \beta \tbf G) - \text{id}_{\bR^N}  \|_{\text F}^2$) and since this is of no relevance for the proof of Theorem \ref{thm:main}, we omit the details.



\end{document}